\documentclass[11pt]{article}

\usepackage[margin=1in]{geometry}
\usepackage{amsmath,amssymb,amsthm,mathtools,bm,bbm}
\usepackage{enumitem}
\usepackage{booktabs}
\usepackage{array}
\usepackage{microtype}
\usepackage[round,authoryear]{natbib}
\usepackage{setspace}
\usepackage{graphicx}
\usepackage[colorlinks=true,citecolor=blue,linkcolor=blue,urlcolor=blue]{hyperref}

\hypersetup{colorlinks=true,linkcolor=black,citecolor=black,urlcolor=blue}

\newcommand{\E}{\mathbb{E}}
\newcommand{\Prob}{\mathbb{P}}
\newcommand{\R}{\mathbb{R}}

\newcommand{\argmin}{\mathop{\mathrm{arg\,min}}}

\newcommand{\supp}{\mathop{\mathrm{supp}}}
\newcommand{\diag}{\mathop{\mathrm{diag}}}
\newcommand{\Var}{\mathop{\mathrm{Var}}}

\newcommand{\norm}[1]{\left\lVert #1 \right\rVert}
\newcommand{\abs}[1]{\left\lvert #1 \right\rvert}
\newcommand{\G}{\mathcal{G}}
\newcommand{\bfPsi}{\boldsymbol{\Psi}}
\newcommand{\X}{\mathcal{X}}
\newcommand{\Z}{\mathcal{Z}}
\newcommand{\Pcal}{\mathcal{P}}
\newcommand{\Wcal}{\mathcal{W}}
\newcommand{\Acal}{\mathcal{A}}
\newcommand{\ThetaSpace}{\Theta}
\newcommand{\Mcal}{\mathcal{M}}

\newcommand{\Jcal}{\mathcal{J}}

\numberwithin{equation}{section}

\theoremstyle{plain}
\newtheorem{assumption}{Assumption}[section]
\newtheorem{lemma}[assumption]{Lemma}
\newtheorem{proposition}[assumption]{Proposition}
\newtheorem{theorem}[assumption]{Theorem}

\theoremstyle{definition}

\theoremstyle{remark}

\title{Uniformly Consistent Semi-nonparametric Demand Estimation with Micro-Data\thanks{I am grateful to Denis Chetverikov for his guidance, detailed comments, and many helpful discussions. I also thank Andres Santos, Jinyong Hahn, Zhipeng Liao, Rosa Matzkin, Martin Hackmann, and John Asker for valuable comments and discussions.}}

\author{Richard Grigorian\thanks{Department of Economics, UCLA. Email: rgrigorian007@ucla.edu.}}
\date{July 2026}

\begin{document}
\maketitle

\begin{doublespace}
\begin{abstract}
This paper develops a profiled sieve minimum-distance estimator for a semi-nonparametric differentiated-products demand model with micro-level choice data. Building on \cite{berry_nonparametric_2024}, the estimator uses within-market variation in consumer covariates to recover a flexible consumer-heterogeneity function and market-specific composite intercepts. Excluded price instruments then separate these intercepts into a flexible price-side function and structural demand shocks. The main statistical challenge is that the number of profiled market intercepts grows with the number of markets. I show that, when both the number of markets and the minimum within-market sample size grow, this profiling step is asymptotically negligible and the common structural functions are uniformly consistently estimated. Monte Carlo evidence supports the consistency result and illustrates the value of flexible price-side estimation for counterfactual demand analysis.

\end{abstract}
\end{doublespace}



\clearpage

% ---------------------------------------------------------------------
\section{Introduction}\label{sec:Intro}
% ---------------------------------------------------------------------

Micro-level choice data provide a strong source of identifying variation for differentiated-product demand systems. With market-level data, researchers observe aggregate shares, prices, product characteristics, and instruments. Identification then generally relies on cross-market variation in observed market outcomes -- such as prices and market shares -- and on instruments that shift endogenous prices without shifting demand. With micro data, we can observe different consumers within markets that face the same prices, product attributes, and market-level demand shocks, but differ in their observed characteristics, thereby granting us another source of variation. \cite{berry_nonparametric_2024} show that this within-market variation can identify nonparametric demand systems using only price instruments. Their result is an identification theorem. This paper studies the subsequent estimation problem.

The object of interest is a semi-nonparametric demand system of the form:
\[ \E[Q_{it}\mid C_{it}] = \sigma\Big(g_0(Z_{it},X_t)+\psi_0(P_t,X_t)+h_{0t}\Big),\quad  \E[h_{0t} \mid W_t, X_t] = 0 \]
where $Q_{it} \in \{0,1\}^J$ is the vector of inside-good choice indicators for consumer $i$ in market $t$, $Z_{it}$ denotes consumer covariates, $P_t$ denotes the vector of prices, $W_t$ denotes excluded price instruments, $X_t$ denotes observed product and market characteristics, and $h_{0t} \in \R^J$ is a vector of product-level demand shocks common to all consumers in market $t$. Here $C_{it}$ denotes the information observed by consumer $i$ in market $t$. The link function $\sigma$ is multinomial logit, but the functions comprising the index, $g_0$ and $\psi_0$, are unknown and flexible. Thus, $g_0$ captures observed consumer heterogeneity, while $\psi_0$ captures the systematic role of prices and observed product or market characteristics. The logit link is used for tractability; the model remains semi-nonparametric because the index functions are not parametrically specified.

The main difficulty in estimation is that the within-market data do not separately identify $\psi_0(P_t,X_t)$ and $h_{0t}$. Holding market $t$ fixed, both terms are constant across consumers and enter individual choice probabilities additively. Therefore, I define the composite market intercept:
\[ a_{0t} := \psi_0(P_t,X_t)  + h_{0t}\]
similar to \cite{hahn_time-invariant_2005}. This transformation will clarify the source of identification, and it will also help handle the incidental parameters problem often generated when $T \to \infty$. Moreover, the micro data identify how choices vary with $Z_{it}$ within a market and therefore provide information about $g_0$ and $a_{0t}$. They do not, by themselves, decompose $a_{0t}$ into the price-side function $\psi_0(P_t,X_t)$ and the demand shock, $h_{0t}$. This decomposition requires excluded price instruments, $W_t$. Therefore, the estimator must combine two sources of variation: within-market variation in $Z_{it}$, which identifies the consumer-heterogeneity component and the composite intercept, and cross-market variation in $(P_t,W_t,X_t)$, which separates the composite intercept into $\psi_0(P_t,X_t)$ and $h_{0t}$.

I propose a profiled sieve minimum-distance estimator that implements this logic. The estimator is a one-step estimator; however, it can be understood in three phases. First, for each candidate heterogeneity function $g$, I profile the market-specific composite intercepts $a_t$ by matching observed and predicted within-market shares. Second, I estimate $g_0$ using micro moments formed from the profiled individual choice residuals. Third, I estimate $\psi_0$ using the price-side conditional moment:
\[ \E[a_{0t}-\psi_0(P_t,X_t)\mid W_t,X_t]=0.\]
This moment will follow from the definition of $a_{0t}$, the normalization $\E[h_{0t} \mid X_t] = 0$, and the excluded-instrument restriction $\E[h_{0t} \mid W_t,X_t] = \E[h_{0t} \mid X_t]$ -- the latter two of which come directly from \cite{berry_nonparametric_2024}. I implement the above moment condition by projecting the implied shock, $a_t-\psi(P_t,X_t)$, onto a growing basis of functions with arguments $(W_t,X_t)$. This is a sieve minimum-distance (SMD), or conditional-moment, step in the sense of \cite{ai_efficient_2003}. The price-side moment is also closely related to the nonparametric IV (NPIV) logic in \cite{newey_instrumental_2003}: completeness of the conditional distribution of prices given instruments is what converts the conditional moment into identification of the unknown function $\psi_0$.

The main result of this paper establishes uniform consistency of the profiled SMD estimator.  The asymptotic framework has both the number of markets and the within-market sample sizes growing: $T \to \infty$ and $\min_{t\le T} n_t \to \infty$\footnote{This differs from large-product asymptotic frameworks for differentiated-products demand, where the number of products grows; see, for example, \cite{freyberger_asymptotic_2015} and \cite{lu_semi-nonparametric_2023}.}. This sampling scheme is relevant because it would ordinarily imply the aforementioned incidental parameters problem: the number of market-specific intercepts grows to infinity with the number of markets. However, each intercept is to be estimated from a growing number of consumers. In conjunction with my profiling approach, the market intercepts will not cause a problem for consistency as they will not be estimated with fixed information. Under regularity, identification, sieve approximation and uniform convergence conditions,
\[ \norm{\hat g - g_0}_\infty + \norm{\hat\psi - \psi_0}_\infty = o_p(1).\]
The logic comes in three steps: first, the inner share-matching problem consistently profiles $a_t$ uniformly over candidate functions $g$. Second, the feasible profiled criterion converges uniformly to its population counterpart. And lastly, the population criterion is uniquely minimized at the true parameter values, subject to identification normalizations.

My work contributes to several literatures. First, it contributes to the econometric and industrial organization literature on differentiated-products demand. \cite{berry_estimating_1994} and \cite{berry_automobile_1995} provide the canonical framework for estimating demand with endogenous prices and market-level data. Subsequent work developed richer sources of variation, micro moments, and improved computational methods; examples include \cite{berry_differentiated_2004}, \cite{nevo_measuring_2001}, \cite{petrin_quantifying_2002}, \cite{dube_improving_2012}, \cite{reynaert_improving_2014}, and \cite{gandhi_measuring_2025}. Rather than specifying a parametric random-coefficients utility model and deriving an estimating equation through a market-share inversion, I begin from the Berry-Haile micro-data identification framework and impose a tractable semi-nonparametric link for estimation.

Second, my work is closely related to Berry and Haile's work on nonparametric identification of differentiated-products demand. \cite{berry_connected_2013} study connected substitutes and demand inversion, while \cite{berry_identification_2014} develop nonparametric identification results using market-level data and instruments. Most relevant to my work, \cite{berry_nonparametric_2024} show that micro data can identify demand using only price instruments. I use this identification result as the foundation for my estimator. In this sense, my paper is an estimation counterpart to the Berry-Haile micro-data identification result.

Third, the paper contributes to the literature on flexible and nonparametric demand estimation. \cite{compiani_market_2022}, \cite{wang_sieve_2023}, and \cite{tebaldi_nonparametric_2023} study flexible demand estimation in differentiated-products settings. \cite{monardo_measuring_2025} develops a scalable method that exploits restrictions such as permutation invariance to reduce the dimensionality of demand estimation. \cite{lu_semi-nonparametric_2023} studies semi-nonparametric estimation of random-coefficients logit demand using aggregate market-share data and a large-(J) asymptotic framework. I take a complementary approach by using micro-level choice data and exploiting within-market variation across consumers instead of using market-level data to flexibly estimate aggregate demand or the distribution of random coefficients. This difference changes both the source of identifying variation and the relevant asymptotic framework: the number of markets and the number of consumers per market grow, while the number of products is treated as fixed.

Fourth, the paper contributes to the literature on nonparametric IV and sieve minimum distance. \cite{newey_instrumental_2003} show how completeness identifies nonparametric IV models, while \cite{hall_nonparametric_2005}, \cite{horowitz_applied_2011}, and related work study the ill-posedness and regularization issues that arise in nonparametric IV problems. \cite{ai_efficient_2003} develop efficient estimation for conditional moment models with unknown functions, and \cite{chen_sieve_2015} develops inference methods for semi/nonparametric conditional moment models. The price-side of my estimator follows this logic: the residual, \(a_t-\psi(P_t,X_t)\), must have zero conditional mean given \((W_t,X_t)\), and the conditional mean is estimated by projection onto an expanding basis. This use of growing unconditional moments to approximate a conditional moment is also related to \cite{donald_empirical_2003}. The distinctive feature of my work is that the residual contains a profiled object, \(a_t\), estimated from within-market micro data.

Finally, this paper relates to work on fixed effects and incidental parameters in nonlinear models. The inner profile step estimates a vector of market-specific intercepts whose dimension grows with $T$. In classical nonlinear panel settings, such nuisance parameters create incidental-parameters bias when each fixed effect is estimated with a fixed number of observations \citep{hahn_jackknife_2004}. In my setting, each market intercept is estimated using $n_t$ within-market observations, and $\min_{t\le T} n_t \to \infty$. This is adjacent to the logic in \cite{hahn_time-invariant_2005} where the nuisance parameters do not contaminate consistency of the common structural functions when the information per nuisance parameter diverges.

The rest of the paper proceeds as follows. Section \ref{sec:Motivation} provides a brief motivation for the estimator. Section \ref{sec:Model} presents the model and defines the composite market intercept. Section \ref{sec:Identification} states the identification argument and the normalizations used to separate the additive components. Section \ref{sec:Estimator} defines the profiled sieve minimum-distance estimator. Section \ref{sec:Consistency} states the main uniform consistency theorem. Section \ref{sec:Simulations} presents Monte Carlo evidence, and Section \ref{sec:Conclusion} concludes. The appendix contains the Berry-Haile primitive assumptions, all of the formal proofs, simulation details, and supplemental discussions regarding the curse of dimensionality.

\section{Motivation}\label{sec:Motivation}
Demand estimation is often used to answer counterfactual questions regarding markets. How would prices change after a cost shock? How much demand would divert from one product to another after a price increase? What are the implied markups for different firms? How would consumer surplus change after a new product, tax, or merger? In differentiated-product markets, these answers depend on the shape of the demand curve. A demand model that fits market shares but imposes the wrong substitution patterns or the wrong price curvature can give misleading counterfactual conclusions.

In applications, researchers would use the empirical distribution of the $n_t$ observed consumers in market $t$ to compute the aggregate inside-good share:
\[     \hat S_t(p)
    :=
    \frac1{n_t}\sum_{i=1}^{n_t}
    \sigma\!\left(
    \hat g(Z_{it},X_t)+\hat\psi(p,X_t)+\hat h_t
    \right), \]
where $p$ is a counterfactual price vector. From this object, one can compute own- and cross-price elasticities, diversion ratios, markup estimates, and cost shock pass-through. These objects depend on the slope or curvature of demand; hence, flexibly recovering the demand surface should improve the accuracy of the aforementioned estimated counterfactual objects.

\clearpage

% ---------------------------------------------------------------------

\section{Model}\label{sec:Model}
% ---------------------------------------------------------------------
This section presents the structural model for which I develop an estimator. This is done in two steps. First, I describe the high-level model studied by \cite{berry_nonparametric_2024} that motivates this paper. In particular, demand is generated by an index that separates within-market consumer heterogeneity from market-level demand shocks. Second, I impose additional structure used for estimation: a multinomial-logit link, a flexible price-side function, and a composite market intercept that can be profiled from within-market micro data. The location normalizations needed to separate the additive components are imposed in Section \ref{sec:Identification}, not in the model statement itself.

\subsection{High-level Model}

Markets are indexed by $t=1,\dots,T$, and consumers in market $t$ are indexed by $i=1,\dots,n_t$. Each consumer chooses among $J$ inside goods and an outside option. Let $\Jcal=\{1,\dots,J\}$  be the set of alternatives and let $Q_{it}=(Q_{1it},\dots,Q_{Jit})'\in\{0,1\}^J$ denote the vector of inside-good choice indicators. Thus $Q_{jit}=1$ if consumer $i$ in market $t$ chooses product $j$, while the outside option is chosen when all components of $Q_{it}$ are equal to zero.

Consumer covariates are denoted by $Z_{it}\in\Z\subseteq\R^{d_z}$. Market-level prices, observed product/market characteristics, and excluded price instruments are denoted by $P_t\in\Pcal\subseteq\R^J$, $X_t\in\X\subseteq\R^{d_x}$, and $W_t\in\Wcal \subseteq \R^{d_w}$, respectively. Market-level unobserved demand shocks are collected in $\Xi_t\in\R^J$. Here $d_z$ denotes the dimension of the consumer covariates, $d_x$ denotes the dimension of the observed product/market characteristics, and $d_w$ denotes the dimension of the excluded price instruments. Throughout the paper, a prime denotes transpose. Following \cite{berry_nonparametric_2024}, the full choice environment faced by consumer $i$ in market $t$ is
\[
    C_{it}:=(Z_{it},P_t,X_t,\Xi_t).
\]
The structural choice probability vector is
\begin{equation}
\label{eq:str_choice_prob}
    s(C_{it})=(s_1(C_{it}),\dots,s_J(C_{it}))'
    :=\E[Q_{it}\mid C_{it}],
\end{equation}
with outside-good probability $s_0(C_{it})=1-\sum_{j=1}^J s_j(C_{it})$. From here, we get to a model of demand through the Berry--Haile representation -- i.e., the existence of a finite-dimensional index that enters an otherwise nonparametric demand link.

\begin{assumption}[Berry--Haile index structure]
\label{ass:BH-index-main}
There exists an unknown link function $\sigma^{BH}$ and an index $\gamma:\Z\times\X\times\R^J\to\R^J$ such that
\begin{equation}
\label{eq:BH-index}
    s(C_{it})
    =
    \sigma^{BH}\!\left(\gamma(Z_{it},X_t,\Xi_t),P_t,X_t\right).
\end{equation}
Moreover, the index is additively separable in consumer-level observables and market-level unobservables: for an unknown function $\Gamma_0:\Z\times\X\to\R^J$,
\begin{equation}
\label{eq:BH-additive-index}
    \gamma(Z_{it},X_t,\Xi_t)=\Gamma_0(Z_{it},X_t)+\Xi_t.
\end{equation}
\end{assumption}

Assumption \ref{ass:BH-index-main} is the economic foundation for using micro data. Within a market, the objects $(P_t,X_t,\Xi_t)$ are fixed across consumers, while $Z_{it}$ varies. Thus, within-market variation in $Z_{it}$ shifts the index without changing the market-level environment. This is the source of variation that Berry and Haile exploit to identify conditional demand without instruments for quantities.

For the estimator, it is convenient to rewrite the Berry--Haile index as a sum of a consumer-heterogeneity function and a market-level shock. I write this as
\begin{equation}
\label{eq:BH-gh-form}
    s(C_{it})
    =
    \sigma^{BH}\!\left(g_0(Z_{it},X_t)+h_{0t},P_t,X_t\right),
\end{equation}
where $g_0:\Z\times\X\to\R^J$ captures the component that varies with consumer covariates, and $h_{0t}\in\R^J$ captures the market-level demand shock. Specifically, $g_0(Z_{it},X_t)\equiv\Gamma_0(Z_{it},X_t)+\E[\Xi_t\mid X_t]$ and $h_{0t}\equiv\Xi_t-\E[\Xi_t\mid X_t]$, such that their sum recovers equation \eqref{eq:BH-additive-index}. At this point, the decomposition between $g_0$ and $h_{0t}$ is only defined up to location normalizations. Those normalizations are stated in Section \ref{sec:Identification}.

\subsection{Semi-nonparametric Restriction}

This paper estimates a specialization of \eqref{eq:BH-gh-form}. The fully nonparametric Berry--Haile link $\sigma^{BH}$ is replaced by a multinomial-logit link, but the systematic role of prices and observed product/market characteristics is left flexible through an unknown function, $\psi_0$.

\begin{assumption}[Logit link with flexible price-side function]
\label{ass:logit-psi-main}
Let $h_0(x,\xi)$ denote the normalized market shock evaluated at $(x,\xi)$, so that $h_{0t}=h_0(X_t,\Xi_t)$. There exists an unknown function $\psi_0:\Pcal\times\X\to\R^J$ such that, for every $(z,p,x,\xi)$ in the support of the data,
\begin{equation}
\label{eq:logit-specialization}
    \sigma^{BH}\!\left(g_0(z,x)+h_0(x,\xi),p,x\right)
    =
    \sigma\!\left(g_0(z,x)+\psi_0(p,x)+h_0(x,\xi)\right),
\end{equation}
where $\sigma:\R^J\to(0,1)^J$ is the multinomial-logit map \citep{mcfadden_measurement_1974}. For any $u=(u_1,\dots,u_J)'\in\R^J$,
\begin{equation}
\label{eq:sigma-main}
    \sigma_j(u)=\frac{\exp\{u_j\}}{1+\sum_{k=1}^J\exp\{u_k\}},\quad j\in\Jcal,
    \qquad
    \sigma_0(u)=\frac{1}{1+\sum_{k=1}^J\exp\{u_k\}}.
\end{equation}
\end{assumption}

Under Assumptions \ref{ass:BH-index-main} and \ref{ass:logit-psi-main}, the maintained estimable demand model is
\begin{equation}
\label{eq:maintained-choice-model}
    \E[Q_{it}\mid C_{it}]
    =
    \sigma\!\left(g_0(Z_{it},X_t)+\psi_0(P_t,X_t)+h_{0t}\right).
\end{equation}
This is the main modeling restriction imposed by this paper beyond the Berry--Haile identification framework. The logit link makes the estimator tractable and gives a simple share-matching problem within each market. At the same time, the model remains semi-nonparametric because both $g_0$ and $\psi_0$ are unknown functions. The function $g_0$ captures how consumer covariates shift relative tastes for inside goods, while $\psi_0$ captures the systematic role of prices and observed product/market characteristics. The vector $h_{0t}$ captures product-level demand shocks that are common to all consumers within market $t$.

The object directly learned from the within-market micro data is not $h_{0t}$ by itself. Holding market $t$ fixed, both $\psi_0(P_t,X_t)$ and $h_{0t}$ are constant across consumers and enter \eqref{eq:maintained-choice-model} only through their sum. Following the fixed-effect profiling logic in \cite{berry_automobile_1995} and \cite{hahn_time-invariant_2005}, I define the composite market intercept
\begin{equation}
\label{eq:a0-def-main}
    a_{0t}:=\psi_0(P_t,X_t)+h_{0t}.
\end{equation}
Then the model can be written as
\begin{equation}
\label{eq:choice-composite-main}
    \E[Q_{it}\mid C_{it}]
    =
    \sigma\!\left(g_0(Z_{it},X_t)+a_{0t}\right).
\end{equation}
This representation is central for both computation and asymptotics. For a candidate $g$, the inner step of the estimator will profile the market-specific vector $a_t$ from the within-market choice data. The price-side step then decomposes the profiled intercept into the systematic component $\psi_0(P_t,X_t)$ and the structural shock $h_{0t}$ using excluded price instruments. 

The market-specific objects $a_{0t}$ and $h_{0t}$ are nuisance objects. They are nevertheless economically important because, after estimating $(g_0,\psi_0)$, the structural market shock is recovered as $h_{0t}=a_{0t}-\psi_0(P_t,X_t)$.

The target structural parameter is
\begin{equation}
\label{eq:theta0-main}
    \theta_0:=(g_0,\psi_0)\in\ThetaSpace:=\G\times\bfPsi.
\end{equation}
Here $\G$ and $\bfPsi$ are function spaces. Let $\mathcal D_g:=\supp(Z_{it},X_t)$ and $\mathcal D_\psi:=\supp(P_t,X_t).$ The space $\G$ contains normalized, bounded, smooth functions $g:\mathcal D_g\to\R^J$, and the space $\bfPsi$ contains bounded, smooth functions $\psi:\mathcal D_\psi\to\R^J$. Heuristically, 
\begin{equation}
\label{eq:function-spaces-main}
\begin{aligned}
    \G
    &\subseteq
    \left\{
    g:\mathcal D_g\to\R^J:
    g_j\in C^{s_g}_{B_g}(\mathcal D_g),\ j\le J,
    \text{ $g$ satisfies the normalization in Assumption \ref{ass:normalizations-main}}
    \right\}, \\
    \bfPsi
    &\subseteq
    \left\{
    \psi:\mathcal D_\psi\to\R^J:
    \psi_j\in C^{s_\psi}_{B_\psi}(\mathcal D_\psi),\ j\le J
    \right\},
\end{aligned}
\end{equation}
where $C^s_B(D)$ denotes a bounded H\"older ball of smoothness $s$ with radius $B$ on the compact domain $D$. The smoothness and boundedness restrictions are regularity conditions used for compactness and sieve approximation. The normalization in $\G$ removes the part of $g$ that is constant across consumers within a market, while the normalization $\E[h_{0t}\mid X_t]=0$ fixes the location of $\psi_0$ relative to the structural shock. Section \ref{sec:Identification} presents and discusses these normalizations formally. Similarly, Appendix \ref{app:function-spaces} defines these function spaces more explicitly and explains why these restrictions imply compactness under the sup norm.

The consistency analysis uses a double-asymptotic sampling scheme. Markets are independent across $t$, and conditional on the market-level environment, the observations within each market are independent.

\begin{assumption}[Sampling]
\label{ass:sampling-main}
Let $\Mcal_t:=(P_t,X_t,W_t,\Xi_t)$ be a market environment so that $\{\Mcal_t\}_{t\ge1}$ are identically distributed. Further, $\left\{
        \left(\Mcal_t,\{(Q_{it},Z_{it})\}_{i=1}^{n_t}\right)
    \right\}_{t=1}^T$ are independent across $t$. Conditional on
$\Mcal_t$, the observations
$\{(Q_{it},Z_{it})\}_{i=1}^{n_t}$ are i.i.d., with a conditional
distribution that is common across markets. The sample sizes $n_t$ are
nonrandom and satisfy
\[
    T\to\infty,
    \qquad
    n_{\min}:=\min_{t\le T}n_t\to\infty.
\]
\end{assumption}

The number of market-specific intercepts grows with $T$, but each intercept is estimated using a growing number of within-market observations. This is why the market effects are nuisance parameters without creating the usual incidental-parameters problem associated with a fixed amount of information per nuisance parameter, as discussed in \cite{hahn_time-invariant_2005}. The estimator combines many markets with large within-market micro samples.

% ---------------------------------------------------------------------
\section{Identification}\label{sec:Identification}
% ---------------------------------------------------------------------
This section states the identification logic needed for the main text. In general, \cite{berry_nonparametric_2024} shows that the high-level demand system in equation (\ref{eq:BH-gh-form}) can be nonparametrically identified under the assumptions stated in Appendix \ref{app:BH-assumptions}. This paper fundamentally relies on their identification result; however, since I impose additional restrictions in equation (\ref{eq:maintained-choice-model}), I will show that the structural functions are identified, given the logit link.

Towards this, I state the main assumptions my estimator relies on and explain how they correspond to the Berry--Haile primitives. There are two objects to identify. The first is the within-market component of demand: the consumer heterogeneity function, $g_0$, and the composite market intercept, $a_{0t}$. The second is the price-side decomposition of the composite intercept $a_{0t} = \psi_0(P_t, X_t) + h_{0t}$.

The first step follows from the Berry--Haile index, support, invertibility, and normalization assumptions. The second step follows from their price-instrument and completeness assumptions.

\subsection{Within-market Identification of $g_0$ and $a_{0t}$}

The additive index contains multiple components, so location normalizations are required before identification can be stated. These normalizations do not impose economic substitution patterns; they simply assign the level of the index to particular components.

\begin{assumption}[Location normalizations]
\label{ass:normalizations-main}
The normalized parameter space imposes the following restrictions.
\begin{enumerate}[label=(\roman*),leftmargin=2em]
    \item The structural market shock satisfies
    \begin{equation}
    \label{eq:h-normalization-main}
    \E[h_{0t}\mid X_t]=0
    \qquad\text{a.s.}
    \end{equation}

    \item The function $g_0$ contains no pure $X$-only component. One convenient normalization is
    \begin{equation}
    \label{eq:g-normalization-main}
    \E[g_0(Z_{it},X_t)\mid X_t=x]=0
    \qquad\text{for all }x\in\X.
    \end{equation}
    An alternative, equivalent for identification purposes, is the baseline normalization
    \begin{equation}
    \label{eq:g-baseline-normalization-main}
        g_0(z^0(x),x)=0
        \qquad\text{for all }x\in\X,
    \end{equation}
    where $z^0(x)$ is a known baseline point in the support of $Z_{it}$ conditional on $X_t=x$.

    \item The outside option has index normalized to zero, as embedded in the denominator of the logit probabilities in \eqref{eq:sigma-main}.
\end{enumerate}
\end{assumption}

The first normalization follows from the construction of $h_{0t}$. The third is the standard outside-option normalization in multinomial logit models. The second normalization, coming directly from \cite{berry_nonparametric_2024}, fixes the location of $g_0$ relative to the market-level intercept. Without \eqref{eq:g-normalization-main} or \eqref{eq:g-baseline-normalization-main}, a component of $g_0$ that depends only on $X_t$ would be indistinguishable from the market intercept. See Appendix \ref{app:normalizations} for a full equivalence-class argument.

The logit specification of $\sigma(\cdot)$ gives a simple way to see what the within-market data identify. Let $s_j(C_{it})$ and $s_0(C_{it})$ denote the inside-good and outside-good probabilities generated by \eqref{eq:choice-composite-main}. Then, for each $j\in\Jcal$,
\begin{equation}
\label{eq:log-odds-main}
    \log\frac{s_j(C_{it})}{s_0(C_{it})}
    =
    g_{0j}(Z_{it},X_t)+a_{0jt}.
\end{equation}
Within a market, $a_{0jt}$ is constant across consumers, while $g_{0j}(Z_{it},X_t)$ varies with $Z_{it}$. Hence, after the normalization on $g_0$, within-market variation identifies the heterogeneity function and the composite market intercepts. The estimator will implement this logic by profiling $a_t$ market-by-market for each candidate $g$.

\begin{proposition}[Within-market identification]
\label{prop:within-market-id-main}
Suppose the Berry--Haile primitive assumptions stated in Appendix \ref{app:BH-assumptions} hold, and suppose the logit specialization in \eqref{eq:choice-composite-main} is correctly specified. Then, subject to the location normalizations in Assumption \ref{ass:normalizations-main}, $(g_0, a_{0t})$ are identified from the conditional choice probabilities.
\end{proposition}
\begin{proof}
    See Appendix \ref{app:identification}.
\end{proof}

\subsection{Price-side Identification of $\psi_0$}

The within-market step identifies the composite object $a_{0t}$, but it does not by itself separate $\psi_0(P_t,X_t)$ from $h_{0t}$. This separation is obtained from excluded price instruments. The relevant price-side restriction is inherited from the Berry--Haile instrument assumption in Appendix \ref{app:BH-assumptions}.

\begin{assumption}[Price-side instruments and completeness]
\label{ass:price-id-main}

\hfill

\begin{enumerate}[label=(\roman*),leftmargin=2em]
    \item The excluded instruments satisfy the price-side validity condition
\begin{equation}
\label{eq:price-iv-main}
    \E[h_{0t}\mid W_t,X_t]
    =
    \E[h_{0t}\mid X_t]
    \qquad\text{a.s.}
\end{equation}

\item For the relevant class of vector-valued functions $\bm{\mathcal F}:\Pcal\times\X\to\R^J$ with finite expectation,
\[
    \E[\bm{\mathcal F}(P_t,X_t)\mid W_t, X_t]=0\quad\text{a.s.}
    \quad\Longrightarrow\quad
    \bm{\mathcal F}(P_t,X_t)=0\quad\text{a.s.}
\]
\end{enumerate}

\end{assumption}

Assumption \ref{ass:price-id-main} is a restatement of the Berry--Haile price-instrument and completeness conditions. Equation \eqref{eq:price-iv-main} says that, after conditioning on $X_t$, the excluded instruments do not predict the structural demand shock. Combining \eqref{eq:price-iv-main} with the normalization $\E[h_{0t}\mid X_t]=0$ gives
\[
    \E[h_{0t}\mid W_t,X_t]=0.
\]
Since $h_{0t}=a_{0t}-\psi_0(P_t,X_t)$, the identified composite intercept satisfies the price-side conditional moment
\begin{equation}
\label{eq:psi-CM-main}
    \E\!\left[a_{0t}-\psi_0(P_t,X_t)\mid W_t,X_t\right]=0.
\end{equation}
This NPIV-style moment condition is what is used to identify and eventually estimate $\psi_0$.

\begin{theorem}[Identification of the normalized structural functions]
\label{thm:identification-main}
Suppose the maintained demand model \eqref{eq:maintained-choice-model} holds. Suppose further that the conditions of Proposition \ref{prop:within-market-id-main} and Assumption \ref{ass:price-id-main} hold. Then the pair $(g_0,\psi_0)$ is identified in the normalized parameter space. The corresponding market shocks are identified by $h_{0t}=a_{0t}-\psi_0(P_t,X_t).$
\end{theorem}

\begin{proof}
    See Appendix \ref{app:identification}
\end{proof}

% ---------------------------------------------------------------------

\section{Estimation}\label{sec:Estimator}
% ---------------------------------------------------------------------

This section defines the estimator that will then be analyzed in Section \ref{sec:Consistency}. Its construction mirrors the identification argument in that the within-market data identify, for each candidate heterogeneity function, the composite market intercept that shifts all consumers in the same market. The cross-market price-side moment then separates this composite intercept into the systematic price component and the structural demand shock.

Before introducing the sample criterion, it is useful to isolate the two population restrictions that make the estimator valid. I define the structural choice residual at the truth by
\begin{equation}
\label{eq:structural-residual-main}
    \varepsilon_{it}^0
    :=
    Q_{it}-\sigma\!\left(g_0(Z_{it},X_t)+a_{0t}\right).
\end{equation}
The first restriction is the structural residual condition,
\begin{equation}
\label{eq:structural-residual-condition-main}
    \E[\varepsilon_{it}^0\mid C_{it}]=0,
\end{equation}
which follows directly from the maintained choice model. The second restriction is the price-side conditional moment,
\begin{equation}
\label{eq:price-side-pop-restriction-main}
    \E[a_{0t}-\psi_0(P_t,X_t)\mid W_t,X_t]=0,
\end{equation}
which follows from $a_{0t}=\psi_0(P_t,X_t)+h_{0t}$, the normalization $\E[h_{0t}\mid X_t]=0$, and the price-instrument condition $\E[h_{0t}\mid W_t,X_t]=\E[h_{0t}\mid X_t]$.

\begin{proposition}[Population restrictions underlying the estimator]
\label{prop:population-restrictions-main}
Under the maintained demand model, the structural residual condition \eqref{eq:structural-residual-condition-main} holds. Consequently, for any square-integrable function $\varphi(Z_{it},X_t)$,
\begin{equation}
\label{eq:generic-micro-pop-restriction-main}
    \E[\varphi(Z_{it},X_t)\otimes\varepsilon_{it}^0]=0.
\end{equation}
If, in addition, Assumptions \ref{ass:normalizations-main} and \ref{ass:price-id-main} hold, then the price-side restriction \eqref{eq:price-side-pop-restriction-main} holds.
\end{proposition}

\begin{proof}
See Appendix \ref{app:population-restrictions}.
\end{proof}

In words, the estimator works as follows. First, choose finite-dimensional sieve spaces for $g_0$ and $\psi_0$. Second, for any candidate $g$, choose the market intercepts $a_t$ that make predicted market shares match observed market shares. Third, use the resulting profiled choice residuals to estimate $g_0$ from within-market variation in $Z_{it}$. Fourth, choose $\psi$ so that the recovered shocks $a_t-\psi(P_t,X_t)$ have zero conditional mean given $(W_t,X_t)$. Fifth, after estimating $(g_0,\psi_0)$, recover the market-specific composite intercepts and structural demand shocks.

\subsection{Structural Sieves and Norms}

Let $\mathcal D_g:=\supp(Z_{it},X_t)$ and $\mathcal D_\psi:=\supp(P_t,X_t)$. For two vector-valued functions $g,\tilde g:\mathcal D_g\to\R^J$ and $\psi,\tilde\psi:\mathcal D_\psi\to\R^J$, define
\begin{equation}
\label{eq:linfty-norms-main}
    \norm{g-\tilde g}_{\infty}
    :=
    \max_{j\le J}\sup_{(z,x)\in\mathcal D_g}
    \abs{g_j(z,x)-\tilde g_j(z,x)},
    \qquad
    \norm{\psi-\tilde\psi}_{\infty}
    :=
    \max_{j\le J}\sup_{(p,x)\in\mathcal D_\psi}
    \abs{\psi_j(p,x)-\tilde\psi_j(p,x)}.
\end{equation}
The product norm on the structural parameter space is
\begin{equation}
\label{eq:theta-norm-main}
    \norm{(g,\psi)-(\tilde g,\tilde\psi)}_{\ThetaSpace}
    :=
    \norm{g-\tilde g}_{\infty}
    +
    \norm{\psi-\tilde\psi}_{\infty}.
\end{equation}
The term ``product norm'' means the distance used on the Cartesian product $\ThetaSpace=\G\times\bfPsi$: it measures the distance between two structural parameters by adding the supremum-norm distance between their $g$ components and the supremum-norm distance between their $\psi$ components. Thus convergence in $\norm{\cdot}_{\ThetaSpace}$ is equivalent to the joint statement $\norm{g-\tilde g}_\infty\to0$ and $\norm{\psi-\tilde\psi}_\infty\to0$. The subscript $\infty$ denotes the usual uniform, or $\ell_\infty$, norm: for vector-valued functions, I take the largest absolute error over products and over the relevant support. Because the objects of interest are functions, uniform consistency is the natural notion of consistency as it requires the entire estimated demand system to converge over its domain, rather than only at fixed points.

Let $b^{K_g}(z,x)=(b_1(z,x),\dots,b_{K_g}(z,x))'\in\R^{K_g}$ be a basis for the consumer-heterogeneity function, and let $r^{K_\psi}(p,x)=(r_1(p,x),\dots,r_{K_\psi}(p,x))'\in\R^{K_\psi}$ be a basis for the price-side function. For each product $j\in\Jcal$, write
\[
    g_{\beta,j}(z,x)=b^{K_g}(z,x)'\beta_j,
    \qquad
    \psi_{\vartheta,j}(p,x)=r^{K_\psi}(p,x)'\vartheta_j,
\]
where $\beta_j\in\R^{K_g}$, $\vartheta_j\in\R^{K_\psi}$, $\beta=(\beta_1',\dots,\beta_J')'\in\R^{JK_g}$, and $\vartheta=(\vartheta_1',\dots,\vartheta_J')'\in\R^{JK_\psi}$. The finite-dimensional structural sieve is
\begin{equation}
\label{eq:structural-sieve-main}
    \ThetaSpace_K:=\G_{K_g}\times\bfPsi_{K_\psi},
    \qquad
    K:=(K_g,K_\psi,K_q),
\end{equation}
where $\ThetaSpace_K$ depends on the structural dimensions $(K_g,K_\psi)$, while the full sample criterion -- defined below -- also depends on the conditional-moment dimension, $K_q$. Thus $K_g$ controls both the sieve for $g$ and the dimension of the micro moment, $K_\psi$ controls only the sieve for $\psi$, and $K_q$ controls only the series approximation to the price-side conditional mean. The component sieves are
\[
\begin{aligned}
    \G_{K_g}
    &:=
    \{g_\beta:\beta\in B_{K_g},\ g_\beta\text{ satisfies the normalization on }g\},\\
    \bfPsi_{K_\psi}
    &:=
    \{\psi_\vartheta:\vartheta\in V_{K_\psi}\},
\end{aligned}
\]
for compact coefficient sets $B_{K_g}$ and $V_{K_\psi}$, respectively.\footnote{In practice, this can be implemented either by imposing the resulting linear restrictions on $\beta$ or by using a centered basis, for example $b^{K_g}(z,x)-\E[b^{K_g}(Z_{it},x)\mid X_t=x]$ under the conditional-mean normalization, or $b^{K_g}(z,x)-b^{K_g}(z^0(x),x)$ under the baseline normalization. The point is that the estimator never treats an $X$-only term in $g$ as separately estimable from the market intercept.}

\subsection{Profiling the Composite Market Intercept}

Fix a candidate $g\in\G_{K_g}$. In market $t$, the object that is constant across consumers is not $h_{0t}$ alone, but the composite intercept $a_{0t}=\psi_0(P_t,X_t)+h_{0t}$. The estimator therefore profiles $a_t$ market-by-market. Let $\hat s_t=n_t^{-1}\sum_{i=1}^{n_t}Q_{it}$. For $a\in\Acal\subset\R^J$, define
\begin{equation}
\label{eq:inner-potential-main}
    \hat L_t(g,a)
    :=
    \frac{1}{n_t}\sum_{i=1}^{n_t}
    \log\!\left(1+\sum_{j=1}^J\exp\{g_j(Z_{it},X_t)+a_j\}\right)
    -\hat s_t'a .
\end{equation}
The profiled composite intercept is
\begin{equation}
\label{eq:ahat-profile-main}
    \hat a_t(g)\in\argmin_{a\in\Acal}\hat L_t(g,a).
\end{equation}
The log-sum-exp criterion is used only as a convex objective whose gradient equals the negative of the within-market share-matching moment. When the minimizer is interior, its first-order condition is exactly the within-market share-matching equation
\begin{equation}
\label{eq:share-matching-main}
    \frac{1}{n_t}\sum_{i=1}^{n_t}
    \left\{Q_{it}-\sigma(g(Z_{it},X_t)+\hat a_t(g))\right\}=0.
\end{equation}
Thus the inner step is not introducing a separate likelihood estimator for the structural functions; it is just a convenient way to solve the just-identified within-market moment for the composite intercept.

For later reference, let $a_t^*(g)$ denote the population counterpart of \eqref{eq:ahat-profile-main}. At the truth, $a_t^*(g_0)=a_{0t}$; Appendix \ref{app:inner-profiling} proves this claim and gives primitive conditions under which the sample profile converges uniformly to $a_t^*(g)$.

\subsection{Micro Block}

The micro block estimates $g_0$ using within-market variation in $Z_{it}$. Define the sample-profiled residual
\begin{equation}
\label{eq:profiled-residual-main}
    \hat\rho_{it}(g):=Q_{it}-\sigma(g(Z_{it},X_t)+\hat a_t(g)),
\end{equation}
and define its population-profile analogue by
\begin{equation}
\label{eq:profiled-residual-pop-main}
    \rho_{it}^*(g):=Q_{it}-\sigma(g(Z_{it},X_t)+a_t^*(g)).
\end{equation}
For each $K_g$, the feasible micro moment is
\begin{equation}
\label{eq:micro-moment-main}
    \hat m_{g,K_g}(g)
    :=
    \frac{1}{N}\sum_{t=1}^T\sum_{i=1}^{n_t}
    b^{K_g}(Z_{it},X_t)\otimes \hat\rho_{it}(g),
    \qquad
    N:=\sum_{t=1}^T n_t .
\end{equation}
Its population counterpart is
\begin{equation}
\label{eq:micro-moment-pop-main}
    m_{g,K_g}(g)
    :=
    \E\!\left[b^{K_g}(Z_{it},X_t)\otimes\rho_{it}^*(g)\right]
    \in\R^{JK_g}.
\end{equation}
The subscript $K_g$ is substantive as the dimension of the moment changes with the number of basis functions. Here $\otimes$ denotes the Kronecker product: if $u\in\R^{K_g}$ and $v\in\R^J$, then $u\otimes v\in\R^{K_gJ}$ stacks all products $u_kv_j$. This moment is centered at the truth only after the profile identity $a_t^*(g_0)=a_{0t}$ has been established.

\begin{proposition}[Validity of the profiled micro moment]
\label{prop:micro-valid-main}
Under the maintained demand model and the population profile definition,
\[
    m_{g,K_g}(g_0)=0
    \qquad\text{for every }K_g.
\]
\end{proposition}

\begin{proof}
See Appendix \ref{app:inner-profiling}.
\end{proof}

With a positive definite weighting matrix $W_{g,K_g}\in\R^{JK_g\times JK_g}$, define the micro-block criterion
\begin{equation}
\label{eq:calQg-main}
    \hat{\mathcal Q}_{g,K_g}(g):=
    \hat m_{g,K_g}(g)'W_{g,K_g}\hat m_{g,K_g}(g).
\end{equation}
The same basis $b^{K_g}$ appears in the sieve for $g$ and in the micro moments. This is intentional as the moment projects the residual onto the same directions used to perturb the unknown function. In a finite-dimensional logit model these would be score equations for the coefficients of $g$; here they are used as GMM moments for a growing sieve.

\subsection{Price-side SMD Block}

For a candidate pair $(g,\psi)\in\ThetaSpace_K$, define the implied structural demand shock
\begin{equation}
\label{eq:hhat-implied-main}
    \hat h_t(g,\psi):=\hat a_t(g)-\psi(P_t,X_t).
\end{equation}
The population analogue is $h_t^*(g,\psi):=a_t^*(g)-\psi(P_t,X_t)$. At the truth, this implied shock equals $h_{0t}$ once $a_t^*(g_0)=a_{0t}$. This connects the profiled market intercept and the price-side conditional moment: the object projected onto functions of $(W_t,X_t)$ is centered at the truth.

Let $R_t:=(W_t,X_t) \in \mathcal{R}$ where $ \mathcal R:=\supp(R_t)\subseteq\Wcal\times\X$ so that $\mathcal R$ is the support of the conditioning variables in the price-side conditional moment. Let $q^{K_q}(r)=(q_1(r),\dots,q_{K_q}(r))'\in\R^{K_q}$ be a series basis for functions of $r\in\mathcal R$. Following \cite{ai_efficient_2003}, the population conditional moment for the price-side block is
\begin{equation}
\label{eq:mpsi-pop-main}
    m_\psi(r;g,\psi)
    :=
    \E[h_t^*(g,\psi)\mid R_t=r]
    =
    \E[a_t^*(g)-\psi(P_t,X_t)\mid R_t=r]
    \in\R^J.
\end{equation}
At the truth, the price-side conditional moment is zero.

\begin{proposition}[Validity of the price-side SMD moment]
\label{prop:price-smd-valid-main}
Under the maintained demand model, the profile identity $a_t^*(g_0)=a_{0t}$, and Assumptions \ref{ass:normalizations-main} and \ref{ass:price-id-main},
\[
    m_\psi(R_t;g_0,\psi_0)=0
    \qquad\text{a.s.}
\]
\end{proposition}

\begin{proof}
See Appendix \ref{app:SMD-details}.
\end{proof}

Let $\bm Q_T\in\R^{T\times K_q}$ be the matrix with $t$-th row $q^{K_q}(R_t)'$, and assume $\bm Q_T'\bm Q_T$ is nonsingular on the event used to compute the criterion. Following \cite{ai_efficient_2003}, the series estimate of the conditional mean of the implied shock is
\begin{equation}
\label{eq:mpsi-hat-main}
    \hat m_{\psi,K_q}(r;g,\psi)
    :=
    \left[\sum_{t=1}^T\hat h_t(g,\psi)q^{K_q}(R_t)'\right]
    (\bm Q_T'\bm Q_T)^{-1}q^{K_q}(r)
    \in\R^J.
\end{equation}
In this expression, $\hat h_t(g,\psi)q^{K_q}(R_t)'$ is a $J\times K_q$ matrix, so the bracketed sum is also $J\times K_q$. The corresponding SMD criterion is
\begin{equation}
\label{eq:calQpsi-hat-main}
    \hat{\mathcal Q}_{\psi,K_q}(g,\psi)
    :=
    \frac{1}{T}\sum_{t=1}^T
    \hat m_{\psi,K_q}(R_t;g,\psi)'\hat\Sigma(R_t)^{-1}
    \hat m_{\psi,K_q}(R_t;g,\psi),
\end{equation}
where $\hat\Sigma(R_t)\in\R^{J\times J}$ is uniformly positive definite. Equivalently, when $\hat\Sigma(R_t)=I_J$, this is GMM with the increasing set of moments generated by $q^{K_q}(R_t)\otimes\{\hat a_t(g)-\psi(P_t,X_t)\}$; as $K_q$ grows, these unconditional moments approximate the conditional moment in \eqref{eq:mpsi-pop-main}, as in \cite{ai_efficient_2003} and \cite{donald_empirical_2003}.

\subsection{Full Profiled Estimator}

The full sample criterion is
\begin{equation}
\label{eq:calQhat-full-main}
    \hat{\mathcal Q}_K(g,\psi)
    :=
    \hat{\mathcal Q}_{g,K_g}(g)
    +
    \hat{\mathcal Q}_{\psi,K_q}(g,\psi).
\end{equation}
The estimator is any approximate minimizer\footnote{The use of an approximate minimizer is standard in extremum-estimation consistency arguments. It avoids requiring the numerical optimizer to attain the exact global minimum of a nonconvex finite-sample criterion; the proof only needs the optimizer to achieve the infimum up to an $o_p(1)$ error. See \cite{newey_chapter_1994}.} over the structural sieve:
\begin{equation}
\label{eq:estimator-main}
    \hat{\mathcal Q}_K(\hat g,\hat\psi)
    \le
    \inf_{(g,\psi)\in\ThetaSpace_K}\hat{\mathcal Q}_K(g,\psi)+o_p(1).
\end{equation}
The final profiled composite intercepts and recovered structural shocks are
\begin{equation}
\label{eq:ahat-final-main}
    \hat a_t:=\hat a_t(\hat g),
    \qquad
    \hat h_t:=\hat a_t-\hat\psi(P_t,X_t).
\end{equation}
The estimated demand system can then be used for counterfactual analysis. Holding the recovered shock fixed, predicted demand in market $t$ at price vector $p$ is
\begin{equation}
\label{eq:share-counterfactual-main}
    \hat S_t(p)
    :=
    \frac1{n_t}\sum_{i=1}^{n_t}
    \sigma\{\hat g(Z_{it},X_t)+\hat\psi(p,X_t)+\hat h_t\}.
\end{equation}
Differentiating \eqref{eq:share-counterfactual-main} with respect to prices gives local elasticities and diversion ratios. Combining those derivatives with an ownership matrix gives Bertrand markups and pass-through. Appendix \ref{app:io-objects} details these formulas. The present paper establishes consistency of the structural objects entering \eqref{eq:share-counterfactual-main}; inference for the resulting IO functionals is left for future work.

\clearpage

% ---------------------------------------------------------------------

\section{Consistency}\label{sec:Consistency}
% ---------------------------------------------------------------------

This section states the main consistency argument. I keep the main text at the level needed to understand why the estimator works and defer the primitive profiling and empirical-process details to Appendix \ref{app:inner-profiling}--\ref{app:main-consistency-proof}. The proof has two moving parts. First, each market-specific intercept is consistently profiled because $n_{\min}:=\min_{t\le T}n_t\to\infty$. Second, after replacing the population profile by this estimated profile, the outer criterion is a moving-sieve minimum-distance criterion. The word ``moving'' is important as the micro moment contains the basis vector $b^{K_g}$, and therefore the population micro criterion changes with $K_g$.

For each sieve dimension $K=(K_g,K_\psi,K_q)$, I define the population micro criterion
\begin{equation}
\label{eq:population-micro-criterion-main}
    \mathcal Q_{g,K_g}(g)
    :=
    m_{g,K_g}(g)'W_{g,K_g}m_{g,K_g}(g).
\end{equation}
The population price-side criterion is the full conditional-mean criterion
\begin{equation}
\label{eq:population-price-criterion-main}
    \mathcal Q_\psi(g,\psi)
    :=
    \E\!\left[
    m_\psi(R_t;g,\psi)'\Sigma(R_t)^{-1}m_\psi(R_t;g,\psi)
    \right].
\end{equation}
Thus the population criterion associated with the finite-sample criterion is
\begin{equation}
\label{eq:population-criterion-main}
    \mathcal Q_K(g,\psi)
    :=
    \mathcal Q_{g,K_g}(g)+\mathcal Q_\psi(g,\psi).
\end{equation}
The price-side sample criterion depends on $K_q$ through the series regression used to estimate the conditional mean, but its population target is the full conditional moment in \eqref{eq:population-price-criterion-main}. The micro population criterion, by contrast, is indexed by $K_g$ because the moment vector itself has dimension $JK_g$.

At the truth, both terms are zero for every $K$. The micro term is zero by Proposition \ref{prop:micro-valid-main}, and the price-side term is zero by Proposition \ref{prop:price-smd-valid-main}. Therefore
\[
    \mathcal Q_K(g_0,\psi_0)=0
    \qquad\text{for every }K.
\]

\begin{assumption}[Sieve approximation, compactness, and criterion preservation]
\label{ass:sieve-approx-main}
The normalized parameter space $\ThetaSpace=\G\times\bfPsi$ is compact under the product sup norm in \eqref{eq:theta-norm-main}. The true pair $\theta_0=(g_0,\psi_0)$ belongs to $\ThetaSpace$. Moreover, there exist sieve approximants
\[
    \theta_K^0:=(g_{0,K_g},\psi_{0,K_\psi})
    \in
    \ThetaSpace_K
\]
such that as $(K_g, K_\psi) \to \infty$ along the maintained sequence of sieve dimensions,
\begin{equation}
\label{eq:sieve-approx-main}
    \norm{\theta_K^0-\theta_0}_{\ThetaSpace}
    =
    \norm{g_{0,K_g}-g_0}_{\infty}
    +
    \norm{\psi_{0,K_\psi}-\psi_0}_{\infty}
    \to0,
\end{equation}
and the same sieve approximation is criterion-preserving: as $K \to \infty$,
\begin{equation}
\label{eq:criterion-preservation-main}
    \mathcal Q_K(\theta_K^0)\to0 .
\end{equation}
\end{assumption}

This assumption plays the usual role in sieve extremum estimation, with one additional requirement needed because the criterion moves with $K$. Compactness prevents the minimization problem from escaping to irregular functions with no convergent subsequence, while sieve approximation says that the finite-dimensional spaces eventually approach the true structural functions. The criterion-preservation condition says that these approximants are not merely close in sup norm; they are also close in the metric actually used by the moving population objective. This is automatic under suitable basis-envelope and weighting restrictions, but it is not implied by ordinary pointwise continuity of a fixed criterion because $\mathcal Q_K$ changes with $K$. Proposition
\ref{prop:app-sieve-compact-approx} in the appendix provides primitive sufficient conditions under which Assumption \ref{ass:sieve-approx-main} holds.

\begin{assumption}[Regularity of the profile map]
\label{ass:profile-main}
Let $\G_{K_g}^0:=\G_{K_g}\cup\{g_0\}$. For every
$g\in\G_{K_g}^0$ and every market $t$, $a_t^*(g)$ is the unique
interior minimizer of the population inner criterion, $L_t(g,a)$. The
empirical inner criteria converge uniformly to their population
counterparts over $(g,a,t)\in\G_{K_g}\times\Acal\times\{1,\ldots,T\}$,
and the population inner criterion is uniformly separated away from
$a_t^*(g)$ over $g\in\G_{K_g}^0$. Finally, the profile map is uniformly
Lipschitz on $\G_{K_g}^0$: for some $C<\infty$,
\[
    \max_{t\le T}\norm{a_t^*(g)-a_t^*(\tilde g)}
    \le C\norm{g-\tilde g}_{\infty}
    \qquad\text{for all }g,\tilde g\in\G_{K_g}^0.
\]
\end{assumption}

This assumption is the formal statement that the inner step is well posed. Uniqueness and separation ensure that the population intercept $a_t^*(g)$ is not ambiguous. The uniform law of large numbers is what lets each empirical market problem approximate its population counterpart, uniformly across candidate $g$ and markets. The Lipschitz condition says that a small error in $g$ cannot generate a large error in the population profile. For the logit link, these restrictions follow from bounded indices, compactness of $\Acal$, and the strong convexity of the log-sum-exp objective in $a$; see Appendix Proposition \ref{prop:app-profile-lipschitz}.

The next result is the first connection between the infeasible population problem and the feasible estimator: uniformly over candidate heterogeneity functions and markets, the empirical profile behaves like the population profile.

\begin{proposition}[Uniform consistency of the profiled intercept]
\label{prop:profile-consistency-main}
Under Assumptions \ref{ass:sampling-main} and \ref{ass:profile-main},
\[
    \sup_{g\in\G_{K_g}}\max_{t\le T}
    \norm{\hat a_t(g)-a_t^*(g)}=o_p(1).
\]
\end{proposition}

\begin{proof}
See Appendix \ref{app:inner-profiling}.
\end{proof}

\begin{assumption}[Moving population identification and separation]
\label{ass:population-id-main}
The moving population criterion is separated from zero away from the truth. That is, for every $\epsilon>0$,
\begin{equation}
\label{eq:moving-separation-main}
    \liminf_{K\to\infty}
    \inf_{\theta\in\ThetaSpace_K:
    \norm{\theta-\theta_0}_{\ThetaSpace}\ge\epsilon}
    \mathcal Q_K(\theta)
    >0 .
\end{equation}
\end{assumption}

This assumption is the argmin identification condition appropriate for a moving-sieve criterion. Here $K\to\infty$ is shorthand for taking the sieve and series dimensions $(K_g,K_\psi,K_q)$ along the same admissible sequence used by the estimator as the sample grows. Since the sample grows through both $T\to\infty$ and $n_{\min}\to\infty$, the dimension sequence may depend on both $T$ and $n_{\min}$; the growth restrictions in Assumption \ref{ass:criterion-UC-main} require it to grow slowly enough relative to both sources of sampling information. The liminf is needed because $\mathcal Q_K$ is not a single fixed population objective; rather, its micro component changes when the dimension of $b^{K_g}$ changes. The condition says that, eventually along this admissible sequence, the population criterion remains uniformly bounded away from zero outside every sup-norm neighborhood of the truth. Proposition \ref{prop:app-population-id} provides primitive sufficient conditions for this separation requirement.

\begin{assumption}[Uniform convergence of the profiled moving criterion]
\label{ass:criterion-UC-main}
The micro block satisfies a uniform law of large numbers, the feasible and oracle micro moments are uniformly asymptotically equivalent, the price-side SMD basis $q^{K_q}$ satisfies the usual Ai--Chen series conditions, the weighting matrix satisfies $\sup_r\norm{\hat\Sigma(r)-\Sigma(r)}=o_p(1)$\footnote{The specific matrix norm here is the \emph{operator norm}: $\norm{A} = \sup_v \norm{Av}_2$ for matrix $A$ and unit vector $v$. For symmetric matrices, this is the largest absolute eigenvalue.} with $\Sigma(r)$ uniformly positive definite\footnote{Uniform positive definiteness: for some $c>0$, $ \inf_{r\in\mathcal R}\lambda_{\min}\{\Sigma(r)\}\ge c$ where $\lambda_{\min}$ is the minimum eigenvalue of the matrix $\Sigma(r)$.}, and the dimensions $(K_g,K_\psi,K_q)$ grow slowly enough relative to $(T,n_{\min})$ that the profile error in Proposition \ref{prop:profile-consistency-main} is negligible in both the micro and price-side blocks.
\end{assumption}

This is the high-level stochastic equicontinuity condition for the full estimator; Appendix \ref{app:consistency-details} details the primitive sufficient conditions for each part. The micro block is a sieve GMM moment based on individual observations. Because it uses the generated profile $\hat a_t(g)$, Appendix \ref{app:SMD-details} first compares it with an oracle micro moment using $a_t^*(g)$. The price-side block is an Ai--Chen conditional-moment estimator implemented by a growing basis in $(W_t,X_t)$. The relevant high-level conditions are standard in sieve conditional-moment estimation; see \cite{ai_efficient_2003}, \cite{chen_large_2007}, and \cite{chen_sieve_2015}. Its generated-profile error is controlled in the empirical $L_2$ norm induced by the realized market values $R_1,\ldots,R_T$.

The next result is the second connection between the infeasible population problem and the feasible estimator. It says that, after profiling the market intercepts, the feasible moving criterion is uniformly close to the moving population criterion that identifies the truth.

\begin{proposition}[Uniform convergence of the full moving criterion]
\label{prop:full-UC-main}
Under Assumption \ref{ass:criterion-UC-main},
\[
    \sup_{(g,\psi)\in\ThetaSpace_K}
    \abs{\hat{\mathcal Q}_K(g,\psi)-\mathcal Q_K(g,\psi)}=o_p(1).
\]
\end{proposition}

\begin{proof}
See Appendix \ref{app:consistency-details}.
\end{proof}

The preceding assumptions separate the consistency proof into three claims: the profile step is uniformly well behaved, the feasible moving criterion converges uniformly to its population analogue, and the moving population criterion separates the truth from all other normalized sieve parameters. Combining these claims gives the main result.

\begin{theorem}[Consistency of the profiled SMD estimator]
\label{thm:consistency-main}
Suppose the maintained demand model holds. Suppose further that the identification conditions in Section \ref{sec:Identification}, the sampling condition in Assumption \ref{ass:sampling-main}, Assumptions \ref{ass:sieve-approx-main}--\ref{ass:profile-main} and Assumptions \ref{ass:population-id-main}--\ref{ass:criterion-UC-main} hold. Let $(\hat g,\hat\psi)$ satisfy \eqref{eq:estimator-main}. Then
\[
    \norm{\hat g-g_0}_{\infty}
    +
    \norm{\hat\psi-\psi_0}_{\infty}
    =o_p(1).
\]
Moreover, the profiled composite intercepts and recovered structural shocks satisfy
\[
    \max_{t\le T}\norm{\hat a_t-a_{0t}}=o_p(1),
    \qquad
    \max_{t\le T}\norm{\hat h_t-h_{0t}}=o_p(1).
\]
Consequently,
\[
    \frac1T\sum_{t=1}^T\norm{\hat a_t-a_{0t}}^2=o_p(1),
    \qquad
    \frac1T\sum_{t=1}^T\norm{\hat h_t-h_{0t}}^2=o_p(1).
\]
\end{theorem}

\begin{proof}
See Appendix \ref{app:main-consistency-proof}.
\end{proof}

% ---------------------------------------------------------------------
\section{Monte Carlo Simulation}\label{sec:Simulations}
% ---------------------------------------------------------------------

This section reports two Monte Carlo exercises. The first studies the consistency logic developed in Theorem \ref{thm:consistency-main}: as the number of markets and the number of consumers per market increase, the feasible profiled estimator should recover the structural functions and the generated market-level objects. The second compares the proposed flexible price-side estimator with standard differentiated-products demand benchmarks in counterfactual exercises of the kind used in empirical industrial organization. The point of the second exercise is to isolate a setting in which the price side of demand is nonlinear and to ask whether imposing a standard linear-in-price random-coefficients logit structure distorts counterfactual predictions. The comparison deliberately favors existing methods by using a heterogeneity structure that a random-coefficients logit model can represent well. The main difference across estimators is therefore the flexibility of the systematic price-side function.

The main text focuses on demonstrating uniform consistency of the estimator and comparing it to existing alternatives. Appendix \ref{app:simulation-details} details the data-generating process, the numerical implementation, the exact norm definitions, and the additional component diagnostics.

\subsection{Consistency Exercise: Design}

The baseline design has two inside goods, a scalar consumer covariate $Z_{it}$, a scalar observed market characteristic $X_t$, and two excluded instruments $W_t=(W_{1t},W_{2t})'$. Markets are independent. In each market, the structural composite intercept is
\[
    a_{0t}=\psi_0(P_t,X_t)+h_{0t},
\]
and individual choices are drawn from the multinomial-logit demand system
\[
    \Pr(Q_{itj}=1\mid Z_{it},P_t,X_t,h_{0t})
    =
    \sigma_j\{g_0(Z_{it},X_t)+a_{0t}\},
    \qquad j=1,2.
\]
The price equation is constructed so that prices are endogenous: prices depend directly on the unobserved demand shock $h_{0t}$ as well as on the excluded instruments. The estimator observes $(Q_{it},Z_{it})$ within each market and observes $(P_t,X_t,W_t)$ at the market level. It does not observe $a_{0t}$, $h_{0t}$, $g_0$, or $\psi_0$.

The full-estimator simulation uses the correctly specified finite-sieve design. The heterogeneity function, $g_0$, lies in the same finite basis used by the estimator, and the price-side function, $\psi_0$, lies in the same finite basis used by the price block. Thus, any remaining error is sampling and optimization error rather than approximation error. The full estimator is run on the grid
\[
    T\in\{50,100,200\},
    \qquad
    n_t\in\{100,250,500\},
\]
with 50 Monte Carlo replications at each design point. The reported estimator is the feasible estimator: for every candidate $g$, the market intercepts are profiled from the simulated micro data, the micro moments are formed from the profiled residuals, and the price-side moment is evaluated using the generated shocks $\hat a_t(g)-\psi(P_t,X_t)$.

\subsection{Consistency Exercise: Results}

Figure \ref{fig:mc-full-consistency} presents the main simulation. Each panel reports the median error across 50 replications at a given pair $(T,n_t)$. The main pattern is consistent with the theorem: the sup norm errors for $g_0$ and $\psi_0$, and the average errors for $a_{0t}$ and $h_{0t}$, all decline as the amount of market-level and within-market information increases.

\begin{figure}[!t] \centering \includegraphics[scale=0.95]{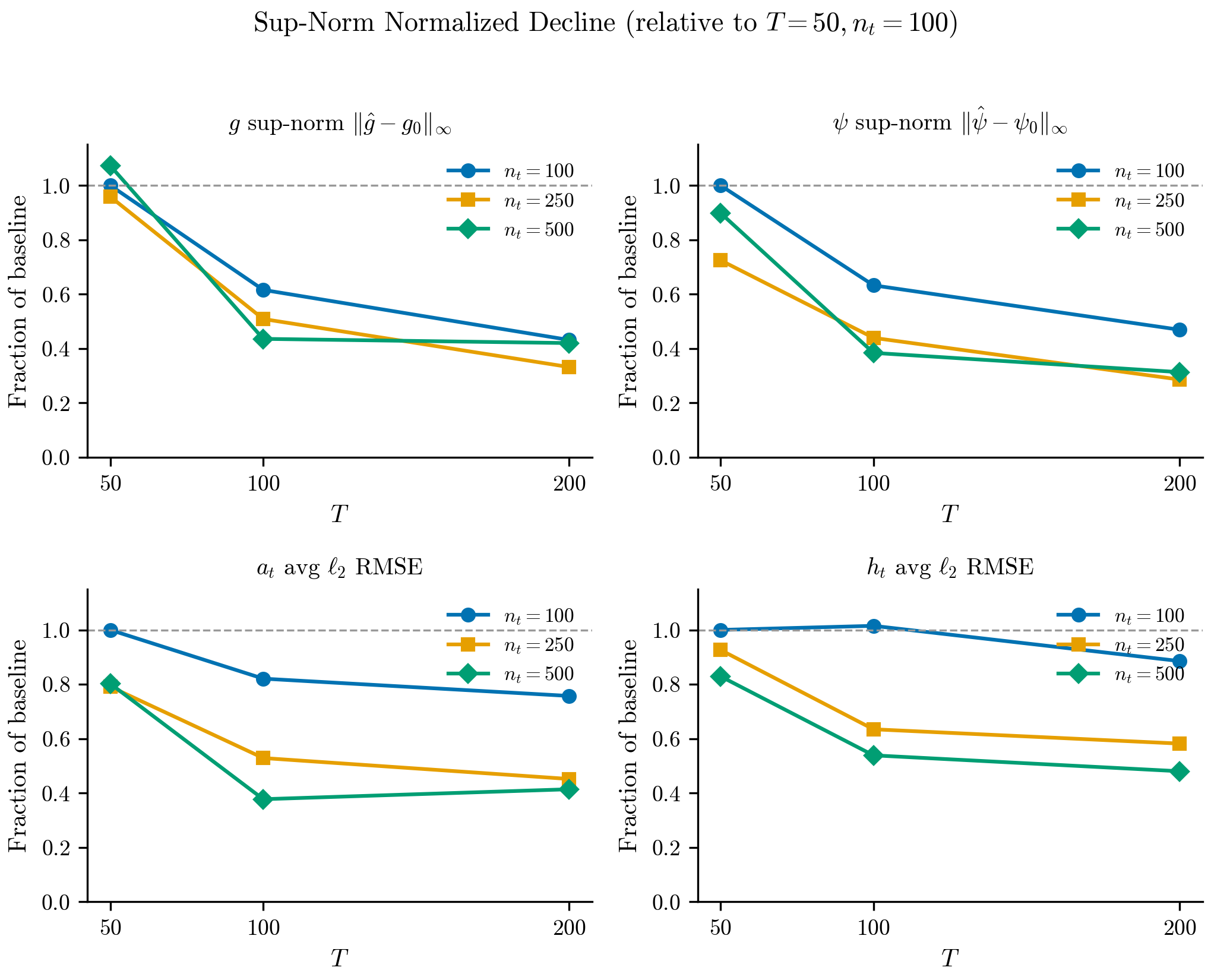} \caption{Normalized Monte Carlo errors for the full profiled estimator.} \label{fig:mc-full-consistency} \vspace{0.5em} \begin{minipage}{0.95\textwidth} \footnotesize \emph{Notes:} The figure reports normalized Monte Carlo errors for the full profiled estimator over the sample-size grid \(T\in\{50,100,200\}\) and \(n_t\in\{100,250,500\}\). Each marker is the median across Monte Carlo replications for the corresponding design. Each error is divided by the value of the same error metric in the baseline design \((T,n_t)=(50,100)\), so the dashed horizontal line at one represents the baseline and values below one indicate improvement. The top panels report grid sup-norm errors for the common structural functions \(g_0\) and \(\psi_0\).
The bottom panels report average market-level \(\ell_2\) errors for the generated objects \(a_{0t}\) and \(h_{0t}\).
Thus, the top panels correspond to the uniform consistency conclusion for the common structural functions, while the bottom panels summarize recovery of the profiled composite intercepts and recovered structural shocks. \end{minipage} \end{figure}

The sup-norm results are deliberately stringent. Unlike an integrated or root-mean-squared functional error, the sup norm records the largest absolute error over products and over the evaluation grid. It is, therefore, sensitive to local estimation error at tail grid points and should not be expected to be as smooth as an average error metric in a finite Monte Carlo design. Nevertheless, the same qualitative pattern emerges. The common structural functions are recovered more accurately in the high-information region of the design, especially as the number of markets increases. This is consistent with the role of cross-market variation in the outer problem: more markets improve the micro moment for $g_0$ and the price-side moment for $\psi_0$. 

The market-level objects behave in the way predicted by the profiling argument. The profiled composite intercept $\hat a_t$ improves as the within-market sample size increases, because each $a_{0t}$ is learned from the consumers inside market $t$. Moving from $(T,n_t)=(50,100)$ to $(200,500)$ reduces the median $a_t$ error from 0.707 to 0.293. The recovered shock $\hat h_t=\hat a_t-\hat\psi(P_t,X_t)$ also improves over the grid, falling from 0.532 at $(50,100)$ to 0.255 at $(200,500)$. This decline reflects both components of the estimator: larger $n_t$ improves the profiled intercepts, while larger $T$ improves the price-side decomposition of $a_{0t}$ into $\psi_0(P_t,X_t)$ and $h_{0t}$. 

The finite-sample pattern is not perfectly monotone cell-by-cell. This is especially true for the sup-norm errors for $g_0$ and $\psi_0$, because a single grid point can determine the reported value. This should not be a cause for major concern. The purpose of the simulation is not to estimate a rate or to claim monotonicity in every finite-sample cell; rather, it is to verify the qualitative consistency logic of the estimator. The relevant diagnostic is that the largest-errors-over-the-grid for the common functions, and the average errors for the generated market-level objects, are smaller in the high-information designs than in the low-information designs. 

Taken together, the simulation supports the main consistency theorem. The inner profile behaves as a well-posed market-by-market share-matching step, the price-side SMD/GMM block recovers the systematic price component using excluded instruments, and the feasible estimator jointly recovers $g_0$, $\psi_0$, $a_{0t}$, and $h_{0t}$ as the amount of within-market and cross-market information increases.

\subsection{Counterfactual Comparison with BLP Benchmarks}

The second exercise studies whether the flexibility of the price-side function matters for economically relevant counterfactuals. I compare the proposed estimator with three standard benchmarks: aggregate Logit-IV, BLP random-coefficients logit, and Micro-BLP. The counterfactual DGP has three inside goods and an outside option. Individual utility is
\[
    U_{ijt}
    =
    g z_i+\psi_{0j}(P_t,X_t)+h_{jt}+\varepsilon_{ijt},
    \qquad j=1,2,3,
    \qquad
    U_{i0t}=\varepsilon_{i0t},
\]
where $\varepsilon_{ijt}$ is Type-I extreme value, $z_i\sim N(0,1)$, and the heterogeneity loading is $g=0.3$. Thus the heterogeneity component is deliberately simple: the scalar term $g z_i$ enters each inside good symmetrically. Equivalently,
\[
    g_0(z_i,X_t)=g z_i\cdot \mathbf 1_J .
\]
This specification creates observed heterogeneity in consumers' inside-good tastes relative to the outside option, but it does not create product-specific heterogeneity across inside goods.

This design is intentionally favorable to the BLP benchmarks. A conventional random-coefficients logit model can represent this source of heterogeneity as a random coefficient on the inside-good constant. Moreover, Micro-BLP is given additional micro information, described below, that is directly informative about this heterogeneity. Hence, the simulation does not maliciously disadvantage BLP by giving the proposed estimator a rich nonparametric heterogeneity function that BLP cannot approximate. Instead, the main source of misspecification for the benchmarks is isolated on the price side: the true function $\psi_0(P_t,X_t)$ is nonlinear in prices, while aggregate Logit-IV, BLP, and Micro-BLP impose the standard linear price term in mean utility. The exercise therefore asks whether flexible recovery of the price-side function matters for counterfactual predictions, even in a setting where the heterogeneity structure is chosen to be relatively favorable to random-coefficients logit.

Prices are endogenous in the DGP. They depend both on excluded cost shifters and on the unobserved demand shocks. The excluded cost shifters are used as instruments. The simulation also draws first and second choices for every consumer. First choices generate the market shares used by all estimators. Second choices are used to construct a micro moment for Micro-BLP: the share of inside-good buyers whose second choice is the outside option. This moment is informative about the inside-versus-outside margin and therefore about the random-coefficient heterogeneity in $g z_i$. In this sense, Micro-BLP is a strong benchmark as it is given micro information that is well aligned with the particular dimension of heterogeneity present in the DGP.

The exercise considers two regimes. In the oracle regime, all estimators are given the true counterfactual demand shocks. This isolates the consequences of functional-form misspecification. In the full regime, each estimator must recover its own demand shocks from the simulated data. The main text reports the full regime because it is the relevant empirical case. Each simulation uses $T=300$ markets, $n_t=500$ consumers per market, and 40 Monte Carlo replications. The sieve order for the proposed estimator is selected by cross-validation and is not chosen using knowledge of the truth.

Figure \ref{fig:mc-counterfactual-full} reports three counterfactual objects in the full regime. Panel (a) reports the error in predicted counterfactual share changes after a price shock to the first inside good. Panel (b) reports the error in diversion ratios. Panel (c) reports the error in welfare changes, measured as the change in expected inclusive value. In all three panels, displayed lines are medians across replications and shaded regions are 10th--90th percentile bands.

\begin{figure}[!t]
    \centering
    \begin{minipage}[t]{0.48\textwidth}
        \centering
        (a) Share-change RMSE\\[-0.1em]
        \includegraphics[width=\textwidth]{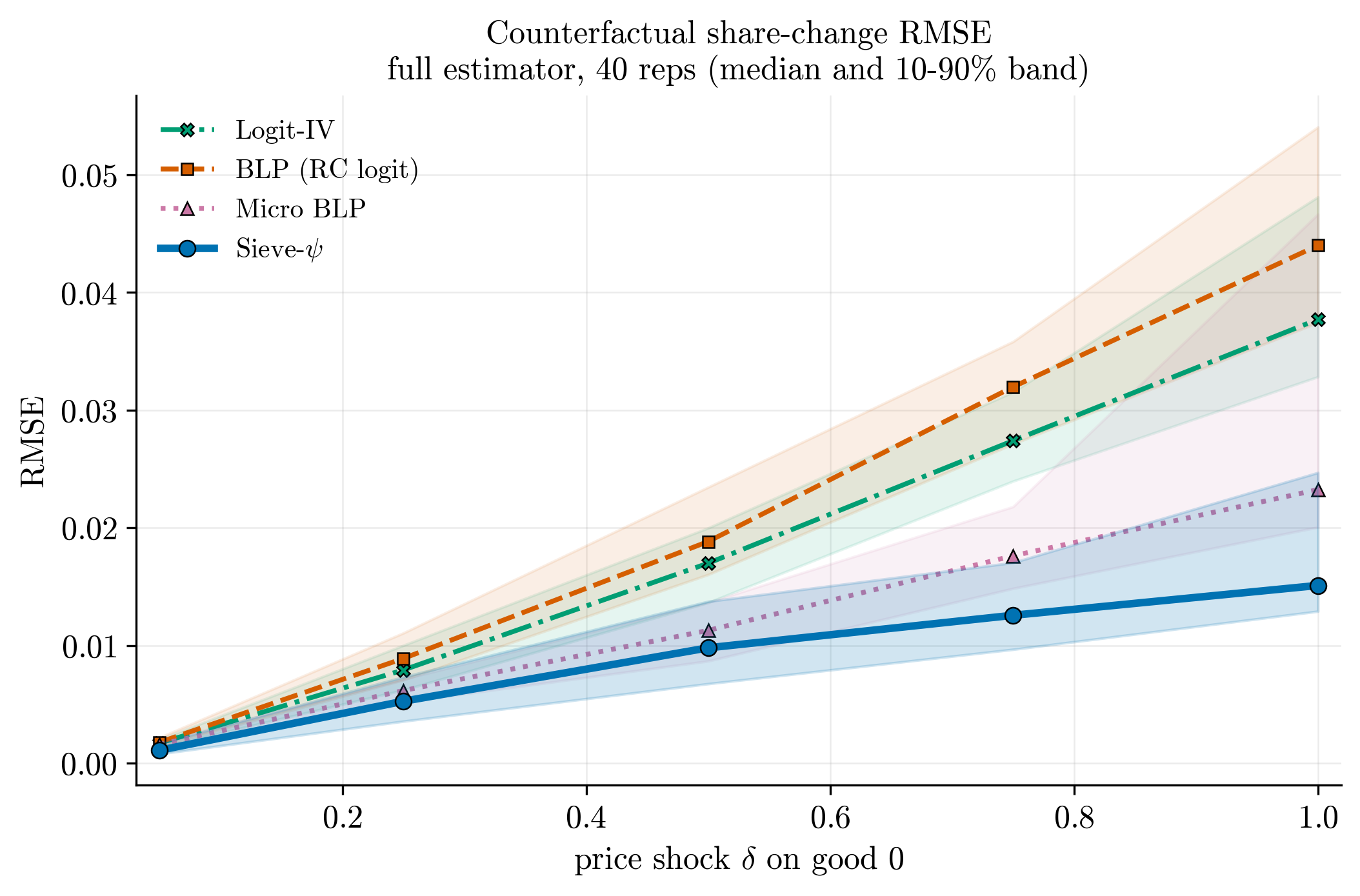}
    \end{minipage}\hfill
    \begin{minipage}[t]{0.48\textwidth}
        \centering
        (b) Diversion-ratio error\\[-0.25em]
        \includegraphics[width=\textwidth]{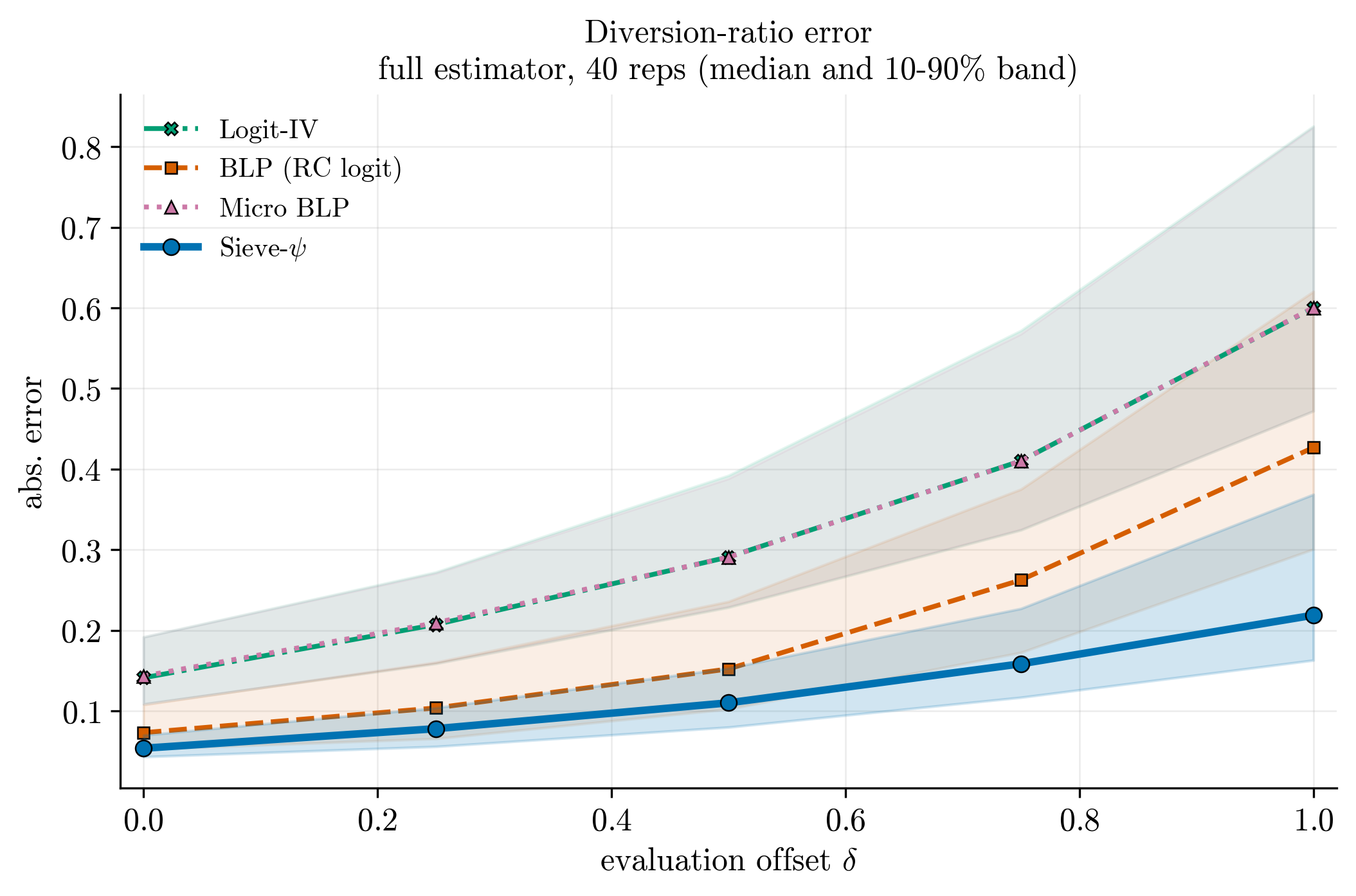}
    \end{minipage}

    \vspace{0.75em}

    \begin{minipage}[t]{0.58\textwidth}
        \centering
        (c) Welfare error\\[-0.25em]
        \includegraphics[width=\textwidth]{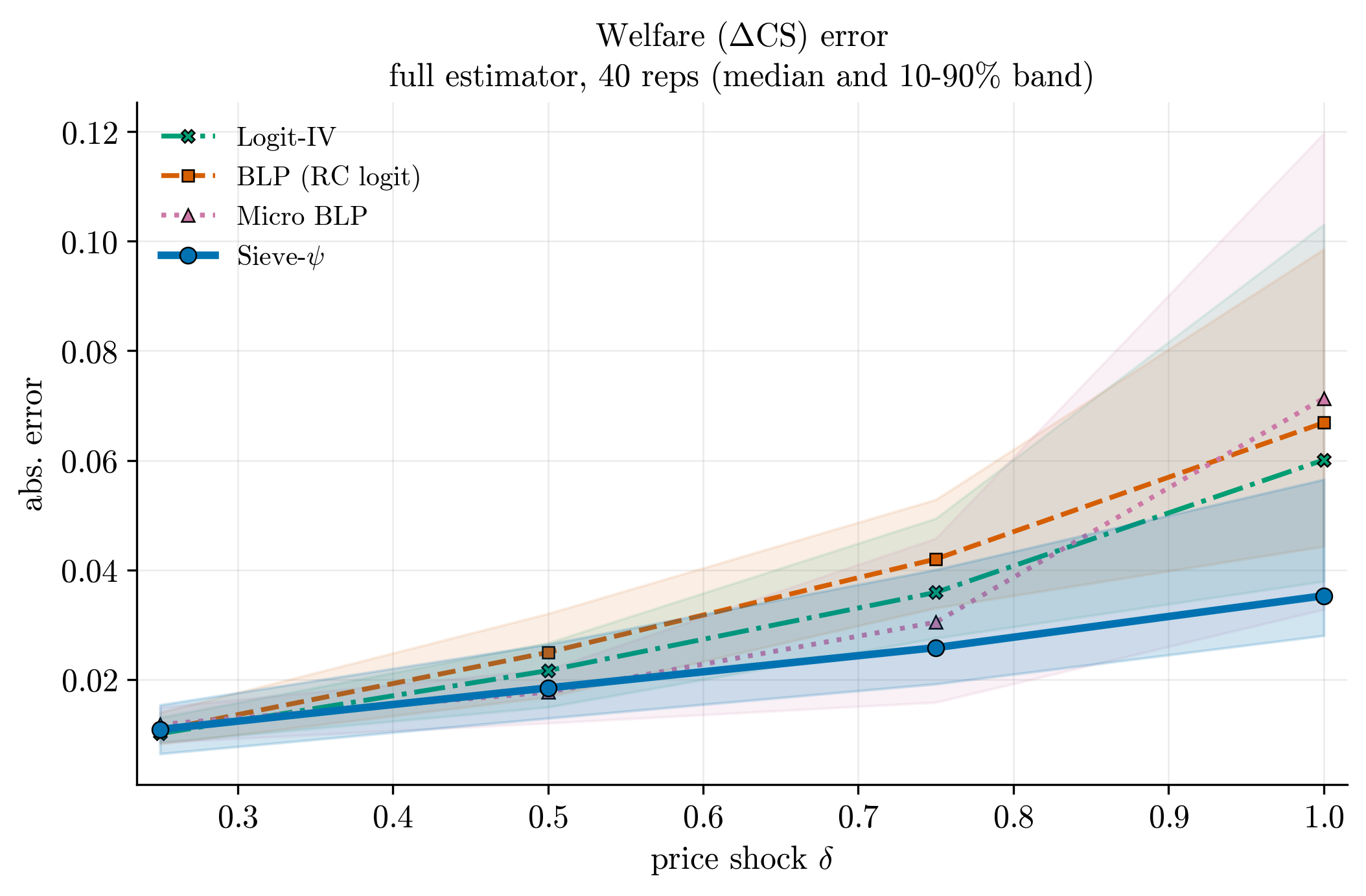}
    \end{minipage}

    \caption{Counterfactual performance against Logit-IV, BLP, and Micro-BLP.}
    \label{fig:mc-counterfactual-full}

    \vspace{0.5em}
    \begin{minipage}{0.95\textwidth}
    \footnotesize
    \emph{Notes:}
    The figure reports full-regime counterfactual errors across 40 Monte Carlo replications. The full regime means that each estimator recovers its own demand shock from the simulated data rather than being given the true shock. Lines report medians and shaded regions report 10--90 percent bands. Panel (a) reports counterfactual share-change RMSE after a price shock to good 0, the first inside good in the simulation. Panel (b) reports absolute diversion-ratio error at different evaluation offsets. Panel (c) reports absolute error in the welfare change $\Delta CS$, measured in utility units. The four estimators are aggregate Logit-IV, BLP random-coefficients logit, Micro-BLP, and the proposed flexible sieve estimator. In panel (b), the Logit-IV and Micro-BLP median curves are nearly coincident over the plotted range, so the Logit-IV curve is largely hidden by the Micro-BLP curve. Exact metric definitions and additional results are given in Appendix \ref{app:counterfactual-details}.
    \end{minipage}
\end{figure}

The flexible estimator performs best on all three reported counterfactual objects. For share changes, the error of the sieve estimator increases much more slowly as the counterfactual price shock becomes larger. This is the object most directly tied to recovery of the demand surface away from observed prices. If the nonlinear price function is approximated by a linear price index, share predictions can deteriorate quickly as the counterfactual moves farther from the observed price distribution. The proposed estimator is designed to avoid this restriction by estimating $\psi_0(P_t,X_t)$ flexibly.

The same pattern appears for diversion ratios. Diversion is a local substitution object, so errors in the slope and curvature of the demand surface are amplified. The BLP and Micro-BLP benchmarks improve on aggregate Logit-IV by allowing consumer heterogeneity on the inside-versus-outside margin. However, that is not the main source of counterfactual difficulty in this DGP. The difficulty is that the price-side component of utility is nonlinear. Because the BLP benchmarks retain a linear price term, their random-coefficients structure cannot fully correct the misspecification in the price surface. The proposed estimator remains below the Logit-IV, BLP, and Micro-BLP benchmarks throughout the plotted range.

The welfare results are smoother, as expected. Welfare changes aggregate over the demand system and are therefore less sensitive than diversion ratios to local errors at a particular price point. Nevertheless, the same ranking appears: the flexible estimator delivers the lowest median welfare error across the counterfactual grid. This suggests that the gains from flexible price-side estimation are not limited to pointwise share predictions. They also matter for welfare calculations, even though welfare is a more aggregate object.

These results should be interpreted narrowly. The simulation is not designed to show that BLP performs poorly in general. On the contrary, the heterogeneity structure is chosen to be relatively favorable to BLP and Micro-BLP. The point is instead more specific: even when the random-coefficients component is easy for BLP-style estimators to represent, counterfactual performance can still suffer if the systematic price-side function is misspecified. The proposed estimator performs well in this exercise because it targets precisely this object. The simulation therefore illustrates the economic value of the estimator developed in this paper; that is, flexible recovery of $\psi_0(P_t,X_t)$ can matter substantially for the counterfactual objects that motivate IO demand estimation, including share changes, diversion ratios, and welfare.

% ---------------------------------------------------------------------
\section{Conclusion}\label{sec:Conclusion}
% ---------------------------------------------------------------------

This paper develops a profiled sieve minimum-distance estimator for a semi-nonparametric differentiated products demand model with micro data. It is motivated by the observation that micro-level choice data contain within-market variation that is absent from market-level demand systems. Consumers in the same market face the same prices, product characteristics, and market-level demand shocks, but differ in observed characteristics. This variation identifies the consumer-heterogeneity component of demand and the composite market intercept. Excluded price instruments then separate the composite intercept into a systematic price-side function and a structural demand shock. The estimator formalizes this logic by combining a within-market profiling step with a cross-market conditional-moment step.

The main result establishes uniform consistency of the common structural functions. The profiled composite intercepts and recovered structural shocks are also consistently recovered. The Monte Carlo evidence supports this logic: the common functions and the generated market-level objects are recovered more accurately as both the number of markets and the number of consumers per market increase.

The paper is intentionally limited in scope. Its purpose is to show that the Berry--Haile micro-data identification argument can be turned into a feasible estimator and that the resulting profiled estimator is uniformly consistent. Several important issues are left open. First, the maintained model uses the multinomial-logit map as a tractable link function. This makes the profile step convex and transparent, but it also imposes the familiar substitution restrictions associated with logit demand. A natural next step is to relax this restriction and move closer to the fully nonparametric demand link in the Berry--Haile framework. Second, the paper studies consistency rather than rates, limiting distributions, or confidence sets. The assumptions are therefore designed to control the profile step and the moving SMD criterion at the level needed for consistency, not to drive an asymptotic distribution result. Third, the price-side block inherits the usual difficulties of nonparametric IV and conditional-moment estimation. Completeness is a strong identification condition, and finite-sample performance will depend on the strength of the instruments, the choice of basis, the amount of regularization, and the growth rates of the sieve dimensions.

Future work should therefore develop the asymptotic distribution theory for the estimator. This requires understanding the joint contribution of the within-market profile error, the growing micro moment, and the price-side SMD projection. In the current estimator, the profiled intercepts are treated in the way needed for consistency, but they are not bias-corrected. This is likely not enough for inference. The profile step generates an additional estimation effect, which may contribute to first-order bias or otherwise complicate the limiting distribution of smooth functionals even when it vanishes asymptotically. A promising direction is to construct a Hahn--Newey-style debiased version of the profiled estimator, following the logic of bias correction for nonlinear models with many nuisance parameters \citep{hahn_jackknife_2004}. Such a correction would aim to remove the leading effect of estimating the market intercepts and could provide the asymptotic linearity needed for rates of convergence, asymptotic normality, and feasible inference. This would connect the present consistency result to the broader sieve-inference literature for semi/nonparametric conditional moment models, including \cite{chen_sieve_2015}.

One major limitation of the method, even if the above advancements are made, is the curse of dimensionality; see Appendix \ref{app:coefficient-counts}. Sieve-based nonparametric estimators require extremely large amounts of data once the covariate space or the number of products begins to grow -- even to relatively modest degrees. Methods that impose shape restrictions or truncate the sieve-approximation are viable, but either run at odds with the spirit of flexible demand estimation or may not be guaranteed to approximate the unknown function. Hence, future work might consider targeting lower dimensional counterfactual functionals directly rather than estimating a globally flexible demand system. 

More broadly, the estimator developed here should be viewed as a first step toward a flexible micro-data approach to demand estimation. The paper shows that within-market consumer heterogeneity and cross-market price instruments can be combined in a single profiled SMD framework, and that this framework consistently recovers the structural functions of interest. The next step would be to turn this consistency result into an asymptotic distribution result and to study the counterfactual objects that motivate flexible demand estimation in industrial organization: elasticities, diversion ratios, pass-through, markups, and welfare. Those objects depend on the same structural functions estimated here. Establishing valid inference for them makes the estimator practical.

\clearpage

% =====================================================================

\appendix

% =====================================================================
\section{Berry--Haile Identification Primitives}
\label{app:BH-assumptions}

This appendix collects the Berry--Haile-style assumptions that motivate the model. The main text uses only their core implications: an index structure, additive separability between consumer heterogeneity and market shocks, invertibility of the demand link, and price-side completeness. The remaining conditions are kept here so that the main text can focus on the estimator.

\begin{assumption}[BH primitive structure]
\label{ass:app-BH-basic}
Latent market-level heterogeneity can be represented by a $J$-vector $\Xi_t$. Conditional on $X_t$, the support of $Z_{it}$ is the same in all markets. The consumer-level covariates have dimension at least $J$, so $d_z\ge J$.
\end{assumption}

\begin{assumption}[BH index and separability]
\label{ass:app-BH-index}
The choice probability vector satisfies
\[
    s(C_{it})=\sigma^{BH}(\gamma(Z_{it},X_t,\Xi_t),P_t,X_t),
\]
where $\gamma(Z_{it},X_t,\Xi_t)\in\R^J$ and
\[
    \gamma(Z_{it},X_t,\Xi_t)=\Gamma_0(Z_{it},X_t)+\Xi_t.
\]
\end{assumption}

\begin{assumption}[BH invertibility of demand]
\label{ass:app-BH-demand-invertible}
For each $(p,x)$ in the support of $(P_t,X_t)$, the map $\gamma\mapsto\sigma^{BH}(\gamma,p,x)$ is injective on the relevant support of the index.
\end{assumption}

\begin{assumption}[BH injectivity of the index]
\label{ass:app-BH-index-injective}
For each $(x,\xi)$ in the support of $(X_t,\Xi_t)$, the map $z\mapsto\gamma(z,x,\xi)$ is injective on the support of $Z_{it}$ conditional on $X_t=x$.
\end{assumption}

\begin{assumption}[BH support and nondegeneracy]
\label{ass:app-BH-support}
For each $x\in\X$, the support of $Z_{it}$ conditional on $X_t=x$ is connected and sufficiently rich. The support of the market shock conditional on relevant values of $(P_t,X_t)$ is also sufficiently rich to identify the conditional demand system.
\end{assumption}

\begin{assumption}[BH smoothness and rank]
\label{ass:app-BH-smooth-rank}
The functions $\Gamma_0(\cdot,x)$ and $\sigma^{BH}(\cdot,p,x)$ are continuously differentiable on the relevant supports. The Jacobians with respect to $z$ and $\gamma$ satisfy the nonsingularity conditions required for the Berry--Haile identification argument.
\end{assumption}

\begin{assumption}[BH price instruments and completeness]
\label{ass:app-BH-IV-completeness}
The excluded price instruments satisfy
\[
    \E[h_{0j}(X_t,\Xi_{t})\mid X_t,W_t]
    =
    \E[h_{0j}(X_t,\Xi_{t})\mid X_t]
    \qquad\text{a.s. for each }j\in\Jcal.
\]
In addition, in the relevant class of functions $\bm{\mathcal F}(P_t,X_t)$ with finite expectation,
\[
    \E[\bm{\mathcal F}(P_t,X_t)\mid X_t,W_t]=0\quad\text{a.s.}
    \quad\Longrightarrow\quad
    \bm{\mathcal F}(P_t,X_t)=0\quad\text{a.s.}
\]
\end{assumption}

\begin{assumption}[BH normalizations]
\label{ass:app-BH-normalizations}
A location normalization is imposed on the additive components. A representative Berry--Haile normalization sets $\E[\Xi_t]=0$ and chooses, for each $x\in\X$, a baseline $z^0(x)$ such that $g(z^0(x),x)=0$ and $\partial_z g(z,x)|_{z=z^0(x)}=I_{J\times J}$ whenever the relevant derivative exists. In the main text, this role is played by Assumption \ref{ass:normalizations-main}.
\end{assumption}

% =====================================================================
\section{Population Restrictions Underlying the Estimator}
\label{app:population-restrictions}

This appendix proves the two population restrictions stated before the estimator is defined. These restrictions are conceptually distinct. The micro restriction comes from the structural choice model conditional on the full choice environment. The price-side restriction comes from the excluded price instruments and the normalization of the structural demand shock.

\begin{proposition}[Population restrictions underlying the estimator, restated]
\label{prop:app-population-restrictions-restated}
Under the maintained demand model, the structural residual
\[
    \varepsilon_{it}^0
    :=
    Q_{it}-\sigma(g_0(Z_{it},X_t)+a_{0t})
\]
satisfies
\[
    \E[\varepsilon_{it}^0\mid C_{it}]=0.
\]
Consequently, for any square-integrable function $\varphi(Z_{it},X_t)$,
\[
    \E[\varphi(Z_{it},X_t)\otimes\varepsilon_{it}^0]=0.
\]
If, in addition, Assumptions \ref{ass:normalizations-main} and \ref{ass:price-id-main} hold, then
\[
    \E[a_{0t}-\psi_0(P_t,X_t)\mid W_t,X_t]=0.
\]
\end{proposition}

\begin{proof}
By the maintained demand model and the definition $a_{0t}=\psi_0(P_t,X_t)+h_{0t}$,
\[
    \E[Q_{it}\mid C_{it}]
    =
    \sigma(g_0(Z_{it},X_t)+a_{0t}).
\]
Subtracting the right-hand side from both sides gives
\[
    \E[\varepsilon_{it}^0\mid C_{it}]=0.
\]
Because $\varphi(Z_{it},X_t)$ is a function only of variables contained in
$C_{it}=(Z_{it},P_t,X_t,\Xi_t)$, it is measurable with respect to the
sigma-algebra generated by $C_{it}$. Hence it can be taken outside the
conditional expectation given $C_{it}$. Therefore, by
the law of iterated expectations,
\[
    \E[\varphi(Z_{it},X_t)\otimes\varepsilon_{it}^0]
    =
    \E\!\left[\varphi(Z_{it},X_t)\otimes\E[\varepsilon_{it}^0\mid C_{it}]\right]
    =0.
\]
This proves the structural residual restriction and its unconditional implication. For the price-side restriction, use $a_{0t}=\psi_0(P_t,X_t)+h_{0t}$ to write
\[
    a_{0t}-\psi_0(P_t,X_t)=h_{0t}.
\]
Assumption \ref{ass:price-id-main} gives $\E[h_{0t}\mid W_t,X_t]=\E[h_{0t}\mid X_t]$, and Assumption \ref{ass:normalizations-main} gives $\E[h_{0t}\mid X_t]=0$. Hence
\[
    \E[a_{0t}-\psi_0(P_t,X_t)\mid W_t,X_t]
    =
    \E[h_{0t}\mid W_t,X_t]
    =0.
\]
\end{proof}

% =====================================================================
\section{Normalizations and Identification Proofs}
\label{app:identification}
\label{app:normalizations}

This appendix supplies the proofs for the identification results in Section \ref{sec:Identification}. The role of the normalizations is to select one representative from an observational equivalence class: without them, $X$-only terms can be moved across the additive components without changing choice probabilities.

\begin{lemma}[An $X$-only component cannot be assigned uniquely without normalization]
\label{lem:app-g-normalization}
Let $c:\X\to\R^J$ be any measurable function. Define
\[
    \tilde g(z,x):=g(z,x)+c(x),
    \qquad
    \tilde\psi(p,x):=\psi(p,x)-c(x).
\]
Then, for all $(z,p,x,h)$,
\[
    \tilde g(z,x)+\tilde\psi(p,x)+h
    =
    g(z,x)+\psi(p,x)+h.
\]
Thus the choice probabilities are unchanged. Therefore the decomposition between $g$ and $\psi$ requires a normalization, such as the conditional mean-zero or baseline normalization stated in Section \ref{sec:Identification}.
\end{lemma}

\begin{proof}
The equality follows immediately by substitution. Since the logit map depends on $(g,\psi,h)$ only through the sum of the indices, the choice probabilities are unchanged.
\end{proof}

\begin{lemma}[Re-centering and uniqueness of the normalized representative]
\label{lem:app-normalization-equivalence}
Fix an equivalence class of decompositions that all generate the same index $g(Z_{it},X_t)+a_t$. For any representative $\bar g$ in this class and any measurable function $d:\X\to\R^J$, define
\[
    g^d(z,x):=\bar g(z,x)-d(x),
    \qquad
    a_t^d:=a_t+d(X_t).
\]
Then
\[
    g^d(Z_{it},X_t)+a_t^d
    =
    \bar g(Z_{it},X_t)+a_t .
\]
Thus an $X$-only component can be moved between $g$ and the market intercept without changing the choice index. Two useful choices of $d$ are
\[
    d_m(x):=\E[\bar g(Z_{it},x)\mid X_t=x],
    \qquad
    d_b(x):=\bar g(z^0(x),x).
\]
The first gives the conditional-mean normalization
\[
    \E[g^{d_m}(Z_{it},X_t)\mid X_t=x]=0,
\]
while the second gives the baseline normalization
\[
    g^{d_b}(z^0(x),x)=0.
\]
Moreover, either normalization selects a unique representative from the equivalence class: if $g$ and $g+c$ both satisfy the same normalization for some measurable $c:\X\to\R^J$, then $c(x)=0$ for all $x$ in the relevant support.
\end{lemma}

\begin{proof}
The index invariance follows by substitution:
\[
    g^d(Z_{it},X_t)+a_t^d
    =
    \bar g(Z_{it},X_t)-d(X_t)+a_t+d(X_t)
    =
    \bar g(Z_{it},X_t)+a_t.
\]
For $d=d_m$,
\[
\begin{aligned}
    \E[g^{d_m}(Z_{it},X_t)\mid X_t=x]
    &=
    \E[\bar g(Z_{it},x)-d_m(x)\mid X_t=x]  \\
    &=
    \E[\bar g(Z_{it},x)\mid X_t=x]
    -
    d_m(x)
    =
    0.
\end{aligned}
\]
For $d=d_b$,
\[
    g^{d_b}(z^0(x),x)
    =
    \bar g(z^0(x),x)-d_b(x)
    =
    0.
\]
Finally, suppose $g$ and $g+c$ both satisfy the conditional-mean normalization. Then
\[
    0
    =
    \E[g(Z_{it},X_t)+c(X_t)\mid X_t=x]
    =
    \E[g(Z_{it},X_t)\mid X_t=x]+c(x)
    =
    c(x).
\]
Similarly, if $g$ and $g+c$ both satisfy the baseline normalization, then
\[
    0
    =
    (g+c)(z^0(x),x)
    =
    g(z^0(x),x)+c(x)
    =
    c(x).
\]
Thus either normalization rules out nonzero $X$-only shifts and therefore selects a unique representative from the equivalence class.
\end{proof}

\begin{lemma}[Log-odds representation]
\label{lem:app-log-odds}
Under the maintained logit specification, for every $j\in\Jcal$,
\[
    \log\frac{s_j(C_{it})}{s_0(C_{it})}
    =
    g_{0j}(Z_{it},X_t)+a_{0jt}.
\]
\end{lemma}

\begin{proof}
By the definition of the multinomial-logit map,
\[
    s_j(C_{it})
    =
    \frac{\exp\{g_{0j}(Z_{it},X_t)+a_{0jt}\}}
    {1+\sum_{k=1}^J\exp\{g_{0k}(Z_{it},X_t)+a_{0kt}\}},
    \qquad
    s_0(C_{it})
    =
    \frac{1}{1+\sum_{k=1}^J\exp\{g_{0k}(Z_{it},X_t)+a_{0kt}\}}.
\]
Taking the ratio and applying the logarithm gives the result.
\end{proof}

\begin{proposition}[Within-market identification, restated]
\label{prop:app-within-market-id-restated}
Suppose the Berry--Haile primitive assumptions stated in Appendix \ref{app:BH-assumptions} hold, and suppose the logit specialization in \eqref{eq:choice-composite-main} is correctly specified. Then, subject to the location normalizations in Assumption \ref{ass:normalizations-main}, $(g_0,a_{0t})$ are identified from the conditional choice probabilities.
\end{proposition}

\begin{proof}
Let $\mathcal Z(x):=\supp(Z_{it}\mid X_t=x)$ denote the conditional support of consumer covariates. By the Berry--Haile primitive conditions, the conditional choice-probability vector $s(C_{it})$ is identified on the relevant support. Since the maintained logit model assigns strictly positive probability to every inside good and to the outside option, the log-odds functions
\[
    \ell_{jt}(z)
    :=
    \log \frac{s_j(z,P_t,X_t,\Xi_t)}{s_0(z,P_t,X_t,\Xi_t)},
    \qquad j=1,\dots,J,
\]
are identified for all $z\in \mathcal Z(X_t)$. By Lemma \ref{lem:app-log-odds}, the true pair satisfies
\begin{equation}
\label{eq:app-id-log-odds-truth-new}
    \ell_{jt}(z)=g_{0j}(z,X_t)+a_{0jt}
    \qquad \text{for all }z\in \mathcal Z(X_t),\ j\le J.
\end{equation}
Now let $(\tilde g,\tilde a_t)$ be any other pair in the normalized parameter space that generates the same conditional choice probabilities. Equality of the identified log odds implies
\begin{equation}
\label{eq:app-id-log-odds-alt-new}
    \tilde g_j(z,X_t)+\tilde a_{jt}
    =g_{0j}(z,X_t)+a_{0jt}
    \qquad \text{for all }z\in \mathcal Z(X_t),\ j\le J.
\end{equation}
Hence
\begin{equation}
\label{eq:app-id-delta-new}
    \tilde g_j(z,X_t)-g_{0j}(z,X_t)
    =a_{0jt}-\tilde a_{jt},
\end{equation}
and the right-hand side is constant in $z$. Therefore, for each $x$ in the support of $X_t$, there is a vector $c(x)\in\R^J$ such that
\[
    \tilde g(z,x)=g_0(z,x)+c(x)
    \qquad \text{for all }z\in\mathcal Z(x).
\]
From here, Lemma \ref{lem:app-normalization-equivalence} gives us that $c(x) = 0$ under either the conditional-mean or baseline normalization. Thus $\tilde g=g_0$ on the relevant support. Substituting this equality back into \eqref{eq:app-id-log-odds-alt-new} yields $\tilde a_{jt}=a_{0jt}$ for every $j$ and $t$. Hence no distinct normalized pair can generate the same conditional choice probabilities, so $(g_0,a_{0t})$ are identified.
\end{proof}

\begin{proposition}[Identification of $\psi_0$ from the composite intercept]
\label{prop:app-psi-id}
Suppose $a_{0t}$ is identified and Assumptions \ref{ass:normalizations-main} and \ref{ass:price-id-main} holds. Then $\psi_0(P_t,X_t)$ is identified almost surely. Consequently $h_{0t}=a_{0t}-\psi_0(P_t,X_t)$ is identified almost surely.
\end{proposition}

\begin{proof}
At the truth, $a_{0t}=\psi_0(P_t,X_t)+h_{0t}$. By Assumptions \ref{ass:normalizations-main} and \ref{ass:price-id-main}, $\E[h_{0t}\mid W_t,X_t]=0$. Therefore
\[
    \E[a_{0t}-\psi_0(P_t,X_t)\mid W_t,X_t]=0.
\]
Let $\psi$ be any other candidate satisfying the same conditional moment. Subtracting the two conditional moments gives
\[
    \E[\psi(P_t,X_t)-\psi_0(P_t,X_t)\mid W_t,X_t]=0.
\]
By completeness, $\psi(P_t,X_t)=\psi_0(P_t,X_t)$ almost surely. Since $a_{0t}$ is identified, $h_{0t}=a_{0t}-\psi_0(P_t,X_t)$ is identified as well.
\end{proof}

\begin{theorem}[Identification of the normalized structural functions, restated]
\label{thm:app-identification-restated}
Suppose the maintained demand model \eqref{eq:maintained-choice-model} holds. Suppose further that the conditions of Proposition \ref{prop:within-market-id-main} and Assumption \ref{ass:price-id-main} hold. Then the pair $(g_0,\psi_0)$ is identified in the normalized parameter space. The corresponding market shocks are identified by $h_{0t}=a_{0t}-\psi_0(P_t,X_t)$.
\end{theorem}

\begin{proof}
Proposition \ref{prop:app-within-market-id-restated} identifies $g_0$ and $a_{0t}$. Proposition \ref{prop:app-psi-id} identifies $\psi_0(P_t,X_t)$ and then $h_{0t}=a_{0t}-\psi_0(P_t,X_t)$. The normalizations in Assumption \ref{ass:normalizations-main} select a unique representative of the observational equivalence class. Hence $(g_0,\psi_0)$ is identified in the normalized parameter space.
\end{proof}

% =====================================================================
\section{Inner Profiling Details}
\label{app:inner-profiling}

This appendix gives primitive conditions under which the profile map $g\mapsto a_t^*(g)$ is well defined, the empirical profile $\hat a_t(g)$ is uniformly consistent, and the true profile satisfies $a_t^*(g_0)=a_{0t}$. Let
\[
    \Mcal_t:=(P_t,X_t,W_t,\Xi_t).
\]
For a candidate $g$, define the population inner criterion
\begin{equation}
\label{eq:inner-potential-pop-app}
    L_t(g,a)
    :=
    \E\!\left[
    \log\!\left(1+\sum_{j=1}^J\exp\{g_j(Z_{it},X_t)+a_j\}\right)
    \Bigg | \  \Mcal_t
    \right]-s_t'a,
\end{equation}
where $s_t:=\E[Q_{it}\mid \Mcal_t]$. Let $a_t^*(g)\in\argmin_{a\in\Acal}L_t(g,a)$. Define
\[
    \G_{K_g}^0:=\G_{K_g}\cup\{g_0\}.
\]
The inclusion of $g_0$ is needed only for population statements such as
the true-profile identity and the Lipschitz comparison
$a_t^*(\hat g)-a_t^*(g_0)$. The estimator itself still minimizes over
the finite-dimensional sieve $\G_{K_g}$.

\begin{assumption}[Primitive profile conditions]
\label{ass:app-profile-primitive}
The following conditions hold.
\begin{enumerate}[label=(\roman*),leftmargin=2em]
    \item $\Acal\subset\R^J$ is compact and convex, and $a_t^*(g)$ is
    unique and interior for all $g\in\G_{K_g}^0$ and all $t\le T$.

    \item There exists $B_u<\infty$ such that
    $\abs{g_j(Z_{it},X_t)+a_j}\le B_u$ for all
    $g\in\G_{K_g}^0$, $a\in\Acal$, $j\le J$, and $t\le T$.

    \item The empirical inner criterion satisfies
    \[
        \sup_{g\in\G_{K_g}}\sup_{a\in\Acal}\max_{t\le T}
        \abs{\hat L_t(g,a)-L_t(g,a)}=o_p(1).
    \]

    \item There exists $C<\infty$ such that, uniformly in $t\le T$ and
    $a\in\Acal$,
    \[
        \norm{\nabla_a L_t(g,a)-\nabla_a L_t(\tilde g,a)}
        \le C\norm{g-\tilde g}_\infty
        \qquad\text{for all }g,\tilde g\in\G_{K_g}^0.
    \]
    
    \item Let $\varepsilon_{it}^0
        :=
        Q_{it}
        -
        \sigma\!\left(g_0(Z_{it},X_t)+a_{0t}\right).$
    The structural choice residual is mean zero after conditioning on the
    market environment:
    \begin{equation}
    \label{eq:app-market-share-compatibility}
        \E[\varepsilon_{it}^0\mid\Mcal_t]=0
        \qquad\text{a.s.}
    \end{equation}
\end{enumerate}
\end{assumption}

Condition (i) makes the population profile a function rather than a set-valued correspondence. It can be justified from a more primitive assumption that the population market shares are not so extreme that the required intercept is infinite. Condition (ii) keeps all logit probabilities bounded away from zero and one, which gives uniform curvature of the inner problem. This can be justified by the fact that the intercept set $\mathcal{A}$ is assumed compact and that the sieve space $\mathcal{G}_{K_g}$ is defined as a bounded H\"older ball. Condition (iii) is the within-market uniform law of large numbers; it is justified by conditional independence within markets and standard empirical-process arguments for finite-dimensional sieve classes whose dimensions grow slowly relative to $n_{\min}$; see \cite{andrews_chapter_1994}, \cite{van_der_vaart_weak_1996}, \cite{newey_convergence_1997}, and \cite{chen_large_2007}. Condition (iv) is a smoothness condition on the profile map and follows from the bounded derivative of the logit map and the sup-norm control of $g$. The population conditions are imposed on $\G_{K_g}^0$ so that the same profile map controls both sieve candidates and the true function $g_0$. Condition (v) is a market-share compatibility condition between the
structural choice probabilities and the population profile problem. It
requires the structural probabilities, averaged over within-market consumer
heterogeneity, to reproduce the conditional market shares. The condition is
implied by the familiar pointwise exclusion
\[
    \E[Q_{it}\mid C_{it},W_t]
    =
    \E[Q_{it}\mid C_{it}],
\]
which holds, for example, when \(W_t\) is an excluded cost shifter or some other standard price instrument. The stated condition is weaker because it imposes this equality
only after averaging over \(Z_{it}\).

\begin{lemma}[Gradient and Hessian of the inner objective]
\label{lem:app-inner-gradient}
For every $g$ and $a$,
\[
    \nabla_a\hat L_t(g,a)
    =
    \frac{1}{n_t}\sum_{i=1}^{n_t}\sigma(g(Z_{it},X_t)+a)-\hat s_t,
\]
and
\[
    \nabla_a L_t(g,a)
    =
    \E[\sigma(g(Z_{it},X_t)+a)\mid \Mcal_t]-s_t.
\]
Moreover,
\[
    \nabla_{aa}L_t(g,a)
    =
    \E[\Lambda(g(Z_{it},X_t)+a)\mid \Mcal_t],
    \qquad
    \Lambda(u):=\diag(\sigma(u))-\sigma(u)\sigma(u)'.
\]
\end{lemma}

\begin{proof}
Differentiate the log-sum-exp term in \eqref{eq:inner-potential-main}. Its derivative with respect to $a_j$ is $\sigma_j(g(Z_{it},X_t)+a)$. The above gradient formulas follow immediately. Differentiating the logit probability vector gives the multinomial-logit Jacobian $\Lambda(u)$, which yields the Hessian after taking the conditional expectation.
\end{proof}

\begin{lemma}[Uniform local separation]
\label{lem:app-inner-separation}
Suppose Assumption \ref{ass:app-profile-primitive} holds. Then there exists
$c_\Lambda>0$ such that, for every $g\in\G_{K_g}^0$, every $t\le T$, and
every $a\in\Acal$,
\[
    L_t(g,a)-L_t(g,a_t^*(g))
    \ge
    \frac{c_\Lambda}{2}\norm{a-a_t^*(g)}^2 .
\]
Consequently, for every $\epsilon>0$,
\[
    \inf_{g\in\G_{K_g}^0}\min_{t\le T}
    \inf_{a\in\Acal:\norm{a-a_t^*(g)}\ge\epsilon}
    \{L_t(g,a)-L_t(g,a_t^*(g))\}
    \ge
    \eta(\epsilon),
\]
where $\eta(\epsilon):=(c_\Lambda/2)\epsilon^2>0$.
\end{lemma}

\begin{proof}
By Assumption \ref{ass:app-profile-primitive}(ii), the inside-good logit probabilities satisfy
\[
    \sigma_j(u)
    =
    \frac{\exp(u_j)}{1+\sum_{k=1}^J\exp(u_k)}
    \ge
    \frac{e^{-B_u}}{1+J e^{B_u}}
    =:\kappa,
    \qquad j=1,\dots,J.
\]
Similarly, the outside-good probability satisfies $
    \sigma_0(u)
    =
    \frac{1}{1+\sum_{k=1}^J\exp(u_k)}
    \ge
    \frac{1}{1+J e^{B_u}}
    \ge
    \kappa$.
Thus all inside and outside probabilities are uniformly bounded below by
$\kappa>0$ on the admissible index set.

We next show that the logit Hessian is uniformly positive definite.
For any $v\in\R^J$,
\[
    v'\Lambda(u)v
    =
    \sum_{j=1}^J \sigma_j(u)v_j^2
    -
    \left(\sum_{j=1}^J \sigma_j(u)v_j\right)^2
\]
where $\Lambda(u):=\diag(\sigma(u))-\sigma(u)\sigma(u)'$. Then, write $p_j=\sigma_j(u)$ for $j=1,\dots,J$ and $p_0=\sigma_0(u)$.
Let $Y$ be the random variable that takes value $0$ with probability
$p_0$ and value $v_j$ with probability $p_j$ for $j=1,\dots,J$. Then
\[
    v'\Lambda(u)v=\Var(Y).
\]
Using the pairwise-variance identity,
\[
    \Var(Y) := \sum_{r=0}^J p_r(y_r - \E[Y])^2
    =
    \sum_{0\le r<s\le J}p_rp_s(y_r-y_s)^2,
\]
where $y_0=0$ and $y_j=v_j$, we obtain
\[
    v'\Lambda(u)v
    \ge
    \sum_{j=1}^J p_0p_j(v_j-0)^2
    \ge
    \kappa^2\sum_{j=1}^Jv_j^2
    =
    \kappa^2\norm{v}^2 .
\]
Therefore the smallest eigenvalue of $\Lambda(u)$ is bounded below
uniformly by  $c_\Lambda:=\kappa^2>0$.

By Lemma \ref{lem:app-inner-gradient},
\[
    \nabla_{aa}L_t(g,a)
    =
    \E[\Lambda(g(Z_{it},X_t)+a)\mid \Mcal_t].
\]
Since each matrix inside the conditional expectation has smallest
eigenvalue at least $c_\Lambda$, the conditional expectation also has
smallest eigenvalue at least $c_\Lambda$. Hence, for every
$v\in\R^J$,
\[
    v'\nabla_{aa}L_t(g,a)v\ge c_\Lambda\norm{v}^2
\]
uniformly over $g\in\G_{K_g}^0$, $a\in\Acal$, and $t\le T$. So, $a \mapsto L_t(g,a)$ is uniformly strongly convex.

Fix $g$ and $t$, and write $a^*=a_t^*(g)$. By Assumption
\ref{ass:app-profile-primitive}(i), $a^*$ is an interior minimizer, so
\[
    \nabla_a L_t(g,a^*)=0.
\]
For any $a\in\Acal$, Taylor's theorem with integral remainder gives
\[
\begin{aligned}
    L_t(g,a)-L_t(g,a^*)
    &=
    \int_0^1(1-s)
    (a-a^*)'
    \nabla_{aa}L_t(g,a^*+s(a-a^*))
    (a-a^*)\,ds \\
    &\ge
    \int_0^1(1-s)c_\Lambda\norm{a-a^*}^2\,ds \\
    &=
    \frac{c_\Lambda}{2}\norm{a-a^*}^2 .
\end{aligned}
\]
This bound is uniform in $g$ and $t$. Therefore, if
$\norm{a-a_t^*(g)}\ge\epsilon$, then
\[
    L_t(g,a)-L_t(g,a_t^*(g))
    \ge
    \frac{c_\Lambda}{2}\epsilon^2 .
\]
Taking $\eta(\epsilon):=\frac{c_\Lambda}{2}\epsilon^2$
proves the separation claim.
\end{proof}

\begin{proposition}[Uniform consistency of the profiled intercept]
\label{prop:app-profile-consistency-restated}
Suppose Assumption \ref{ass:app-profile-primitive} holds. Then
\[
    \sup_{g\in\G_{K_g}}\max_{t\le T}
    \norm{\hat a_t(g)-a_t^*(g)}=o_p(1).
\]
\end{proposition}

\begin{proof}
Fix $\epsilon>0$ and let $\eta(\epsilon)$ be the separation constant from
Lemma \ref{lem:app-inner-separation}. Define the event
\[
    \mathcal E_T(\epsilon)
    :=
    \left\{
    \sup_{g\in\G_{K_g}}\sup_{a\in\Acal}\max_{t\le T}
    \abs{\hat L_t(g,a)-L_t(g,a)}
    \le
    \eta(\epsilon)/3
    \right\}.
\]
By Assumption \ref{ass:app-profile-primitive}(iii), $\Prob(\mathcal E_T(\epsilon))\to1.$ We show that, on $\mathcal E_T(\epsilon)$,
\[
    \sup_{g\in\G_{K_g}}\max_{t\le T}
    \norm{\hat a_t(g)-a_t^*(g)}
    <\epsilon .
\]
Suppose not. Then for some $g\in\G_{K_g}$ and some $t\le T$,
\[
    \norm{\hat a_t(g)-a_t^*(g)}\ge\epsilon.
\]
By the separation lemma,
\[
    L_t(g,\hat a_t(g))
    \ge
    L_t(g,a_t^*(g))+\eta(\epsilon).
\]
On the event $\mathcal E_T(\epsilon)$, uniform convergence of the inner
criteria implies
\[
    \hat L_t(g,\hat a_t(g))
    \ge
    L_t(g,\hat a_t(g))-\eta(\epsilon)/3
    \ge
    L_t(g,a_t^*(g))+2\eta(\epsilon)/3.
\]
Again using $\mathcal E_T(\epsilon)$,
\[
    \hat L_t(g,a_t^*(g))
    \le
    L_t(g,a_t^*(g))+\eta(\epsilon)/3.
\]
Combining the last two lines gives
\[
    \hat L_t(g,\hat a_t(g))
    >
    \hat L_t(g,a_t^*(g)).
\]
This contradicts the definition of $\hat a_t(g)$ as a minimizer of
$a\mapsto \hat L_t(g,a)$ over $\Acal$. Hence, on
$\mathcal E_T(\epsilon)$,
\[
    \sup_{g\in\G_{K_g}}\max_{t\le T}
    \norm{\hat a_t(g)-a_t^*(g)}
    <\epsilon .
\]
Since $\Prob(\mathcal E_T(\epsilon))\to1$ and $\epsilon>0$ was arbitrary,
the displayed supremum is $o_p(1)$.
\end{proof}

\begin{proposition}[The true intercept solves the true profile problem]
\label{prop:app-true-profile}
At the truth, $a_t^*(g_0)=a_{0t}$ for every $t$.
\end{proposition}

\begin{proof}
By \eqref{eq:app-market-share-compatibility},
\[
    \E[Q_{it}\mid \Mcal_t]
    =
    \E[\sigma(g_0(Z_{it},X_t)+a_{0t})\mid \Mcal_t].
\]
Therefore $\nabla_a L_t(g_0,a_{0t})=0$ by Lemma \ref{lem:app-inner-gradient}. Since the population minimizer is unique and interior, $a_t^*(g_0)=a_{0t}$.
\end{proof}

\begin{proposition}[Validity of the profiled micro moment, restated]
\label{prop:app-micro-valid-restated}
Under the maintained demand model and the population profile definition,
\[
    m_{g,K_g}(g_0)=0
    \qquad\text{for every }K_g.
\]
\end{proposition}

\begin{proof}
Let
\[
    B_{it,K_g}:=b^{K_g}(Z_{it},X_t),
    \qquad
    \varepsilon_{it}^0:=Q_{it}-\sigma(g_0(Z_{it},X_t)+a_{0t}).
\]
By Proposition \ref{prop:app-true-profile}, $a_t^*(g_0)=a_{0t}$ for every market $t$. Substituting this identity into the population profiled residual gives
\[
    \rho_{it}^*(g_0)
    =
    Q_{it}-\sigma(g_0(Z_{it},X_t)+a_t^*(g_0))
    =
    Q_{it}-\sigma(g_0(Z_{it},X_t)+a_{0t})
    =
    \varepsilon_{it}^0 .
\]
Therefore
\begin{equation}
\label{eq:app-mg-truth-proof}
    m_{g,K_g}(g_0)
    =
    \E[B_{it,K_g}\otimes \varepsilon_{it}^0].
\end{equation}
The vector $B_{it,K_g}$ is measurable with respect to the sigma-algebra generated by $C_{it}$ because it is a function only of $(Z_{it},X_t)$. Moreover, $a_{0t}=\psi_0(P_t,X_t)+h_{0t}$ is also measurable with respect to this sigma-algebra, since $h_{0t}=h_0(X_t,\Xi_t)$ and $C_{it}=(Z_{it},P_t,X_t,\Xi_t)$. By the maintained demand model,
\[
    \E[Q_{it}\mid C_{it}]
    =
    \sigma(g_0(Z_{it},X_t)+a_{0t}),
\]
so
\[
    \E[\varepsilon_{it}^0\mid C_{it}]=0.
\]
Assuming the basis functions are square-integrable\footnote{If the basis functions are polynomials, splines, Fourier terms, etc. on compact supports, then each basis function is bounded for each fixed $K_g$, hence square-integrable automatically.}, each component of
$B_{it,K_g}\otimes \varepsilon_{it}^0$ is integrable. Applying iterated expectations componentwise to \eqref{eq:app-mg-truth-proof} yields
\[
    m_{g,K_g}(g_0)
    =
    \E\!\left[
    B_{it,K_g}\otimes\E[\varepsilon_{it}^0\mid C_{it}]
    \right]
    =
    0.
\]
Thus the profiled micro moment is centered at the truth for each $K_g$.
\end{proof}

\begin{proposition}[Lipschitz continuity of the population profile]
\label{prop:app-profile-lipschitz}
Under Assumption \ref{ass:app-profile-primitive}, there exists
$C_a<\infty$ such that
\[
    \max_{t\le T}\norm{a_t^*(g)-a_t^*(\tilde g)}
    \le C_a\norm{g-\tilde g}_{\infty}
\]
for all $g,\tilde g\in\G_{K_g}^0$.
\end{proposition}

\begin{proof}
Fix $g,\tilde g\in\G_{K_g}^0$ and $t\le T$, and write
\[
    a:=a_t^*(g),
    \qquad
    \tilde a:=a_t^*(\tilde g).
\]
By Assumption \ref{ass:app-profile-primitive}(i), both profile values
are interior minimizers. Hence
\[
    \nabla_aL_t(g,a)=0,
    \qquad
    \nabla_aL_t(\tilde g,\tilde a)=0.
\]
Adding and subtracting $\nabla_aL_t(g,\tilde a)$ gives
\[
    \nabla_aL_t(g,a)-\nabla_aL_t(g,\tilde a)
    =
    \nabla_aL_t(\tilde g,\tilde a)-\nabla_aL_t(g,\tilde a).
\]
Take the inner product of both sides with $a-\tilde a$. By the uniform
strong convexity established in Lemma \ref{lem:app-inner-separation},
\[
    (a-\tilde a)'
    \{\nabla_aL_t(g,a)-\nabla_aL_t(g,\tilde a)\}
    \ge
    c_\Lambda\norm{a-\tilde a}^2.
\]
For the right-hand side, Cauchy--Schwarz and Assumption
\ref{ass:app-profile-primitive}(iv) imply
\[
\begin{aligned}
    \left|
    (a-\tilde a)'
    \{\nabla_aL_t(\tilde g,\tilde a)-\nabla_aL_t(g,\tilde a)\}
    \right|
    &\le
    \norm{a-\tilde a}
    \norm{\nabla_aL_t(\tilde g,\tilde a)-\nabla_aL_t(g,\tilde a)} \\
    &\le
    C\norm{a-\tilde a}\norm{g-\tilde g}_\infty .
\end{aligned}
\]
Combining the preceding two displays yields
\[
    c_\Lambda\norm{a-\tilde a}^2
    \le
    C\norm{a-\tilde a}\norm{g-\tilde g}_\infty .
\]
If $a=\tilde a$, the desired inequality is immediate. Otherwise, divide
by $\norm{a-\tilde a}$ to obtain
\[
    \norm{a-\tilde a}
    \le
    \frac{C}{c_\Lambda}\norm{g-\tilde g}_\infty.
\]
The constants do not depend on $t$, $g$, or $\tilde g$. Taking the maximum
over $t\le T$ proves the result with $C_a:=C/c_\Lambda$.
\end{proof}

% =====================================================================
\section{Sieve, SMD, and Uniform Convergence Details}
\label{app:consistency-details}
\label{app:SMD-details}

This appendix gives primitive sufficient conditions for Assumptions \ref{ass:sieve-approx-main}, \ref{ass:population-id-main}, and \ref{ass:criterion-UC-main}. These are standard in sieve extremum and Ai--Chen SMD arguments, except for two features: the micro block uses the profiled intercept, and the price-side residual contains the profiled object $a_t^*(g)$. The notation below keeps track of the moving criterion induced by the growing micro basis.

\subsection{Function spaces and compactness}
\label{app:function-spaces}

This subsection makes explicit the function-space notation introduced in \eqref{eq:theta0-main}--\eqref{eq:function-spaces-main}. Let
\[
    \mathcal D_g:=\supp(Z_{it},X_t),
    \qquad
    \mathcal D_\psi:=\supp(P_t,X_t),
\]
and suppose these supports are compact. Compact support is not essential, but it keeps the consistency argument focused on the main issue: combining within-market profiling with a price-side conditional moment.

For $s>0$, write $s=m+\alpha$, where $m:=\lfloor s\rfloor$ when $s$ is not an integer and $\alpha:=s-m\in(0,1]$; when $s$ is an integer, take $m=s-1$ and $\alpha=1$. For a compact set $D\subset\R^d$ and a scalar function $f:D\to\R$, use multi-index notation $\ell=(\ell_1,\dots,\ell_d)$, $|\ell|=\sum_{r=1}^d \ell_r$, and $D^\ell f$ for the corresponding partial derivative. The $C^s$ H\"older norm is
\[
    \norm{f}_{C^s(D)}
    :=
    \max_{|\ell|\le m}\sup_{u\in D}|D^\ell f(u)|
    +
    \max_{|\ell|=m}\sup_{u\ne v,\ u,v\in D}
    \frac{|D^\ell f(u)-D^\ell f(v)|}{\norm{u-v}^\alpha}.
\]
The corresponding bounded H\"older ball of radius $B<\infty$ is
\[
    C^s_B(D):=\{f:D\to\R:\norm{f}_{C^s(D)}\le B\}.
\]
Thus a H\"older ball is a class of functions whose derivatives up to order $m$ are uniformly bounded and whose order-$m$ derivatives are uniformly H\"older-continuous with exponent $\alpha$. For a vector-valued function $f=(f_1,\dots,f_J)$, write $f\in C^s_B(D)^J$ when each component belongs to $C^s_B(D)$. For vector-valued functions, the H\"older restriction is imposed
componentwise. Unless otherwise stated, $\norm{\cdot}$ denotes the
Euclidean norm on finite-dimensional vectors, while
$\norm{\cdot}_\infty$ denotes the function sup norm defined above.

A primitive version of the full parameter space is therefore
\[
    \G
    =
    \left\{g:\mathcal D_g\to\R^J:
    g_j\in C^{s_g}_{B_g}(\mathcal D_g),\ j\le J, \ 
    \mathcal N_g(g)(x)=0\ \text{for all }x\in\X
    \right\},
\]
where the normalization operator is either
\[
    \mathcal N_g(g)(x):=\E[g(Z_{it},x)\mid X_t=x] = \int g(z,x)dF_{Z|X=x}(z)
    \qquad\text{or}\qquad
    \mathcal N_g(g)(x):=g(z^0(x),x),
\]
and
\[
    \bfPsi
    =
    \left\{\psi:\mathcal D_\psi\to\R^J:
    \psi_j\in C^{s_\psi}_{B_\psi}(\mathcal D_\psi),\ j\le J
    \right\}.
\]
The location of $\psi$ relative to $h_t$ is not fixed by a separate centering restriction on $\psi$; it is fixed by the shock normalization $\E[h_{0t}\mid X_t]=0$ together with $a_{0t}=\psi_0(P_t,X_t)+h_{0t}$. If one tried to replace $\psi_0$ by $\psi_0+c(X_t)$, the implied shock would be $h_{0t}-c(X_t)$, whose conditional mean given $X_t$ is $-c(X_t)$. The shock normalization therefore forces $c(X_t)=0$.

The following compactness argument is standard for compact infinite-dimensional parameter spaces; see \cite{freyberger_practical_2019}. 
Bounded H\"older balls on compact domains are uniformly bounded and equicontinuous, and hence are relatively compact in the sup norm by the Arzel\`a--Ascoli theorem. The normalization restrictions above are closed under uniform convergence: if $g_n\to g$ uniformly and each $g_n$ satisfies the normalization, then $g$ satisfies the same normalization. For the baseline normalization this is immediate from \[ |g(z^0(x),x)| = |g(z^0(x),x) - g_n(z^0(x),x)| \le \|g-g_n\|_\infty \to 0. \] For the conditional-mean normalization, \[ \left\| \E[g(Z_{it},x)\mid X_t=x] \right\|  \le \E[ \| g(Z_{it}, x) - g_n(Z_{it},x) \| \mid X_t=x]\le \|g_n-g\|_\infty \to 0, \] so the normalization passes to the uniform limit. Therefore the normalized parameter spaces are closed subsets of relatively compact H\"older balls, and hence are compact under the sup norm. The product space $\ThetaSpace=\G\times\bfPsi$ is compact under the product norm in \eqref{eq:theta-norm-main}.

\begin{assumption}[Primitive compact function classes]
\label{ass:app-compact-primitive}
The supports $\mathcal D_g=\supp(Z_{it},X_t)$ and $\mathcal D_\psi=\supp(P_t,X_t)$ are compact. There exist smoothness levels $s_g,s_\psi>0$ and finite constants $B_g,B_\psi<\infty$ such that $\G$ and $\bfPsi$ are closed, normalized subsets of bounded H\"older balls:
\[
    \G\subseteq\{g:g_j\in C^{s_g}_{B_g}(\mathcal D_g),\ j\le J\},
    \qquad
    \bfPsi\subseteq\{\psi:\psi_j\in C^{s_\psi}_{B_\psi}(\mathcal D_\psi),\ j\le J\}.
\]
The normalizations are imposed as closed restrictions under uniform convergence.
\end{assumption}

\begin{assumption}[Primitive sieve approximation and basis growth]
\label{ass:app-sieve-approx-primitive}
The bases $b^{K_g}$ and $r^{K_\psi}$ are rich enough\footnote{Specifically, $\bigcup_{K_g}\G_{K_g}$ is dense in $\G$ and $\bigcup_{K_\psi}\bm{\Psi}_{K_\psi}$ is dense in $\bm{\Psi}$ under the sup norm. Since $(g_0,\psi_0)\in\ThetaSpace$, these density conditions imply the existence of the stated approximants.} that there exist
$\theta_K^0=(g_{0,K_g},\psi_{0,K_\psi})\in\ThetaSpace_K$ satisfying
\[
    e_{g,K_g}+e_{\psi,K_\psi}\to0 \quad \text{as} \quad (K_g, K_\psi) \to \infty,
\]
where we define
\[
    e_{g,K_g}:=\norm{g_{0,K_g}-g_0}_\infty,
    \qquad
    e_{\psi,K_\psi}:=\norm{\psi_{0,K_\psi}-\psi_0}_\infty .
\]
Furthermore, the micro weighting matrices are symmetric and positive semidefinite, and the
micro basis and approximation error satisfy
\[
    \lambda_{\max}\{W_{g,K_g}\}^{1/2}
    \left\{
    \E \Big[\norm{b^{K_g}(Z_{it},X_t)}^2 \Big]
    \right\}^{1/2}
    e_{g,K_g}
    \to0 .
\]
Finally, the population price-side weighting matrix is uniformly nonsingular: there exists a constant $C_\Sigma>0$ such that
\[
    \inf_{r\in\supp(R_t)}
    \lambda_{\min}\! \big\{\Sigma(r)\big\}
    \ge C_\Sigma.
\]
\end{assumption}
The additional rate condition is needed because the micro population criterion
moves with $K_g$. As the dimension of the micro moment increases, a small
sup-norm approximation error in $g_{0,K_g}$ can be amplified by the size of
$b^{K_g}(Z_{it},X_t)$ and by the weighting matrix $W_{g,K_g}$. In the common
normalization where
$\E\norm{b^{K_g}(Z_{it},X_t)}^2=O(K_g)$ and  $\lambda_{\max}(W_{g,K_g})=O(1),$
the displayed rate condition reduces to
\[
    \sqrt{K_g}\norm{g_{0,K_g}-g_0}_\infty\to0.
\]
For example, if $d_g$ denotes the dimension of $\mathcal D_g$ and a standard
tensor-product spline or wavelet sieve delivers
\[
    \norm{g_{0,K_g}-g_0}_\infty
    =
    O(K_g^{-s_g/d_g}),
\]
then $s_g>d_g/2$ is sufficient. No analogous dimension-dependent rate is
needed for the price-side population criterion as it is defined here:
$\mathcal Q_\psi$ is the full conditional-mean criterion, rather than a
finite-$K_q$ projection. Thus $K_q$ affects estimation of the sample
conditional mean but does not enter the population approximation argument.
These are standard sieve-approximation restrictions strengthened only to
account for the growing micro moment; see \cite{newey_chapter_1994},
\cite{newey_convergence_1997}, \cite{ai_efficient_2003}, and
\cite{chen_large_2007}.

\begin{proposition}[Primitive conditions imply compactness, approximation, and criterion preservation]
\label{prop:app-sieve-compact-approx}
Suppose the profile-Lipschitz conclusion of Proposition
\ref{prop:app-profile-lipschitz} holds. Then Assumptions
\ref{ass:app-compact-primitive} and
\ref{ass:app-sieve-approx-primitive} imply Assumption
\ref{ass:sieve-approx-main}.
\end{proposition}

\begin{proof}
By the Arzel\`a--Ascoli theorem, bounded H\"older balls on compact
domains are relatively compact in the sup norm. Assumption
\ref{ass:app-compact-primitive} imposes $\G$ and $\bfPsi$ as closed
subsets of these balls, with normalizations that are closed under uniform
convergence. Hence $\G$ and $\bfPsi$ are compact in the sup norm, and
$\ThetaSpace=\G\times\bfPsi$ is compact under the product norm. Assumption
\ref{ass:app-sieve-approx-primitive} gives an approximating sequence
$\theta_K^0\in\ThetaSpace_K$ such that
\[
    \norm{\theta_K^0-\theta_0}_{\ThetaSpace}
    =
    e_{g,K_g}+e_{\psi,K_\psi}
    \to0.
\]
It remains to establish criterion preservation. Define $B_{it,K_g}:=b^{K_g}(Z_{it},X_t)$.
By Proposition \ref{prop:app-profile-lipschitz}, there exists $C_a<\infty$
such that
\[
    \norm{a_t^*(g_{0,K_g})-a_t^*(g_0)}
    \le
    C_a \cdot e_{g,K_g}
\]
uniformly in $t$. Let $L_\sigma<\infty$ be a Lipschitz constant for the
multinomial-logit probability map. The definition of the profiled residual
then gives
\[
\begin{aligned}
    \norm{
    \rho_{it}^*(g_{0,K_g})-\rho_{it}^*(g_0)
    }
    &\le L_\sigma  \left\{
    \norm{
    g_{0,K_g}(Z_{it},X_t)-g_0(Z_{it},X_t)
    +
    a_t^*(g_{0,K_g})-a_t^*(g_0)
    }
    \right\}
    \\
    &\le
    L_\sigma
    \left\{
    \norm{
    g_{0,K_g}(Z_{it},X_t)-g_0(Z_{it},X_t)
    }
    +
    \norm{
    a_t^*(g_{0,K_g})-a_t^*(g_0)
    }
    \right\} \\
    &\le
    L_\sigma(1+C_a)e_{g,K_g}.
\end{aligned}
\]
Since $m_{g,K_g}(g_0)=0$,
\[
\begin{aligned}
    \norm{m_{g,K_g}(g_{0,K_g})} = \norm{m_{g,K_g}(g_{0,K_g}) - m_{g,K_g}(g_0)}
    &=
    \norm{
    \E\!\left[
    B_{it,K_g}\otimes
    \{\rho_{it}^*(g_{0,K_g})-\rho_{it}^*(g_0)\}
    \right]
    } \\
    &\le \E\big[\norm{B_{it, K_g}} \cdot \norm{\rho^*_{it}(g_{0,K_g}) - \rho_{it}^*(g_0)}\big]
    \\
    &\le
    L_\sigma(1+C_a)
    \E\big[\norm{B_{it,K_g}}\big]
    e_{g,K_g} \\
    &\le
    L_\sigma(1+C_a)
    \left\{
    \E \Big[\norm{B_{it,K_g}}^2 \Big]
    \right\}^{1/2}
    e_{g,K_g},
\end{aligned}
\]
where the final inequality follows from Cauchy--Schwarz. Therefore\footnote{For any symmetric positive semidefinite matrix $W$ and vector $x$, $x'Wx \le \lambda_{\max}\{W\} \norm{x}^2$.},
\[
\begin{aligned}
    \mathcal Q_{g,K_g}(g_{0,K_g})
    &=
    m_{g,K_g}(g_{0,K_g})'
    W_{g,K_g}
    m_{g,K_g}(g_{0,K_g}) \\
    &\le
    \lambda_{\max}\{W_{g,K_g}\}
    \norm{m_{g,K_g}(g_{0,K_g})}^2 \\
    &\le
    L_\sigma^2(1+C_a)^2
    \lambda_{\max}\{W_{g,K_g}\}
    \E \Big[\norm{B_{it,K_g}}^2 \Big]
    e_{g,K_g}^2
    \to0
\end{aligned}
\]
by Assumption \ref{ass:app-sieve-approx-primitive}.

For the price-side component, define
\[
    \delta_{t,K}
    :=
    a_t^*(g_{0,K_g})-a_t^*(g_0)
    -
    \{\psi_{0,K_\psi}(P_t,X_t)-\psi_0(P_t,X_t)\}.
\]
Because $m_\psi(R_t;\theta_0)=0$ almost surely,
\[
    m_\psi(R_t;\theta_K^0)
    =
    \E[\delta_{t,K}\mid R_t].
\]
The profile-Lipschitz bound and sup-norm approximation imply
\[
    \norm{\delta_{t,K}}
    \le
    C_a e_{g,K_g}+e_{\psi,K_\psi}
\]
almost surely. Conditional Jensen's inequality therefore gives\footnote{For each realization of $R_t$, $\norm{\E[\delta_{t,K} \mid R_t]}^2 \le \E[\norm{\delta_{t,K}}^2 \mid R_t]$. Then take unconditional expectation of both sides and apply law of iterated expectations.}
\[
\begin{aligned}
    \E\Big[\norm{m_\psi(R_t;\theta_K^0)}^2\Big]
    &=
    \E\Big[\norm{\E[\delta_{t,K}\mid R_t]}^2\Big] \le
    \E \big[\norm{\delta_{t,K}}^2 \big] \le
    \left(
    C_a e_{g,K_g}+e_{\psi,K_\psi}
    \right)^2.
\end{aligned}
\]
It follows that\footnote{The uniform eigenvalue bound in Assumption \ref{ass:app-sieve-approx-primitive} implies
\(
    \lambda_{\max}\!\left(\Sigma(r)^{-1}\right)
    =
    \frac{1}{\lambda_{\min}(\Sigma(r))}
    \le
    \frac{1}{C_\Sigma}
\)
uniformly over $r\in\supp(R_t)$.}
\[
\begin{aligned}
    \mathcal Q_\psi(\theta_K^0)
    &=
    \E\!\left[
        m_\psi(R_t;\theta_K^0)'
        \Sigma(R_t)^{-1}
        m_\psi(R_t;\theta_K^0)
    \right] \\
    &\le
    \frac{1}{C_\Sigma}
    \E\Big[\norm{m_\psi(R_t;\theta_K^0)}^2\Big] \\
    &\le
    \frac{1}{C_\Sigma}
    \left(
        C_a e_{g,K_g}+e_{\psi,K_\psi}
    \right)^2
    \to0.
\end{aligned}
\]
Combining the two blocks yields
\[
    \mathcal Q_K(\theta_K^0)
    =
    \mathcal Q_{g,K_g}(g_{0,K_g})
    +
    \mathcal Q_\psi(\theta_K^0)
    \to0.
\]
Thus, the approximating sequence is criterion-preserving, and Assumption
\ref{ass:sieve-approx-main} follows.
\end{proof}

\begin{assumption}[Primitive conditions for moving population separation]
\label{ass:app-population-id-primitive}
The following conditions hold:
\begin{enumerate}[label=(\roman*),leftmargin=2em]
    \item The micro weighting matrices are uniformly well conditioned:
there exist constants $0<c_g\le C_g<\infty$ such that
\[
    \inf_{K_g} \ 
    \lambda_{\min}\!\big\{W_{g,K_g}\big\}
    \ge c_g,
    \qquad
    \sup_{K_g} \ 
    \lambda_{\max}\!\big\{W_{g,K_g}\big\}
    \le C_g.
\]

    \item The population conditions in Assumption
    \ref{ass:app-profile-primitive}(i), (ii), and (iv) hold with $\G$ in
    place of $\G_{K_g}^0$. Thus $a_t^*(g)$ is unique and interior for every
    $g\in\G$ and every $t\le T$, the index is uniformly bounded over
    $\G\times\Acal$, and the gradient bound in Assumption
    \ref{ass:app-profile-primitive}(iv) holds for every
    $g,\tilde g\in\G$, with constants independent of $t$.
\end{enumerate}
\end{assumption}

In Section \ref{sec:Estimator}, the micro basis is defined as
$b^{K_g}=(b_1,\ldots,b_{K_g})'$. Hence, for every fixed $k$, the moment
generated by $b_k$ is a component of the micro moment vector whenever
$K_g\ge k$. Assumption \ref{ass:app-sieve-approx-primitive} supplies the
required basis integrability and density conditions, as well as uniform
nonsingularity of the population price-side weighting matrix. Condition (i)
serves two purposes: the lower bound prevents
the moving micro criterion from assigning asymptotically vanishing weight to
informative moment directions and is used in the population-separation
argument. The upper bound prevents the weighting matrices from magnifying
sampling or profiling errors and is used in the subsequent uniform-convergence
arguments.

Condition (ii) extends only the \emph{population} profile conditions from the
sieves to their compact function space $\G$; it does not enlarge the
estimator's search space, which remains $\G_{K_g}$. This extension is needed
because a uniformly convergent sequence $g_\ell\in\G_{K_{g,\ell}}$ can have
a limit $g^*\in\G$ that belongs to no fixed sieve. The separation proof must
therefore evaluate $a_t^*(g^*)$ and pass the profiled moments to this limit.
The extension is natural under the maintained function-space restrictions.
Because $\G$ is contained in a bounded H\"older ball and $\Acal$ is compact,
the bounded-index and gradient-Lipschitz conditions hold uniformly on
$\G\times\Acal$ for the logit map. Uniform strong convexity then gives
uniqueness. The remaining interiority requirement amounts to choosing
$\Acal$ large enough to contain the population share-matching intercepts for
all $g\in\G$; this is ensured, for example, when the population inside- and
outside-good shares are uniformly bounded away from zero. Thus condition
(ii) is a standard global well-posedness restriction on the population
profile, not an additional restriction on the finite-dimensional estimator.

\begin{proposition}[Primitive conditions imply moving population separation]
\label{prop:app-population-id}
Suppose the maintained demand model and the identification conditions in
Section \ref{sec:Identification} hold. Then Assumptions
\ref{ass:app-compact-primitive},
\ref{ass:app-sieve-approx-primitive}, and
\ref{ass:app-population-id-primitive} imply Assumption
\ref{ass:population-id-main}.
\end{proposition}

\begin{proof}
Suppose, to the contrary, that Assumption
\ref{ass:population-id-main} fails. Since $\mathcal Q_K\ge0$, there exist
$\epsilon>0$, a subsequence $K_\ell = (K_{g,\ell}, K_{\psi,\ell}, K_{q,\ell})\to\infty$, and
$\theta_\ell=(g_\ell,\psi_\ell)\in\ThetaSpace_{K_\ell}$ such that
\[
    \norm{\theta_\ell-\theta_0}_{\ThetaSpace}\ge\epsilon,
    \qquad
    \mathcal Q_{K_\ell}(\theta_\ell)\to0.
\]
Indeed, along a subsequence on which the infimum in
\eqref{eq:moving-separation-main} converges to zero, one may choose a
$1/\ell$-approximate minimizer. By Assumption
\ref{ass:app-compact-primitive}, $\ThetaSpace$ is compact under the product
sup norm. Hence, along a further subsequence and retaining the same notation,
\[
    \norm{\theta_\ell-\theta^*}_{\ThetaSpace}\to0
\]
for some $\theta^*:=(g^*,\psi^*)\in\ThetaSpace$. The reverse triangle
inequality then yields
\begin{equation}
\label{eq:app-limit-away}
    \norm{\theta^*-\theta_0}_{\ThetaSpace}\ge\epsilon.
\end{equation}

Because the two components of $\mathcal Q_{K_\ell}$ are nonnegative,
\[
    c_g\norm{m_{g,K_{g,\ell}}(g_\ell)}^2
    \le
    \mathcal Q_{g,K_{g,\ell}}(g_\ell)
    \le
    \mathcal Q_{K_\ell}(\theta_\ell)
    \to0,
    \qquad
    \mathcal Q_\psi(\theta_\ell)\to0,
\]
where the first inequality uses Assumption
\ref{ass:app-population-id-primitive}(i). Thus
$\norm{m_{g,K_{g,\ell}}(g_\ell)}\to0$.

For each fixed $k\ge1$, let
\[
    \mu_k(g)
    :=
    \E\!\left[
        b_k(Z_{it},X_t)\otimes\rho_{it}^*(g)
    \right].
\]
Since $b^{K_g}=(b_1,\ldots,b_{K_g})'$ and
$K_{g,\ell}\to\infty$, $\mu_k(g_\ell)$ is a subvector of
$m_{g,K_{g,\ell}}(g_\ell)$ for all sufficiently large $\ell$. Thus
$\mu_k(g_\ell)\to0$. Assumption
\ref{ass:app-population-id-primitive}(ii) extends the profile-Lipschitz
argument in Proposition \ref{prop:app-profile-lipschitz} to all of $\G$.
Therefore, for a constant $C_a<\infty$,
\[
\begin{aligned}
    \norm{\mu_k(g_\ell)-\mu_k(g^*)}
    &\le
    L_\sigma(1+C_a)
    \E\!\left[\abs{b_k(Z_{it},X_t)}\right]
    \norm{g_\ell-g^*}_\infty
    \to0.
\end{aligned}
\]
Here $L_\sigma$ is a Lipschitz constant for the logit map, and the
expectation is finite by Assumption
\ref{ass:app-sieve-approx-primitive}. Consequently,
\[
    \mu_k(g^*)=0
    \qquad\text{for every }k\ge1.
\]

For positive integer tending to infinity $M$, use the density of the sieve space to choose vector-valued finite basis expansions
$\tilde g_{0,M}$ and $\tilde g_M^*$ that converge uniformly to $g_0$ and
$g^*$, respectively. Then
$v_M:=\tilde g_{0,M}-\tilde g_M^*$ has components in
$\operatorname{span}\{b_1,\ldots,b_M\}$ and
$v_M\to g_0-g^*$ uniformly. The equalities $\mu_k(g^*)=0$ imply
$\E[v_M(Z_{it},X_t)'\rho_{it}^*(g^*)]=0$ for every $M$. Since
$\rho_{it}^*(g^*)$ is bounded, uniform convergence gives
\begin{equation}
\label{eq:app-limit-score-direction}
    \E\!\left[
        \{g_0(Z_{it},X_t)-g^*(Z_{it},X_t)\}'
        \rho_{it}^*(g^*)
    \right]
    =0.
\end{equation}

The global interiority condition in Assumption
\ref{ass:app-population-id-primitive}(ii) and Lemma
\ref{lem:app-inner-gradient} give the profile first-order condition
\[
    \E\!\left[\rho_{it}^*(g^*)\mid\Mcal_t\right]=0.
\]
where $\Mcal_t := (P_t, X_t, W_t, \Xi_t).$ Because $a_{0t}-a_t^*(g^*)$ is measurable with respect to the $\sigma$-algebra generated by $\Mcal_t$,
\begin{equation}
\label{eq:app-limit-profile-direction}
    \E\!\left[
        \{a_{0t}-a_t^*(g^*)\}'
        \rho_{it}^*(g^*)
    \right]
    =0.
\end{equation}

Define
\[
    u_{it}^*
    :=
    g^*(Z_{it},X_t)+a_t^*(g^*),
    \qquad
    u_{it}^0
    :=
    g_0(Z_{it},X_t)+a_{0t},
    \qquad
    d_{it}:=u_{it}^0-u_{it}^*.
\]
Equations \eqref{eq:app-limit-score-direction} and
\eqref{eq:app-limit-profile-direction} imply
\[
    \E[d_{it}'\rho_{it}^*(g^*)]=0.
\]
Let $\varepsilon_{it}^0:=Q_{it}-\sigma(u_{it}^0)$. The maintained demand
model gives $\E[\varepsilon_{it}^0\mid C_{it}]=0$, and hence
\[
    \E\!\left[
        \{g_0(Z_{it},X_t)-g^*(Z_{it},X_t)\}'
        \varepsilon_{it}^0
    \right]=0.
\]
Moreover, Proposition \ref{prop:app-true-profile} gives
$a_t^*(g_0)=a_{0t}$. The profile first-order condition at $g_0$ therefore
implies $\E[\varepsilon_{it}^0\mid\Mcal_t]=0$, so
\[
    \E\!\left[
        \{a_{0t}-a_t^*(g^*)\}'\varepsilon_{it}^0
    \right]=0.
\]
Combining the last three displays yields
\[
    0
    =
    \E\!\left[
        d_{it}'\{\rho_{it}^*(g^*)-\varepsilon_{it}^0\}
    \right]
    =
    \E\!\left[
        d_{it}'\{\sigma(u_{it}^0)-\sigma(u_{it}^*)\}
    \right].
\]
By Assumption \ref{ass:app-population-id-primitive}(ii), the logit
Jacobian has the same uniform lower eigenvalue bound $c_\Lambda>0$ as in
Lemma \ref{lem:app-inner-separation}. The mean-value formula therefore gives
\[
    d_{it}'\{\sigma(u_{it}^0)-\sigma(u_{it}^*)\}
    \ge
    c_\Lambda\norm{d_{it}}^2.
\]
It follows that $d_{it}=0$ almost surely, or equivalently
$u_{it}^*=u_{it}^0$ almost surely. Proposition
\ref{prop:app-within-market-id-restated} and the normalization on $g$ then
imply
\begin{equation}
\label{eq:app-limit-within-id}
    g^*=g_0,
    \qquad
    a_t^*(g^*)=a_{0t}.
\end{equation}

It remains to pass the price-side moment to the limit. Proposition
\ref{prop:app-profile-lipschitz}, extended to $\G$ by Assumption
\ref{ass:app-population-id-primitive}(ii), and conditional Jensen's inequality give
\[
\begin{aligned}
    \Delta_\ell
    &:=
    \operatorname*{ess\,sup}_{r\in\supp(R_t)}
    \norm{
        m_\psi(r;\theta_\ell)-m_\psi(r;\theta^*)
    } \le
    C_a\norm{g_\ell-g^*}_\infty
    +
    \norm{\psi_\ell-\psi^*}_\infty
    \to0.
\end{aligned}
\]
where the essential supremum is taken with respect to the distribution of
$R_t$. The conditional moments are uniformly essentially bounded because $\Acal$ and
$\bfPsi$ are uniformly bounded, while Assumption
\ref{ass:app-sieve-approx-primitive} implies
$\sup_r\norm{\Sigma(r)^{-1}}<\infty$. Hence, for some finite constant $C$,
\[
    \abs{
        \mathcal Q_\psi(\theta_\ell)
        -
        \mathcal Q_\psi(\theta^*)
    }
    \le C\Delta_\ell\to0.
\]
Since $\mathcal Q_\psi(\theta_\ell)\to0$, we obtain
$\mathcal Q_\psi(\theta^*)=0$. Nonnegativity of its integrand and positive
definiteness of $\Sigma(R_t)^{-1}$ then imply
\[
    m_\psi(R_t;\theta^*)=0
    \qquad\text{a.s.}.
\]

Using \eqref{eq:app-limit-within-id}, this last equality becomes
\[
    \E[a_{0t}-\psi^*(P_t,X_t)\mid W_t,X_t]=0.
\]
Subtracting the corresponding moment for $\psi_0$ and applying completeness
from Assumption \ref{ass:price-id-main} yields
$\psi^*(P_t,X_t)=\psi_0(P_t,X_t)$ almost surely. Because both functions are
continuous and $\mathcal D_\psi=\supp(P_t,X_t)$, this equality holds on
$\mathcal D_\psi$. Together with \eqref{eq:app-limit-within-id}, it follows
that $\theta^*=\theta_0$ under the product sup norm, contradicting
\eqref{eq:app-limit-away}. Therefore Assumption
\ref{ass:population-id-main} holds.
\end{proof}

\begin{assumption}[Oracle micro-block uniform law]
\label{ass:app-micro-UC}
Define the oracle micro moment
\[
    \hat m_{g,K_g}^{\circ}(g)
    :=
    \frac{1}{N}\sum_{t=1}^T\sum_{i=1}^{n_t}
    b^{K_g}(Z_{it},X_t)\otimes
    \{Q_{it}-\sigma(g(Z_{it},X_t)+a_t^*(g))\}.
\]
Then
\[
    \sup_{g\in\G_{K_g}}
    \norm{\hat m_{g,K_g}^{\circ}(g)-m_{g,K_g}(g)}
    =
    o_p(1).
\]
In addition, the oracle moment convergence is fast enough relative to the
size of the population moment that
\[
    \lambda_{\max}(W_{g,K_g})
    \left[
    \sup_{g\in\G_{K_g}}
    \norm{\hat m_{g,K_g}^{\circ}(g)-m_{g,K_g}(g)}
    \right]
    \left[
    \sup_{g\in\G_{K_g}}
    \norm{m_{g,K_g}(g)}
    \right]
    =
    o_p(1).
\]
\end{assumption}

This assumption is a high-level uniform law of large numbers for the oracle micro block. It is reasonable because, after replacing the estimated profile $\hat a_t(g)$ by the population profile $a_t^*(g)$, the micro moment is an ordinary sample average indexed by the finite-dimensional sieve parameter $g\in\G_{K_g}$. The residual term is uniformly bounded: $Q_{it}$ is a vector of choice indicators and $\sigma(\cdot)$ is a vector of choice probabilities. Thus the only part of the summand whose size can grow with $K_g$ is the basis vector $b^{K_g}(Z_{it},X_t)$. The Lipschitz continuity of the logit map and of the population profile implies that small changes in $g$ produce small changes in the profiled residual. Therefore the oracle moment class has the same basic finite-dimensional structure as the sieve class for $g$, with no additional nonparametric object introduced by profiling. Under bounded-index conditions, square-integrability of the basis functions, and sufficiently slow growth of $K_g$ relative to the market and within-market sample sizes, standard uniform laws for sieve extremum and sieve GMM estimators imply the displayed convergence; see \cite{newey_chapter_1994}, \cite{andrews_chapter_1994}, \cite{newey_convergence_1997}, \cite{ai_efficient_2003}, and \cite{chen_large_2007}.

\begin{proposition}[Oracle micro criterion convergence]
\label{prop:app-oracle-micro-criterion-UC}
Under Assumptions \ref{ass:app-population-id-primitive}(i) and \ref{ass:app-micro-UC},
\[
    \sup_{g\in\G_{K_g}}
    \abs{
    \hat{\mathcal Q}_{g,K_g}^{\circ}(g)
    -
    \mathcal Q_{g,K_g}(g)
    }
    =
    o_p(1).
\]
\end{proposition}

\begin{proof}
Let
\[
    \Delta_g:=\hat m_{g,K_g}^{\circ}(g)-m_{g,K_g}(g).
\]
Since $\hat m_{g,K_g}^{\circ}(g)=m_{g,K_g}(g)+\Delta_g$,
\[
\hat{\mathcal Q}_{g,K_g}^{\circ}(g)-\mathcal Q_{g,K_g}(g)
=
2\Delta_g'W_{g,K_g}m_{g,K_g}(g)
+
\Delta_g'W_{g,K_g}\Delta_g .
\]
Using $\abs{x'Wy}\le \lambda_{\max}(W)\norm{x}\norm{y}$ for symmetric positive semidefinite $W$,
\[
\begin{aligned}
\sup_{g\in\G_{K_g}}
\abs{
\hat{\mathcal Q}_{g,K_g}^{\circ}(g)
-
\mathcal Q_{g,K_g}(g)
}
&\le
2\lambda_{\max}(W_{g,K_g})
\sup_{g\in\G_{K_g}}\norm{\Delta_g}
\sup_{g\in\G_{K_g}}\norm{m_{g,K_g}(g)}
\\
&\quad+
\lambda_{\max}(W_{g,K_g})
\left(
\sup_{g\in\G_{K_g}}\norm{\Delta_g}
\right)^2 .
\end{aligned}
\]
The first term is $o_p(1)$ by the second display in Assumption
\ref{ass:app-micro-UC}. The second term is $o_p(1)$ because
$\sup_g\norm{\Delta_g}=o_p(1)$ and
$\lambda_{\max}(W_{g,K_g})=O(1)$. Therefore
\[
    \sup_{g\in\G_{K_g}}
    \abs{
    \hat{\mathcal Q}_{g,K_g}^{\circ}(g)
    -
    \mathcal Q_{g,K_g}(g)
    }
    =
    o_p(1).
\]
\end{proof}

\begin{assumption}[Micro feasible-profile negligibility]
\label{ass:app-micro-profile-negligible}
Let
\[
    e_T:=\sup_{g\in\G_{K_g}}\max_{t\le T}
    \norm{\hat a_t(g)-a_t^*(g)} .
\]
The profile error is small relative to the empirical second moment of the growing basis
\[
    e_T
    \left\{
    \frac1N\sum_{t=1}^T\sum_{i=1}^{n_t}
    \norm{b^{K_g}(Z_{it},X_t)}^2
    \right\}^{1/2}
    =
    o_p(1).
\]
In addition, the feasible-oracle moment difference is small enough relative to the size of the oracle moment that
\[
    \lambda_{\max}(W_{g,K_g})
    \left[
    \sup_{g\in\G_{K_g}}
    \norm{\hat m_{g,K_g}(g)-\hat m_{g,K_g}^{\circ}(g)}
    \right]
    \left[
    \sup_{g\in\G_{K_g}}
    \norm{\hat m_{g,K_g}^{\circ}(g)}
    \right]
    =
    o_p(1).
\]
\end{assumption}
This assumption ensures that replacing the population profile $a_t^*(g)$ by the estimated profile $\hat a_t(g)$ does not affect the micro block asymptotically. Proposition \ref{prop:app-profile-consistency-restated} gives uniform consistency of the profile, so $e_T=o_p(1)$. The additional requirement above accounts for the fact that the profile error enters the micro moment multiplied by the growing basis vector $b^{K_g}(Z_{it},X_t)$. Since the logit map is Lipschitz on bounded index sets,
\[
\norm{
\sigma(g(Z_{it},X_t)+\hat a_t(g))
-
\sigma(g(Z_{it},X_t)+a_t^*(g))
}
\le C e_T
\]
uniformly in $g,i,t$. Hence the feasible-oracle difference in the micro moment is bounded by the profile error times the empirical size of the basis. The first displayed condition requires this product to vanish. The bounded-eigenvalue condition on $W_{g,K_g}$ then prevents the weighting matrix from magnifying a negligible moment difference into a non-negligible criterion difference.

The final display is needed only because the micro criterion is quadratic in the moment. When the criterion is expanded around the oracle moment, a cross term appears:
\[
    2d_g'W_{g,K_g}\hat m_{g,K_g}^{\circ}(g),
    \qquad
    d_g:=\hat m_{g,K_g}(g)-\hat m_{g,K_g}^{\circ}(g).
\]
The final condition requires the feasible-oracle moment error to vanish fast enough that this cross term is negligible. This avoids imposing the stronger and potentially unnatural requirement that the growing-dimensional oracle moment itself be uniformly $O_p(1)$.

\begin{proposition}[Micro feasible-oracle equivalence]
\label{prop:app-micro-profile-negligible}
Under Assumptions \ref{ass:app-population-id-primitive}(i) and \ref{ass:app-micro-profile-negligible},
\[
    \sup_{g\in\G_{K_g}}
    \norm{\hat m_{g,K_g}(g)-\hat m_{g,K_g}^{\circ}(g)}
    =
    o_p(1).
\]
Consequently,
\[
    \sup_{g\in\G_{K_g}}
    \abs{
    \hat{\mathcal Q}_{g,K_g}(g)
    -
    \hat{\mathcal Q}_{g,K_g}^{\circ}(g)
    }
    =
    o_p(1),
\]
where
\[
    \hat{\mathcal Q}_{g,K_g}^{\circ}(g)
    :=
    \hat m_{g,K_g}^{\circ}(g)'W_{g,K_g}\hat m_{g,K_g}^{\circ}(g).
\]
\end{proposition}

\begin{proof}
Fix $g\in\G_{K_g}$. Subtracting the oracle micro moment from the feasible micro moment cancels the $Q_{it}$ term, so
\[
\begin{aligned}
\hat m_{g,K_g}(g)-\hat m_{g,K_g}^{\circ}(g)
&=
\frac1N\sum_{t=1}^T\sum_{i=1}^{n_t}
b^{K_g}(Z_{it},X_t)\otimes
\Big[
\sigma(g(Z_{it},X_t)+a_t^*(g))
-
\sigma(g(Z_{it},X_t)+\hat a_t(g))
\Big].
\end{aligned}
\]
Let $L_\sigma<\infty$ be a Lipschitz constant for the logit map on the bounded index set. Using $\norm{u\otimes v}=\norm{u}\norm{v}$ and the triangle inequality,
\[
\begin{aligned}
\norm{\hat m_{g,K_g}(g)-\hat m_{g,K_g}^{\circ}(g)}
&\le
\frac{L_\sigma}{N}
\sum_{t=1}^T\sum_{i=1}^{n_t}
\norm{b^{K_g}(Z_{it},X_t)}
\norm{\hat a_t(g)-a_t^*(g)} .
\end{aligned}
\]
Taking the supremum over $g$ and using the definition of $e_T$ gives
\[
\sup_{g\in\G_{K_g}}
\norm{\hat m_{g,K_g}(g)-\hat m_{g,K_g}^{\circ}(g)}
\le
L_\sigma e_T
\frac1N\sum_{t=1}^T\sum_{i=1}^{n_t}
\norm{b^{K_g}(Z_{it},X_t)} .
\]
By Cauchy--Schwarz, the average basis norm is bounded by the square root of its empirical second moment. Hence
\[
\sup_{g\in\G_{K_g}}
\norm{\hat m_{g,K_g}(g)-\hat m_{g,K_g}^{\circ}(g)}
\le
L_\sigma e_T
\left\{
\frac1N\sum_{t=1}^T\sum_{i=1}^{n_t}
\norm{b^{K_g}(Z_{it},X_t)}^2
\right\}^{1/2}
=o_p(1)
\]
by Assumption \ref{ass:app-micro-profile-negligible}. This proves the moment equivalence.

For the criterion equivalence, write
\[
d_g:=\hat m_{g,K_g}(g)-\hat m_{g,K_g}^{\circ}(g),
\qquad
\hat m_g^\circ:=\hat m_{g,K_g}^{\circ}(g).
\]
Since $\hat m_{g,K_g}(g)=\hat m_g^\circ+d_g$,
\[
\hat{\mathcal Q}_{g,K_g}(g)-\hat{\mathcal Q}_{g,K_g}^{\circ}(g)
=
2d_g'W_{g,K_g}\hat m_g^\circ+d_g'W_{g,K_g}d_g .
\]
Using $\abs{x'Wy}\le \lambda_{\max}(W)\norm{x}\norm{y}$ for symmetric positive semidefinite $W$,
\[
\begin{aligned}
\sup_{g\in\G_{K_g}}
\abs{\hat{\mathcal Q}_{g,K_g}(g)-\hat{\mathcal Q}_{g,K_g}^{\circ}(g)}
&\le
2\lambda_{\max}(W_{g,K_g})
\sup_{g\in\G_{K_g}}\norm{d_g}
\sup_{g\in\G_{K_g}}\norm{\hat m_g^\circ} +
\lambda_{\max}(W_{g,K_g})
\left(\sup_{g\in\G_{K_g}}\norm{d_g}\right)^2 .
\end{aligned}
\]
The first term is $o_p(1)$ by the final display in Assumption
\ref{ass:app-micro-profile-negligible}. For the second term,
the moment equivalence gives $\sup_g\norm{d_g}=o_p(1)$, and the
bounded-eigenvalue condition gives $\lambda_{\max}(W_{g,K_g})=O(1)$.
Thus the quadratic remainder is also $o_p(1)$. Therefore
\[
    \sup_{g\in\G_{K_g}}
    \abs{
    \hat{\mathcal Q}_{g,K_g}(g)
    -
    \hat{\mathcal Q}_{g,K_g}^{\circ}(g)
    }
    =
    o_p(1),
\]
which proves the criterion equivalence.
\end{proof}

Before stating the Ai--Chen conditions, define the series design matrix
\[
    \bm Q_T :=
    \begin{pmatrix}
    q^{K_q}(R_1)' \\
    \vdots \\
    q^{K_q}(R_T)'
    \end{pmatrix}
    \in\mathbb R^{T\times K_q}.
\]
The population and empirical Gram matrices are
\[
    Q_K:=\E[q^{K_q}(R_t)q^{K_q}(R_t)'],
    \qquad
    \hat Q_K:=\frac1T\bm Q_T'\bm Q_T.
\]
For a matrix $A$, $\norm{A}_F$ denotes the Frobenius norm; in particular, if $A$ is $T\times J$ with $t$-th row $A_t'$, then
\[
    \norm{A}_F^2=\sum_{t=1}^T\norm{A_t}^2.
\]

\begin{assumption}[Ai-Chen Conditions]
\label{ass:app-AC-primitive}
Let $R_t=(W_t,X_t)$ and $\mathcal R=\supp(R_t)$. Write
$\theta=(g,\psi)$, $h_t^*(\theta):=a_t^*(g)-\psi(P_t,X_t),$ and $m_\psi(r;\theta):=\E[h_t^*(\theta)\mid R_t=r].$
The following conditions hold.
\begin{enumerate}[label=(\roman*),leftmargin=2em]
    \item Assumption \ref{ass:sampling-main} holds, and
    $\mathcal R=\supp(R_t)$ is compact.

    \item The basis $q^{K_q}(R_t)$ has dense linear span in $L_2(R_t)$.
    Its population Gram matrix $Q_K$ has eigenvalues bounded away from zero and infinity uniformly in
    $K_q$.

    \item The series dimension satisfies $K_q\to\infty$, $K_q/T\to0$, and the
    empirical Gram matrix, $\hat{Q}_K$, is well-conditioned with probability approaching
    one: $\lambda_{\min}\{\hat{Q}_K\}\ge c$ and $\lambda_{\max} \{\hat{Q}_K\}\le C$
    for constants $0<c<C<\infty$.

    \item The conditional-mean class is uniformly approximable by the
    series basis at the realized market points. Let $\Pi_{K_q}$ denote
    projection onto $\mathrm{span}\{q^{K_q}\}$. Then
    \[
        \sup_{\theta\in\ThetaSpace_K}
        \frac1T\sum_{t=1}^T
        \norm{
        \Pi_{K_q}m_\psi(R_t;\theta)-m_\psi(R_t;\theta)
        }^2
        =
        o_p(1).
    \]

    \item The implied shocks $h_t^*(\theta)$ are uniformly
    square-integrable and Lipschitz in $\theta$ under
    $\norm{\cdot}_{\ThetaSpace}$. The centered residuals $\eta_t(\theta):=
        h_t^*(\theta)-m_\psi(R_t;\theta)$
    satisfy the standard uniform stochastic projection bound
    \[
        \sup_{\theta\in\ThetaSpace_K}
        \frac1T\norm{P_{q,T}\eta(\theta)}_F^2=o_p(1),
    \]
    where $P_{q,T}:=\bm Q_T(\bm Q_T'\bm Q_T)^{-1}\bm Q_T'$ is the empirical projection matrix onto the span of the series basis and $\eta(\theta)$ is the $T\times J$ matrix with $t$-th row
    $\eta_t(\theta)'$.

    \item The estimated weighting matrix satisfies
    \[
        \sup_{r\in\mathcal R}\norm{\hat\Sigma(r)-\Sigma(r)}=o_p(1),
    \]
    where $\Sigma(r)$ is uniformly positive definite and uniformly
    bounded.

    \item The population quadratic class satisfies the uniform law \[ \sup_{\theta\in\ThetaSpace_K} \left| \frac1T\sum_{t=1}^T m_\psi(R_t;\theta)'\Sigma(R_t)^{-1}m_\psi(R_t;\theta) - \E\!\left[ m_\psi(R_t;\theta)'\Sigma(R_t)^{-1}m_\psi(R_t;\theta) \right] \right| = o_p(1). \]
\end{enumerate}
\end{assumption}

These conditions are high-level versions of standard restrictions used in sieve conditional-moment and nonparametric IV estimation; see \cite{newey_instrumental_2003}, \cite{ai_efficient_2003}, \cite{donald_empirical_2003}, \cite{hall_nonparametric_2005}, and \cite{chen_sieve_2015}.

Conditions \ref{ass:app-AC-primitive}(i)--(vii) collect the regularity
requirements for the price-side conditional-moment SMD block. Conditions
(i)--(iii) are primitive series-design conditions: markets are independent,
the support of $R_t$ is compact, the basis $q^{K_q}$ is rich in
$L_2(\mathcal R)$, and the population and empirical Gram matrices are well
conditioned as the number of series terms grows. Condition (iv) is the
series approximation requirement. It is stated directly in the empirical
form used in the proof: the projection
$\Pi_{K_q}m_\psi(\cdot;\theta)$ must approximate
$m_\psi(\cdot;\theta)$ uniformly over
$\theta\in\ThetaSpace_K$ at the realized market points. A more primitive
formulation would impose uniform approximation in $L_2(\mathcal R)$
together with a uniform law of large numbers for the squared
approximation-error class.\footnote{A sufficient primitive route to
condition (iv) is the following. Let
$e_K(r;\theta):=\Pi_{K_q}m_\psi(r;\theta)-m_\psi(r;\theta)$. If
$\sup_{\theta\in\ThetaSpace_K}\E\norm{e_K(R_t;\theta)}^2\to0$ and
\[
    \sup_{\theta\in\ThetaSpace_K}
    \left|
    T^{-1}\sum_{t=1}^T\norm{e_K(R_t;\theta)}^2
    -
    \E\norm{e_K(R_t;\theta)}^2
    \right|=o_p(1),
\]
then
$\sup_{\theta\in\ThetaSpace_K}T^{-1}\sum_{t=1}^T
\norm{e_K(R_t;\theta)}^2=o_p(1)$. This is exactly condition (iv).}
Condition (vi) is the regularity condition for the estimated weighting
matrix.

Conditions (v) and (vii) are more high-level. Condition (v) controls the
stochastic part of the series regression. Writing
$h_t^*(\theta)=m_\psi(R_t;\theta)+\eta_t(\theta)$, the residual class
is $\{\eta_t(\theta):\theta\in\ThetaSpace_K\}$, and each residual is
conditionally mean zero given $R_t$. Condition (v) requires that the
fitted component obtained by projecting this centered noise onto the
series basis be asymptotically negligible, uniformly over $\theta$.
Condition (vii) is a uniform law of large numbers for the population
quadratic class appearing in the SMD objective. These conditions isolate
the stochastic projection and sample-average pieces from which convergence
of the SMD criterion follows. They could be derived from more primitive
entropy, moment, and growth restrictions on the residual and quadratic
classes, as in standard sieve minimum-distance arguments. I state them
directly to keep the appendix focused on the additional issue created by
profiling. In the present model, profiling enters through the generated
shock $h_t^*(\theta)=a_t^*(g)-\psi(P_t,X_t)$, and the Lipschitz
continuity of $g\mapsto a_t^*(g)$ ensures that this generated shock
varies smoothly with the structural functions.

\begin{proposition}[Validity of the price-side SMD moment, restated]
\label{prop:app-price-smd-valid-restated}
Under the maintained demand model, the profile identity $a_t^*(g_0)=a_{0t}$, and Assumptions \ref{ass:normalizations-main} and \ref{ass:price-id-main},
\[
    m_\psi(R_t;g_0,\psi_0)=0
    \qquad\text{a.s.}
\]
\end{proposition}

\begin{proof}
By definition,
\[
    m_\psi(R_t;g_0,\psi_0)
    =
    \E[a_t^*(g_0)-\psi_0(P_t,X_t)\mid R_t].
\]
Since $R_t=(W_t,X_t)$ and $a_t^*(g_0)=a_{0t}$,
\[
    m_\psi(R_t;g_0,\psi_0)
    =
    \E[a_{0t}-\psi_0(P_t,X_t)\mid W_t,X_t].
\]
The right-hand side is zero by Proposition \ref{prop:population-restrictions-main}. Hence the price-side SMD moment is centered at the truth.
\end{proof}

Define the oracle implied shock $h_t^*(g,\psi):=a_t^*(g)-\psi(P_t,X_t)$ and the oracle series estimate $\hat m_{\psi,K_q}^{\circ}$ by replacing $\hat a_t(g)$ with $a_t^*(g)$ in \eqref{eq:mpsi-hat-main}. Let $\hat{\mathcal Q}_{\psi,K_q}^{\circ}$ denote the corresponding oracle SMD criterion.

\begin{proposition}[Oracle price-side SMD convergence]
\label{prop:app-oracle-smd}
Under Assumption \ref{ass:app-AC-primitive},
\[
    \sup_{(g,\psi)\in\ThetaSpace_K}
    \abs{\hat{\mathcal Q}_{\psi,K_q}^{\circ}(g,\psi)-\mathcal Q_\psi(g,\psi)}=o_p(1),
\]
where $\mathcal Q_\psi(g,\psi):=\E[m_\psi(R_t;g,\psi)'\Sigma(R_t)^{-1}m_\psi(R_t;g,\psi)]$.
\end{proposition}

\begin{proof} 
Write $\theta=(g,\psi)$. Let $H^*(\theta)$ be the $T\times J$ matrix with $t$-th row $h_t^*(\theta)'$, and let $M(\theta)$ be the $T\times J$ matrix with $t$-th row $m_\psi(R_t;\theta)'$. Thus $H^*(\theta)$ stacks the oracle implied shocks, while $M(\theta)$ stacks their population conditional means at the observed values of $R_t$. The oracle fitted conditional mean is the matrix of fitted values from the series regression of $H^*(\theta)$ on $q^{K_q}(R_t)$: \[ \hat M^{\circ}(\theta):=P_{q,T}H^*(\theta), \qquad P_{q,T}:=\bm Q_T(\bm Q_T'\bm Q_T)^{-1}\bm Q_T' . \] Decompose \[ H^*(\theta)=M(\theta)+\eta(\theta), \] where $\eta(\theta)$ has $t$-th row \( \eta_t(\theta)'= h_t^*(\theta)'-m_\psi(R_t;\theta)'. \) Then \[ \begin{aligned} \hat M^{\circ}(\theta)-M(\theta) &=P_{q,T}H^*(\theta)-M(\theta) \\ &=P_{q,T}\{M(\theta)+\eta(\theta)\}-M(\theta) \\ &=P_{q,T}\eta(\theta)+\{P_{q,T}M(\theta)-M(\theta)\}. \end{aligned} \] The first term is the fitted component obtained by projecting centered residual noise onto the series basis, and Assumption \ref{ass:app-AC-primitive}(v) gives \[ \sup_{\theta\in\ThetaSpace_K} \frac1T\norm{P_{q,T}\eta(\theta)}_F^2=o_p(1). \] It remains to bound the approximation term $P_{q,T}M(\theta)-M(\theta)$. Let $M_K(\theta)$ be the $T\times J$ matrix with $t$-th row $\Pi_{K_q}m_\psi(R_t;\theta)'$. Since $\Pi_{K_q}m_\psi(\cdot;\theta)$ belongs to the span of $q^{K_q}$, there exists a $K_q\times J$ matrix $\Gamma_K(\theta)$ such that \[ M_K(\theta)=\bm Q_T\Gamma_K(\theta). \] Hence each column of $M_K(\theta)$ lies in the column space of $\bm Q_T$. Since $P_{q,T}$ is the orthogonal projection onto that column space, \[ P_{q,T}M_K(\theta)=M_K(\theta). \] Therefore \[ \begin{aligned} P_{q,T}M(\theta)-M(\theta) &= P_{q,T}\{M(\theta)-M_K(\theta)\} - \{M(\theta)-M_K(\theta)\}. \end{aligned} \] 
Because $P_{q,T}$ is an orthogonal projection, it is a contraction in Euclidean norm: \[ \norm{P_{q,T}}_{op} := \sup_{v\neq 0}\frac{\norm{P_{q,T}v}_2}{\norm{v}_2} \le 1 . \] Thus\footnote{Let $A(\theta):=M(\theta)-M_K(\theta)$. Since $M_K(\theta)$ lies in the column space of $\bm Q_T$, $P_{q,T}M_K(\theta)=M_K(\theta)$. Therefore \[ P_{q,T}M(\theta)-M(\theta) = P_{q,T}\{M_K(\theta)+A(\theta)\} - \{M_K(\theta)+A(\theta)\} = P_{q,T}A(\theta)-A(\theta). \] By the triangle inequality, \[ \norm{P_{q,T}A(\theta)-A(\theta)}_F \le \norm{P_{q,T}A(\theta)}_F+\norm{A(\theta)}_F . \] The operator-norm bound $\norm{P_{q,T}}_{op}\le1$ implies $\norm{P_{q,T}A(\theta)}_F\le\norm{A(\theta)}_F$, because left multiplication by $P_{q,T}$ applies the projection separately to each column of $A(\theta)$. Hence \[ \norm{P_{q,T}A(\theta)-A(\theta)}_F \le 2\norm{A(\theta)}_F . \] Squaring both sides and dividing by $T$ yields the stated inequality.}
\[ \frac1T\norm{P_{q,T}M(\theta)-M(\theta)}_F^2 \le \frac{4}{T}\norm{M(\theta)-M_K(\theta)}_F^2. \] 
By Assumption \ref{ass:app-AC-primitive}(iv) and the definition of $M_K(\theta)$, \[ \sup_{\theta\in\ThetaSpace_K} \frac1T\norm{M(\theta)-M_K(\theta)}_F^2 = \sup_{\theta\in\ThetaSpace_K} \frac1T\sum_{t=1}^T \norm{\Pi_{K_q}m_\psi(R_t;\theta)-m_\psi(R_t;\theta)}^2 = o_p(1). \] Together with the previous bound, we can see that \[ \sup_{\theta\in\ThetaSpace_K} \frac1T\norm{P_{q,T}M(\theta)-M(\theta)}_F^2=o_p(1). \] Now recall $\hat M^{\circ}(\theta)-M(\theta) = P_{q,T}\eta(\theta)+\{P_{q,T}M(\theta)-M(\theta)\}.$ Using $\norm{A+B}_F^2\le 2\norm{A}_F^2+2\norm{B}_F^2$, Assumption \ref{ass:app-AC-primitive}(v), and the preceding display imply \[ \sup_{\theta\in\ThetaSpace_K} \frac1T\norm{\hat M^{\circ}(\theta)-M(\theta)}_F^2=o_p(1). \] Equivalently, \begin{equation} \label{eq:oracle-fitted-L2-new} \sup_{\theta\in\ThetaSpace_K} \frac1T\sum_{t=1}^T \norm{\hat m_{\psi,K_q}^{\circ}(R_t;\theta)-m_\psi(R_t;\theta)}^2=o_p(1). \end{equation}

It remains to pass from fitted-mean convergence to criterion convergence. Define \[ \Delta_t(\theta):=\hat m_{\psi,K_q}^{\circ}(R_t;\theta)-m_\psi(R_t;\theta), \qquad m_t(\theta):=m_\psi(R_t;\theta), \qquad S_t:=\Sigma(R_t)^{-1}, \qquad \hat S_t:=\hat\Sigma(R_t)^{-1}. \] By Assumption \ref{ass:app-AC-primitive}(vi), uniform positive definiteness of $\Sigma$ and uniform convergence of $\hat\Sigma$ imply \[ \sup_{t\le T}\norm{\hat S_t-S_t}=o_p(1), \qquad \sup_{t\le T}\norm{\hat S_t}=O_p(1). \] Since $\hat m_{\psi,K_q}^{\circ}(R_t;\theta)=m_t(\theta)+\Delta_t(\theta)$, \[ \begin{aligned} &\hat m_{\psi,K_q}^{\circ}(R_t;\theta)'\hat S_t\hat m_{\psi,K_q}^{\circ}(R_t;\theta) - m_t(\theta)'S_t m_t(\theta) \\ &\quad = 2\Delta_t(\theta)'\hat S_t m_t(\theta) + \Delta_t(\theta)'\hat S_t\Delta_t(\theta) + m_t(\theta)'(\hat S_t-S_t)m_t(\theta). \end{aligned} \] We bound the three terms after averaging over $t$. First, by Cauchy--Schwarz\footnote{Specifically, \begin{align*}
    \left| \frac{1}{T} \sum_{t=1}^T 2 \Delta_t(\theta)' \hat{S}_t m_t(\theta) \right| & \le \frac{2}{T} \sum_{t=1}^T |\Delta_t(\theta)' \hat{S}_t m_t(\theta)| \le \frac{2}{T} \sum_{t=1}^T \| \Delta_t(\theta)\| \| \hat{S}_t\| \|m_t(\theta) \| \le 2 \sup_{s \le T} \|\hat{S}_s\| \frac{1}{T} \sum_{t=1}^T \|\Delta_t(\theta) \| \|m_t(\theta)\|
\end{align*}
then apply Cauchy--Schwarz to the remaining scalar sum and take the supremum over $\ThetaSpace_K$ of both sides.}
, \[ \begin{aligned} &\sup_{\theta\in\ThetaSpace_K} \left| \frac1T\sum_{t=1}^T 2\Delta_t(\theta)'\hat S_t m_t(\theta) \right| \\ &\quad\le 2\sup_{t\le T}\norm{\hat S_t} \left\{ \sup_{\theta\in\ThetaSpace_K} \frac1T\sum_{t=1}^T\norm{\Delta_t(\theta)}^2 \right\}^{1/2} \left\{ \sup_{\theta\in\ThetaSpace_K} \frac1T\sum_{t=1}^T\norm{m_t(\theta)}^2 \right\}^{1/2} =o_p(1). \end{aligned} \] Here the first bracketed term is $o_p(1)$ by \eqref{eq:oracle-fitted-L2-new}. 
The second bracketed term is $O_p(1)$ by Assumptions \ref{ass:app-AC-primitive}(v)--(vii).\footnote{To see this, write $m_t(\theta):=m_\psi(R_t;\theta)$ and $S_t:=\Sigma(R_t)^{-1}$. Since $\Sigma(r)$ is uniformly positive definite, there exists $c_S>0$ such that \[ m_t(\theta)'S_t m_t(\theta)\ge c_S\norm{m_t(\theta)}^2 \] for all $t$ and $\theta$. Hence \[ \sup_{\theta\in\ThetaSpace_K} \frac1T\sum_{t=1}^T\norm{m_t(\theta)}^2 \le c_S^{-1} \sup_{\theta\in\ThetaSpace_K} \frac1T\sum_{t=1}^T m_t(\theta)'S_t m_t(\theta). \] By Assumption \ref{ass:app-AC-primitive}(vii), the right-hand side is bounded by the corresponding population quadratic plus $o_p(1)$. The population quadratic is uniformly bounded because $S_t$ is uniformly bounded and, by conditional Jensen's inequality, \[ \E\norm{m_\psi(R_t;\theta)}^2 = \E\norm{\E[h_t^*(\theta)\mid R_t]}^2 \le \E\norm{h_t^*(\theta)}^2, \] which is uniformly bounded by Assumption \ref{ass:app-AC-primitive}(v). Therefore \( \sup_{\theta\in\ThetaSpace_K} T^{-1}\sum_{t=1}^T\norm{m_t(\theta)}^2=O_p(1). \)}
Second, \[ \begin{aligned} \sup_{\theta\in\ThetaSpace_K} \left| \frac1T\sum_{t=1}^T \Delta_t(\theta)'\hat S_t\Delta_t(\theta) \right| &\le \sup_{t\le T}\norm{\hat S_t} \sup_{\theta\in\ThetaSpace_K} \frac1T\sum_{t=1}^T\norm{\Delta_t(\theta)}^2 \\ &=o_p(1). \end{aligned} \] Third, Assumption \ref{ass:app-AC-primitive}(vi) implies $\sup_{s \le T} \|\hat{S}_s - S_s\| = o_p(1)$. So, \[ \begin{aligned} \sup_{\theta\in\ThetaSpace_K} \left| \frac1T\sum_{t=1}^T m_t(\theta)'(\hat S_t-S_t)m_t(\theta) \right| &\le \sup_{t\le T}\norm{\hat S_t-S_t} \sup_{\theta\in\ThetaSpace_K} \frac1T\sum_{t=1}^T\norm{m_t(\theta)}^2 \\ &=o_p(1). \end{aligned} \] Combining the three bounds yields \[ \sup_{\theta\in\ThetaSpace_K} \left| \hat{\mathcal Q}_{\psi,K_q}^{\circ}(\theta) - \frac1T\sum_{t=1}^T m_\psi(R_t;\theta)'S_t m_\psi(R_t;\theta) \right| = o_p(1). \] Finally, Assumption \ref{ass:app-AC-primitive}(vii) gives \[ \sup_{\theta\in\ThetaSpace_K} \left| \frac1T\sum_{t=1}^T m_\psi(R_t;\theta)'S_t m_\psi(R_t;\theta) - \mathcal Q_\psi(\theta) \right| = o_p(1). \] 
Combining the preceding display with Assumption
\ref{ass:app-AC-primitive}(vii), we obtain by triangle inequality
\[
\begin{aligned}
&\sup_{\theta\in\ThetaSpace_K}
\abs{\hat{\mathcal Q}_{\psi,K_q}^{\circ}(\theta)-\mathcal Q_\psi(\theta)} \\
&\le
\sup_{\theta\in\ThetaSpace_K}
\left|
\hat{\mathcal Q}_{\psi,K_q}^{\circ}(\theta)
-
\frac1T\sum_{t=1}^T
m_\psi(R_t;\theta)'S_t m_\psi(R_t;\theta)
\right| 
+
\sup_{\theta\in\ThetaSpace_K}
\left|
\frac1T\sum_{t=1}^T
m_\psi(R_t;\theta)'S_t m_\psi(R_t;\theta)
-
\mathcal Q_\psi(\theta)
\right| \\
&=o_p(1).
\end{aligned}
\]
This proves the desired uniform convergence.
\end{proof}

\begin{assumption}[Feasible-profile regularity] \label{ass:app-profile-error-negligible} The profile error is uniformly negligible: \[ e_T:=\sup_{g\in\G_{K_g}}\max_{t\le T} \norm{\hat a_t(g)-a_t^*(g)} = o_p(1). \]
\end{assumption}
Assumption \ref{ass:app-profile-error-negligible} is the uniform profile consistency condition established in Proposition \ref{prop:profile-consistency-main}. It is restated here only to keep the feasible-oracle comparison self-contained. 

\begin{proposition}[Profile-error negligibility] \label{prop:app-profile-negligible} Under Assumptions \ref{ass:app-AC-primitive} and \ref{ass:app-profile-error-negligible}, \[ \sup_{\theta\in\ThetaSpace_K} \abs{\hat{\mathcal Q}_{\psi,K_q}(\theta)-\hat{\mathcal Q}_{\psi,K_q}^{\circ}(\theta)} = o_p(1). \] Consequently, \[ \sup_{\theta\in\ThetaSpace_K} \abs{\hat{\mathcal Q}_{\psi,K_q}(\theta)-\mathcal Q_\psi(\theta)} = o_p(1). \] 
\end{proposition}

\begin{proof} Write $\theta=(g,\psi)$ and define \[ e_t(g):=\hat a_t(g)-a_t^*(g), \qquad e_T:=\sup_{g\in\G_{K_g}}\max_{t\le T}\norm{e_t(g)}. \] Let $E_g$ be the $T\times J$ matrix with $t$-th row $e_t(g)'$. The feasible and oracle implied shocks differ only because $\hat a_t(g)$ replaces $a_t^*(g)$. The matrices $\hat M(\theta)$ and $\hat M^{\circ}(\theta)$ collect the feasible and oracle fitted conditional moments evaluated at the sample values $R_1,\ldots,R_T$. Since both fitted moments are obtained by projecting the corresponding generated shocks onto the same series space, their difference is \[ \hat M(\theta)-\hat M^{\circ}(\theta)=P_{q,T}E_g. \]
where $\hat M(\theta)$ and $\hat M^{\circ}(\theta)$ are the $T\times J$ matrices with $t$-th rows $\hat m_{\psi,K_q}(R_t;\theta)'$ and $\hat m_{\psi,K_q}^{\circ}(R_t;\theta)'$, respectively. Since $P_{q,T}$ is an orthogonal projection, \[ \norm{P_{q,T}E_g}_F^2 \le \norm{E_g}_F^2 = \sum_{t=1}^T\norm{e_t(g)}^2 \le T e_T^2. \] Hence \begin{equation} \label{eq:profile-error-L2-new} \sup_{\theta\in\ThetaSpace_K} \frac1T\sum_{t=1}^T \norm{\hat m_{\psi,K_q}(R_t;\theta)-\hat m_{\psi,K_q}^{\circ}(R_t;\theta)}^2 \le e_T^2 = o_p(1). \end{equation} Next define \[ d_t(\theta):=\hat m_{\psi,K_q}(R_t;\theta)-\hat m_{\psi,K_q}^{\circ}(R_t;\theta), \qquad m_t^{\circ}(\theta):=\hat m_{\psi,K_q}^{\circ}(R_t;\theta), \qquad \hat S_t:=\hat\Sigma(R_t)^{-1}. \] By Assumption \ref{ass:app-AC-primitive}(vi), uniform positive definiteness of $\hat\Sigma$ implies \[ \sup_{t\le T}\norm{\hat S_t}=O_p(1). \] 
Since $\hat m_{\psi,K_q}(R_t;\theta)=m_t^{\circ}(\theta)+d_t(\theta)$, we have \[ \begin{aligned} \hat{\mathcal Q}_{\psi,K_q}(\theta) &= \frac1T\sum_{t=1}^T \{m_t^{\circ}(\theta)+d_t(\theta)\}'\hat S_t \{m_t^{\circ}(\theta)+d_t(\theta)\}, \\ \hat{\mathcal Q}_{\psi,K_q}^{\circ}(\theta) &= \frac1T\sum_{t=1}^T m_t^{\circ}(\theta)'\hat S_t m_t^{\circ}(\theta). \end{aligned} \] Therefore, \[ \begin{aligned} \hat{\mathcal Q}_{\psi,K_q}(\theta)-\hat{\mathcal Q}_{\psi,K_q}^{\circ}(\theta) &= \frac1T\sum_{t=1}^T \Big[ \{m_t^{\circ}(\theta)+d_t(\theta)\}'\hat S_t \{m_t^{\circ}(\theta)+d_t(\theta)\} - m_t^{\circ}(\theta)'\hat S_t m_t^{\circ}(\theta) \Big] \\ &= \frac1T\sum_{t=1}^T \left[ d_t(\theta)'\hat S_t m_t^{\circ}(\theta) + m_t^{\circ}(\theta)'\hat S_t d_t(\theta) + d_t(\theta)'\hat S_t d_t(\theta) \right]. \end{aligned} \] Since $\hat S_t=\hat\Sigma(R_t)^{-1}$ is symmetric, $m_t^{\circ}(\theta)'\hat S_t d_t(\theta)=d_t(\theta)'\hat S_t m_t^{\circ}(\theta)$. Hence \[ \begin{aligned} \hat{\mathcal Q}_{\psi,K_q}(\theta)-\hat{\mathcal Q}_{\psi,K_q}^{\circ}(\theta) &= \frac1T\sum_{t=1}^T \left[ 2d_t(\theta)'\hat S_t m_t^{\circ}(\theta) + d_t(\theta)'\hat S_t d_t(\theta) \right]. \end{aligned} \]
Taking absolute values, using $|x'Ay|\le \norm{A}\norm{x}\norm{y}$, and then applying Cauchy--Schwarz yields \[ \begin{aligned} &\sup_{\theta\in\ThetaSpace_K} \abs{\hat{\mathcal Q}_{\psi,K_q}(\theta)-\hat{\mathcal Q}_{\psi,K_q}^{\circ}(\theta)} \\ &\le 2\sup_{t\le T}\norm{\hat S_t} \left\{ \sup_{\theta\in\ThetaSpace_K} \frac1T\sum_{t=1}^T\norm{d_t(\theta)}^2 \right\}^{1/2} \left\{ \sup_{\theta\in\ThetaSpace_K} \frac1T\sum_{t=1}^T\norm{m_t^{\circ}(\theta)}^2 \right\}^{1/2} + \sup_{t\le T}\norm{\hat S_t} \sup_{\theta\in\ThetaSpace_K} \frac1T\sum_{t=1}^T\norm{d_t(\theta)}^2. \end{aligned} \] The empirical $L_2$ norm of $d_t(\theta)$ is $o_p(1)$ uniformly by \eqref{eq:profile-error-L2-new}. The oracle fitted moments are uniformly bounded in empirical $L_2$.\footnote{By \eqref{eq:oracle-fitted-L2-new}, \[ \sup_{\theta\in\ThetaSpace_K} \frac1T\sum_{t=1}^T \norm{\hat m_{\psi,K_q}^{\circ}(R_t;\theta)-m_\psi(R_t;\theta)}^2=o_p(1). \] Moreover, the proof of Proposition \ref{prop:app-oracle-smd} established \[ \sup_{\theta\in\ThetaSpace_K} \frac1T\sum_{t=1}^T \norm{m_\psi(R_t;\theta)}^2=O_p(1). \] Therefore, using $\norm{x+y}^2\le 2\norm{x}^2+2\norm{y}^2$, \[ \begin{aligned} \sup_{\theta\in\ThetaSpace_K} \frac1T\sum_{t=1}^T \norm{\hat m_{\psi,K_q}^{\circ}(R_t;\theta)}^2 &\le 2\sup_{\theta\in\ThetaSpace_K} \frac1T\sum_{t=1}^T \norm{\hat m_{\psi,K_q}^{\circ}(R_t;\theta)-m_\psi(R_t;\theta)}^2 + 2\sup_{\theta\in\ThetaSpace_K} \frac1T\sum_{t=1}^T \norm{m_\psi(R_t;\theta)}^2 = O_p(1). \end{aligned} \]} Thus the first term in the preceding bound is \[ O_p(1)\cdot o_p(1)^{1/2}\cdot O_p(1)^{1/2}=o_p(1), \] and the second term is \[ O_p(1)\cdot o_p(1)=o_p(1). \] Hence \[ \sup_{\theta\in\ThetaSpace_K} \abs{\hat{\mathcal Q}_{\psi,K_q}(\theta)-\hat{\mathcal Q}_{\psi,K_q}^{\circ}(\theta)} = o_p(1). \] Finally, by the triangle inequality and Proposition \ref{prop:app-oracle-smd}, \[ \begin{aligned} \sup_{\theta\in\ThetaSpace_K} \abs{\hat{\mathcal Q}_{\psi,K_q}(\theta)-\mathcal Q_\psi(\theta)} &\le \sup_{\theta\in\ThetaSpace_K} \abs{\hat{\mathcal Q}_{\psi,K_q}(\theta)-\hat{\mathcal Q}_{\psi,K_q}^{\circ}(\theta)} + \sup_{\theta\in\ThetaSpace_K} \abs{\hat{\mathcal Q}_{\psi,K_q}^{\circ}(\theta)-\mathcal Q_\psi(\theta)} = o_p(1). \end{aligned} \] This proves the result. 
\end{proof}

\begin{proposition}[Uniform convergence of the full moving criterion, restated]
\label{prop:app-full-UC}
Under Assumption \ref{ass:criterion-UC-main} (implied by \ref{ass:app-population-id-primitive}(i), \ref{ass:app-micro-UC},
\ref{ass:app-micro-profile-negligible},
\ref{ass:app-AC-primitive}, and
\ref{ass:app-profile-error-negligible}),
\[
    \sup_{(g,\psi)\in\ThetaSpace_K}
    \abs{\hat{\mathcal Q}_K(g,\psi)-\mathcal Q_K(g,\psi)}=o_p(1).
\]
\end{proposition}

\begin{proof}
Write
\[
\begin{aligned}
    \hat{\mathcal Q}_K(g,\psi)-\mathcal Q_K(g,\psi)
    &=
    \{\hat{\mathcal Q}_{g,K_g}(g)-\mathcal Q_{g,K_g}(g)\} +
    \{\hat{\mathcal Q}_{\psi,K_q}(g,\psi)-\mathcal Q_\psi(g,\psi)\}.
\end{aligned}
\]
The price-side term is uniformly $o_p(1)$ by Proposition
\ref{prop:app-profile-negligible}. For the micro term, add and subtract
the oracle micro criterion:
\[
\begin{aligned}
\sup_{g\in\G_{K_g}}
\abs{\hat{\mathcal Q}_{g,K_g}(g)-\mathcal Q_{g,K_g}(g)}
&\le
\sup_{g\in\G_{K_g}}
\abs{\hat{\mathcal Q}_{g,K_g}(g)-\hat{\mathcal Q}_{g,K_g}^{\circ}(g)}
+
\sup_{g\in\G_{K_g}}
\abs{\hat{\mathcal Q}_{g,K_g}^{\circ}(g)-\mathcal Q_{g,K_g}(g)} .
\end{aligned}
\]
The first term is $o_p(1)$ by Proposition
\ref{prop:app-micro-profile-negligible}. The second term is $o_p(1)$ by
Proposition \ref{prop:app-oracle-micro-criterion-UC}. Combining these
bounds gives
\[
    \sup_{(g,\psi)\in\ThetaSpace_K}
    \abs{
    \hat{\mathcal Q}_K(g,\psi)-\mathcal Q_K(g,\psi)
    }
    =
    o_p(1).
\]
\end{proof}

% =====================================================================
\section{Proof of the Main Consistency Theorem}
\label{app:main-consistency-proof}

\begin{theorem}[Consistency of the profiled SMD estimator, restated]
\label{thm:app-consistency-restated}
Suppose the maintained demand model holds. Suppose further that the identification conditions in Section \ref{sec:Identification}, the sampling condition in Assumption \ref{ass:sampling-main}, Assumptions \ref{ass:sieve-approx-main}--\ref{ass:profile-main} and Assumptions \ref{ass:population-id-main}--\ref{ass:criterion-UC-main} hold. Let $(\hat g,\hat\psi)$ satisfy \eqref{eq:estimator-main}. Then
\[
    \norm{\hat g-g_0}_{\infty}+\norm{\hat\psi-\psi_0}_{\infty}=o_p(1).
\]
Moreover,
\[
    \max_{t\le T}\norm{\hat a_t-a_{0t}}=o_p(1),
    \qquad
    \max_{t\le T}\norm{\hat h_t-h_{0t}}=o_p(1).
\]
Consequently,
\[
    \frac1T\sum_{t=1}^T\norm{\hat a_t-a_{0t}}^2=o_p(1),
    \qquad
    \frac1T\sum_{t=1}^T\norm{\hat h_t-h_{0t}}^2=o_p(1).
\]
\end{theorem}

\begin{proof}
The proof has two parts\footnote{The proof follows the standard extremum-estimation consistency logic; see \cite{newey_chapter_1994}.}. First, we show that the population criterion, evaluated at the estimated parameters, vanishes. Second, we use separation of the population criterion to conclude that any sequence with vanishing criterion must be close to the true parameter.

Towards this, define
\[
    \hat\theta:=(\hat g,\hat\psi),
    \qquad
    \theta_0:=(g_0,\psi_0),
    \qquad
    \theta_K^0:=(g_{0,K_g},\psi_{0,K_\psi}),
\]
where $\theta_K^0\in\ThetaSpace_K$ is the approximating sequence from Assumption \ref{ass:sieve-approx-main}. Thus
\begin{equation}
\label{eq:app-main-sieve-approx}
    \norm{\theta_K^0-\theta_0}_{\ThetaSpace}
    \to0,
    \qquad
    \mathcal Q_K(\theta_K^0)\to0.
\end{equation}
Let
\[
    \delta_T:=
    \sup_{\theta\in\ThetaSpace_K}
    \abs{\hat{\mathcal Q}_K(\theta)-\mathcal Q_K(\theta)}.
\]
By Proposition \ref{prop:app-full-UC}, $\delta_T=o_p(1)$.

Since $\hat\theta$ satisfies the approximate-minimization condition
\eqref{eq:estimator-main} and $\theta_K^0\in\ThetaSpace_K$,
\[
    \hat{\mathcal Q}_K(\hat\theta)
    \le
    \inf_{\theta\in\ThetaSpace_K}\hat{\mathcal Q}_K(\theta)+o_p(1)
    \le
    \hat{\mathcal Q}_K(\theta_K^0)+o_p(1).
\]
Using the definition of $\delta_T$ at both $\hat\theta$ and $\theta_K^0$, \[ \mathcal Q_K(\hat\theta) \le \hat{\mathcal Q}_K(\hat\theta)+\delta_T, \qquad \hat{\mathcal Q}_K(\theta_K^0) \le \mathcal Q_K(\theta_K^0)+\delta_T. \] Combining these inequalities with approximate minimization gives \( \mathcal Q_K(\hat\theta) \le \mathcal Q_K(\theta_K^0)+2\delta_T+o_p(1). \)
By criterion preservation in \eqref{eq:app-main-sieve-approx} and $\delta_T=o_p(1)$,
\begin{equation}
\label{eq:app-main-QKhat-vanishes}
    \mathcal Q_K(\hat\theta)=o_p(1).
\end{equation}
This is the step where the moving-criterion structure matters. We do not use ordinary continuity of a fixed population criterion; instead, we use the criterion-preserving approximating sequence required by Assumption \ref{ass:sieve-approx-main}.

We now use moving population separation to convert convergence of the criterion into convergence of the parameter. Fix $\epsilon>0$ and define
\[
    c_{\epsilon,K}
    :=
    \inf_{\theta\in\ThetaSpace_K:
    \norm{\theta-\theta_0}_{\ThetaSpace}\ge\epsilon}
    \mathcal Q_K(\theta).
\]
By Assumption \ref{ass:population-id-main},
\[
    L_\epsilon := \liminf_{K\to\infty}c_{\epsilon,K}>0.
\]
Choose $c_\epsilon \in (0, L_\epsilon)$. Then there exists $K_\epsilon < \infty$ such that 
$c_{\epsilon,K}\ge c_\epsilon$ for all $K \ge K_\epsilon$. Since $\hat\theta\in\ThetaSpace_K$,
\[
    \left\{ 
    \norm{\hat\theta-\theta_0}_{\ThetaSpace}\ge\epsilon
    \right\}
    \subseteq
    \left\{
    \mathcal Q_K(\hat\theta)\ge c_{\epsilon,K}
    \right\}
    \subseteq
    \left\{
    \mathcal Q_K(\hat\theta)\ge c_\epsilon
    \right\}
\]
for all sufficiently large $K$\footnote{The braces denote events, i.e. sets of sample outcomes. The first event is the event that the realized estimator is at least $\epsilon$ away from the truth. On that event, the realized value $\hat\theta$ belongs to the parameter set \[ \{\theta\in\ThetaSpace_K:\norm{\theta-\theta_0}_{\ThetaSpace}\ge\epsilon\}. \] Since $c_{\epsilon,K}$ is the infimum of $\mathcal Q_K(\theta)$ over that parameter set, it follows that $\mathcal Q_K(\hat\theta)\ge c_{\epsilon,K}$. The second inclusion follows from $c_{\epsilon,K}\ge c_\epsilon$ for all sufficiently large $K$.}. Therefore, by \eqref{eq:app-main-QKhat-vanishes},
\[
    \Prob\!\left(
    \norm{\hat\theta-\theta_0}_{\ThetaSpace}\ge\epsilon
    \right)
    \le
    \Prob\!\left(
    \mathcal Q_K(\hat\theta)\ge c_\epsilon
    \right)
    \to0.
\]
Since $\epsilon>0$ was arbitrary,
\[
    \norm{\hat\theta-\theta_0}_{\ThetaSpace}
    =
    \norm{\hat g-g_0}_{\infty}
    +
    \norm{\hat\psi-\psi_0}_{\infty}
    =
    o_p(1).
\]

It remains to prove consistency of the market-level objects. By definition, $\hat a_t=\hat a_t(\hat g)$. For each $t$, \[ \begin{aligned} \norm{\hat a_t-a_{0t}} &\le \norm{\hat a_t(\hat g)-a_t^*(\hat g)} + \norm{a_t^*(\hat g)-a_t^*(g_0)} + \norm{a_t^*(g_0)-a_{0t}} . \end{aligned} \] Taking the maximum over $t\le T$ gives \[ \begin{aligned} \max_{t\le T}\norm{\hat a_t-a_{0t}} &\le \sup_{g\in\G_{K_g}}\max_{t\le T} \norm{\hat a_t(g)-a_t^*(g)} + \max_{t\le T} \norm{a_t^*(\hat g)-a_t^*(g_0)} \\ &\quad+ \max_{t\le T} \norm{a_t^*(g_0)-a_{0t}} . \end{aligned} \] The first term is $o_p(1)$ by Proposition \ref{prop:profile-consistency-main}. The second term is bounded by $C_a\norm{\hat g-g_0}_\infty=o_p(1)$ by Proposition \ref{prop:app-profile-lipschitz}, which applies on $\G_{K_g}^0=\G_{K_g}\cup\{g_0\}$. The third term is zero by Proposition \ref{prop:app-true-profile}. Hence \[ \max_{t\le T}\norm{\hat a_t-a_{0t}}=o_p(1). \] Next, by the definitions $\hat h_t=\hat a_t-\hat\psi(P_t,X_t)$ and $h_{0t}=a_{0t}-\psi_0(P_t,X_t)$, \[ \hat h_t-h_{0t} = (\hat a_t-a_{0t}) - \{\hat\psi(P_t,X_t)-\psi_0(P_t,X_t)\}. \] Therefore,\[
\begin{aligned}
    \max_{t\le T}\norm{\hat h_t-h_{0t}}
    &\le
    \max_{t\le T}\norm{\hat a_t-a_{0t}}
    +
    \max_{t\le T}
    \norm{\hat\psi(P_t,X_t)-\psi_0(P_t,X_t)}  \\
    &\le
    \max_{t\le T}\norm{\hat a_t-a_{0t}}
    +
    \norm{\hat\psi-\psi_0}_\infty
    =
    o_p(1).
\end{aligned}
\] Thus \[ \max_{t\le T}\norm{\hat h_t-h_{0t}}=o_p(1). \] The average $L_2$ statements follow immediately from the uniform statements: \[ \frac1T\sum_{t=1}^T\norm{\hat a_t-a_{0t}}^2 \le \max_{t\le T}\norm{\hat a_t-a_{0t}}^2 = o_p(1), \] and \[ \frac1T\sum_{t=1}^T\norm{\hat h_t-h_{0t}}^2 \le \max_{t\le T}\norm{\hat h_t-h_{0t}}^2 = o_p(1). \] This proves all claims. 
\end{proof}

\clearpage

% =====================================================================

% =====================================================================

\clearpage

\section{Post-estimation Industrial Organization Objects}\label{app:io-objects}
% ---------------------------------------------------------------------

This appendix describes how the estimated demand system can be used to compute the standard IO objects that enter counterfactual analysis. Once the structural demand system is consistently recovered, the same estimated functions can be inserted into the formulas for elasticities, diversion ratios, markups, and pass-through.

Fix a market \(t\). For a candidate structural parameter \(\theta=(g,\psi)\) and a market shock \(h_t\), define the individual choice-probability vector
\begin{equation}
\label{eq:io-individual-share}
    s_t(z,p;\theta,h_t)
    :=
    \sigma\!\left(g(z,X_t)+\psi(p,X_t)+h_t\right)
    \in(0,1)^J.
\end{equation}
The corresponding aggregate inside-good share is
\begin{equation}
\label{eq:io-aggregate-share}
    S_t(p;\theta,h_t)
    :=
    \E\!\left[
    s_t(Z_{it},p;\theta,h_t)\mid \Mcal_t
    \right].
\end{equation}
In applications, the conditional expectation can be replaced by the empirical distribution of observed consumers in market \(t\):
\begin{equation}
\label{eq:io-share-estimator}
    \hat S_t(p)
    :=
    \frac1{n_t}\sum_{i=1}^{n_t}
    \sigma\!\left(
    \hat g(Z_{it},X_t)+\hat\psi(p,X_t)+\hat h_t
    \right).
\end{equation}
At the observed price vector \(P_t\), this is the fitted market share. At a counterfactual price vector \(p\), it is the predicted demand holding the recovered structural shock \(\hat h_t\) fixed.

The local price derivatives are obtained by differentiating through the logit map. Let
\[
    D_\psi(p,x)
    :=
    \nabla_p\psi(p,x)
    \in\R^{J\times J},
    \qquad
    [D_\psi(p,x)]_{\ell k}
    =
    \frac{\partial \psi_\ell(p,x)}{\partial p_k}.
\]
For \(u\in\R^J\), recall the multinomial-logit Jacobian
\[
    \Lambda(u):=\diag(\sigma(u))-\sigma(u)\sigma(u)'.
\]
Then the \(J\times J\) matrix of aggregate price derivatives is
\begin{equation}
\label{eq:io-aggregate-derivative}
    G_t(p)
    :=
    \frac{\partial S_t(p;\theta_0,h_{0t})}{\partial p'}
    =
    \E\!\left[
    \Lambda\!\left(g_0(Z_{it},X_t)+\psi_0(p,X_t)+h_{0t}\right)
    D_{\psi,0}(p,X_t)
    \mid \Mcal_t
    \right],
\end{equation}
where the \((j,k)\)-entry is \(\partial S_{jt}(p)/\partial p_k\). The plug-in estimate is
\begin{equation}
\label{eq:io-aggregate-derivative-hat}
    \hat G_t(p)
    :=
    \frac1{n_t}\sum_{i=1}^{n_t}
    \Lambda\!\left(
    \hat g(Z_{it},X_t)+\hat\psi(p,X_t)+\hat h_t
    \right)
    \nabla_p\hat\psi(p,X_t).
\end{equation}

The own- and cross-price elasticities are then
\begin{equation}
\label{eq:io-elasticity}
    \mathcal E_{jk,t}(p)
    :=
    \frac{p_k}{S_{jt}(p)}
    \frac{\partial S_{jt}(p)}{\partial p_k}
    =
    \frac{p_k}{S_{jt}(p)}[G_t(p)]_{jk},
\end{equation}
whenever \(S_{jt}(p)>0\). The diversion ratio from product \(j\) to product \(k\) is
\begin{equation}
\label{eq:io-diversion}
    \mathcal D_{j\to k,t}(p)
    :=
    -
    \frac{\partial S_{kt}(p)/\partial p_j}
    {\partial S_{jt}(p)/\partial p_j}
    =
    -
    \frac{[G_t(p)]_{kj}}{[G_t(p)]_{jj}},
    \qquad k\ne j,
\end{equation}
whenever the own-price derivative in the denominator is nonzero. These formulas are standard IO objects; the difference is that utility function entering these formulas is no longer parametrically specified. 

The same derivative matrix gives markups under static Bertrand-Nash pricing. Let \(\Omega_t\in\{0,1\}^{J\times J}\) be the ownership matrix, where \((\Omega_t)_{jk}=1\) if products \(j\) and \(k\) are owned by the same firm and zero otherwise. Let \(mc_t\) denote the vector of marginal costs. The first-order conditions for multi-product Bertrand pricing can be written as
\begin{equation}
\label{eq:io-bertrand-foc}
    S_t(P_t)+\{\Omega_t\circ G_t(P_t)\}'(P_t-mc_t)=0,
\end{equation}
where \(\circ\) denotes elementwise multiplication. If \(\{\Omega_t\circ G_t(P_t)\}'\) is nonsingular, the implied markup vector is
\begin{equation}
\label{eq:io-markup}
    P_t-mc_t
    =
    -
    [\left\{\Omega_t\circ G_t(P_t)\right\}^{\prime}]^{-1}S_t(P_t).
\end{equation}
Replacing \(S_t\) and \(G_t\) by \(\hat S_t\) and \(\hat G_t\) gives the plug-in markup estimate.

Finally, pass-through can be computed by differentiating the same first-order conditions. Define
\[
    F_t(p,c)
    :=
    S_t(p)+\{\Omega_t\circ G_t(p)\}'(p-c).
\]
At an interior equilibrium \(F_t(P_t,mc_t)=0\). If the Jacobian \(F_{p,t}:=\partial F_t(p,c)/\partial p'|_{(P_t,mc_t)}\) is nonsingular, the local cost pass-through matrix is
\begin{equation}
\label{eq:io-pass-through}
    \Pi_t
    :=
    \frac{\partial P_t}{\partial mc_t'}
    =
    -
    F_{p,t}^{-1}F_{c,t}
    =
    F_{p,t}^{-1}\{\Omega_t\circ G_t(P_t)\}'.
\end{equation}
This object requires second-order price derivatives of demand because \(F_{p,t}\) differentiates the demand derivative matrix \(G_t(p)\). Since \(\psi_0\) is estimated by a smooth sieve, these derivatives can be computed analytically from the estimated basis expansion whenever the basis is differentiable.

The dependence of pass-through on demand curvature can be seen directly from the Jacobian of the pricing first-order conditions. For each component,
\[
    [F_{p,t}]_{j\ell}
    =
    [G_t(P_t)]_{j\ell}
    +
    (\Omega_t)_{\ell j}[G_t(P_t)]_{\ell j}
    +
    \sum_{k=1}^J
    (\Omega_t)_{kj}(P_{kt}-mc_{kt})
    \frac{\partial^2 S_{kt}(P_t)}{\partial p_j\partial p_\ell}.
\]
Thus local pass-through depends on both the slope and the curvature of demand. For a small marginal-cost shock, \eqref{eq:io-pass-through} gives the first-order price response. For a large marginal-cost shock \(\Delta c_t\), the counterfactual price vector solves the nonlinear system
\[
    F_t(p,mc_t+\Delta c_t)=0.
\]
The resulting price change depends on the behavior of \(S_t(p)\), \(G_t(p)\), and the second derivatives of demand along the equilibrium path. This is why flexible estimation of the price-side function matters for counterfactuals. A misspecified linear or low-order demand curve may approximate the observed market shares but still impose the wrong curvature, and hence the wrong pass-through, when costs or taxes change. The estimator developed here is designed to recover these derivative objects consistently under additional smoothness conditions; inference for them is left to future work.

\clearpage

% =====================================================================
\section{Monte Carlo Details}\label{app:simulation-details}

This appendix describes the simulation design underlying Section \ref{sec:Simulations}. It records the exact data-generating process, the basis functions, the numerical implementation, the norm definitions, and the auxiliary Monte Carlo exercises used to audit the two pieces of the estimator separately.

\subsection{Consistency simulation: data-generating process}\label{app:simulation-dgp}

The baseline design has $J=2$ inside goods\footnote{The simulation design uses a scalar consumer covariate, \(d_z=1\), with
\(J=2\) inside goods. This differs from the high-level Berry--Haile
primitive condition \(d_z\ge J\). The Monte Carlo is therefore not intended
as a simulation of the full Berry--Haile nonparametric identification
environment. Rather, it studies the finite-sample behavior of the proposed
estimator under the maintained semi-nonparametric logit specialization.
Because the logit link is known, the product-specific log odds identify
\(g_{0j}(z,x)+a_{0jt}\) directly, and a scalar covariate can shift multiple
product-specific indices through the vector-valued function
\(g_0(z,x)=(g_{01}(z,x),g_{02}(z,x))'\).}. For each market $t=1,\ldots,T$, draw
\[
    X_t\sim N(0,1),
    \qquad
    W_t\sim N(0,I_2),
    \qquad
    h_{0t}\sim N(0,\Sigma_h),
    \qquad
    \nu_t\sim N(0,\Sigma_\nu),
\]
mutually independently, where
\[
    \Sigma_h=\begin{pmatrix}0.25&0.08\\0.08&0.22\end{pmatrix},
    \qquad
    \Sigma_\nu=\begin{pmatrix}0.20&0.03\\0.03&0.20\end{pmatrix}.
\]
Prices are generated by
\begin{equation}
\label{eq:app-mc-price-equation}
    P_t
    =
    2\mathbf 1_2
    +BW_t
    +\lambda\Lambda h_{0t}
    +\gamma_x X_t
    +\nu_t,
\end{equation}
where
\[
    B=\begin{pmatrix}1.00&0.25\\0.20&0.95\end{pmatrix},
    \qquad
    \Lambda=\begin{pmatrix}0.85&0.10\\0.10&0.80\end{pmatrix},
    \qquad
    \gamma_x=(0.35,-0.25)',
    \qquad
    \lambda=1.5.
\]
The term $\lambda\Lambda h_{0t}$ makes prices endogenous. At the same time, $W_t$ is independent of $h_{0t}$ conditional on $X_t$, so the simulated instruments satisfy the price-side validity condition used by the estimator.

Conditional on the market-level variables, consumers are independent. For each consumer $i$ in market $t$, draw $Z_{it}\sim N(0,1)$. Define
\[
    a_{0t}=\psi_0(P_t,X_t)+h_{0t},
    \qquad
    \eta_{it}=g_0(Z_{it},X_t)+a_{0t}.
\]
The inside-good choice probabilities are $\sigma(\eta_{it})$, with the outside option normalized to zero. Choices are drawn from the multinomial distribution over the outside option and the two inside goods. The simulation uses a scalar consumer covariate even though the high-level Berry--Haile primitive conditions allow for richer within-market variation. This is mainly to keep the setting simple, but it also qualitatively demonstrates that limited information does not automatically preclude my main consistency result. Additionally, the Monte Carlo studies the maintained semi-nonparametric logit specialization, not the full primitive nonparametric identification environment.

\subsection{Consistency simulation: basis functions and true coefficients}\label{app:simulation-bases}

The baseline simulation uses correctly specified finite sieves. Let
\[
    H_z(z)=(z,\ z^2-1,\ z^3-3z)'
    \qquad\text{and}\qquad
    L_x(x)=(1,\ x,\ x^2-1)'.
\]
The $g$ basis is the tensor product
\[
    b(z,x)
    =
    \bigl(H_z(z)',\ xH_z(z)',\ (x^2-1)H_z(z)'\bigr)'
    \in\mathbb R^9.
\]
Thus $g_{0j}(z,x)=b(z,x)'\beta_{0j}$ for $j=1,2$, with
\[
    \beta_0
    =
    \begin{pmatrix}
        0.95&-0.32&0.08&0.25&-0.08&0.025&-0.06&0.030&-0.012\\
       -0.80& 0.24&-0.06&-0.18& 0.07&-0.020& 0.05&-0.025& 0.010
    \end{pmatrix}.
\]
The price-side basis is
\[
    r(p,x)
    =
    \bigl(p_1-2,\ p_2-2,\ x,\ (p_1-2)^2,\ (p_2-2)^2,\ (p_1-2)(p_2-2)\bigr)'
    \in\mathbb R^6,
\]
and $\psi_{0j}(p,x)=r(p,x)'\vartheta_{0j}$, where
\[
    \vartheta_0
    =
    \begin{pmatrix}
        -1.10&0.35&0.25&-0.15&0.05&0.10\\
         0.30&-1.00&-0.20&0.04&-0.12&0.08
    \end{pmatrix}.
\]
The conditional-moment basis is
\[
    q(w,x)
    =
    \bigl(1,\ w_1,\ w_2,\ x,\ w_1^2-1,\ w_2^2-1,\ w_1w_2,\ w_1x,\ w_2x\bigr)'
    \in\mathbb R^9.
\]
The $g$ basis excludes a pure constant and uses mean-zero Hermite terms in $z$, which is consistent with the normalization that removes pure $X$-only components from $g$.

\subsection{Consistency simulation: implementation}\label{app:simulation-implementation}

For a candidate coefficient matrix $\beta$, the inner step profiles each market's composite intercept by solving the share-matching system
\[
    \frac{1}{n_t}\sum_{i=1}^{n_t}
    \{Q_{it}-\sigma(g_\beta(Z_{it},X_t)+a)\}=0.
\]
The implementation first attempts Newton's method using the analytic multinomial-logit Jacobian. If this root solve fails or returns a point outside the compact box $[-8,8]^2$, the code falls back to bounded minimization of the convex log-sum-exp objective function. This fallback is the numerical analogue of imposing a compact parameter space $\mathcal A$ for the inner profile.

The outer optimizer searches over the $J\times 9$ coefficients of $g_\beta$. For each trial value of $\beta$, it profiles all $a_t$, forms the micro residual moment, and computes the closed-form price-side coefficient estimate. The reported baseline uses the identity-weighted unconditional-moment version of the price-side SMD/GMM block,
\[
    \frac{1}{T}\sum_{t=1}^T
    q(W_t,X_t)\otimes\{\hat a_t(g)-\psi(P_t,X_t)\}.
\]
Equivalently, if $A$ is the $T\times J$ matrix of profiled intercepts, $R$ is the $T\times K_\psi$ matrix with rows $r(P_t,X_t)'$, and $Q$ is the $T\times K_q$ matrix with rows $q(W_t,X_t)'$, the price step solves
\[
    \hat\vartheta'
    =
    (R'QQ'R)^{-1}R'QQ'A.
\]
The code can also implement the projected SMD form (as seen in the main text), but the reported baseline uses this identity-weighted version for comparability across the simulation designs. No second-step estimate of $\Sigma(R_t)$ is used.

The full outer criterion is minimized by nonlinear least squares. The micro residual vector and the price-side residual vector are each divided by the square root of their respective moment counts. This block scaling affects finite-sample weighting but does not change the population zero of the moment system.

\subsection{Consistency simulation: evaluation metrics}\label{app:simulation-metrics}

The evaluation grids are fixed across designs and replications and are not used by the estimator. The grid for $g$ contains 41 equally spaced points in $z\in[-2,2]$ and 7 equally spaced points in $x\in[-1,1]$. Denote the resulting grid by $\mathcal E_g$. The grid for $\psi$ contains 600 fixed pseudo-random points: $p_1$ and $p_2$ are drawn uniformly from $[1,3]$, and $x$ is drawn uniformly from $[-1,1]$, using a fixed seed. Denote this grid by $\mathcal E_\psi$.

For the common structural functions, the code computes both root-mean-squared grid errors and grid sup-norm errors. For $f\in\{g,\psi\}$, let $u=(z,x)$ and $\mathcal E_f=\mathcal E_g$ when $f=g$, and let $u=(p,x)$ and $\mathcal E_f=\mathcal E_\psi$ when $f=\psi$. The generic metrics are
\[
    \widehat{\mathrm{RMSE}}_f
    =
    \left\{
    \frac{1}{J|\mathcal E_f|}
    \sum_{j=1}^J\sum_{u\in\mathcal E_f}
    \bigl(\hat f_j(u)-f_{0j}(u)\bigr)^2
    \right\}^{1/2},
    \qquad
    \widehat{\mathrm{SUP}}_f
    =
    \max_{j\le J}\max_{u\in\mathcal E_f}
    |\hat f_j(u)-f_{0j}(u)|.
\]
The main figure uses $\widehat{\mathrm{SUP}}_g$ and $\widehat{\mathrm{SUP}}_\psi$, because these are the finite-grid analogues of the uniform consistency statement for the common functions.

For the market-level generated objects $v\in\{a,h\}$, the reported metric is the average coordinate-level $\ell_2$ error
\[
    \widehat{\mathrm{RMSE}}_v
    =
    \left\{
    \frac{1}{TJ}\sum_{t=1}^T\sum_{j=1}^J
    (\hat v_{jt}-v_{0jt})^2
    \right\}^{1/2}.
\]
This equals $\{J^{-1}T^{-1}\sum_{t=1}^T\|\hat v_t-v_{0t}\|_2^2\}^{1/2}$ and therefore differs from the average market-level squared-error object in Theorem \ref{thm:consistency-main} only by the fixed normalization $J^{-1/2}$.

The normalized curves in Figure \ref{fig:mc-full-consistency} report, for each metric $M$,
\[
    \widetilde M(T,n_t)
    =
    \frac{\mathrm{median}_{r}\{M_r(T,n_t)\}}
    {\mathrm{median}_{r}\{M_r(50,100)\}},
\]
where $r$ indexes Monte Carlo replications. Thus a value below one means that the median error is smaller than the corresponding baseline error.

\subsection{Consistency simulation: Full-estimator design}\label{app:simulation-full-design}

The primary full-estimator design uses the correctly specified finite-sieve DGP and runs the feasible estimator on the nine-point grid
\[
    T\in\{50,100,200\},
    \qquad
    n_t\in\{100,250,500\},
\]
with 50 replications per design point. The figure in the main text summarizes the medians across replications. Since the top panels use sup-norm errors, the plotted values are intentionally stringent: a single poorly estimated grid point can determine the value of $\widehat{\mathrm{SUP}}_g$ or $\widehat{\mathrm{SUP}}_\psi$. This explains why the finite-sample curves need not be perfectly monotone cell-by-cell even though the broad pattern is consistent with the theorem.

\subsection{Consistency simulation: growing-sieve diagnostic}\label{app:simulation-growing-sieve}

The next exercise studies the role of sieve approximation on the price side. To isolate this mechanism, the simulation supplies the true composite intercepts $a_{0t}$ to the price-side estimator. Thus there is no inner-profile error and no outer optimization error. The DGP uses the same baseline price equation, but the true price-side function is no longer contained in the baseline six-dimensional basis. Specifically, for $p_c=(p_1-2,p_2-2)$,
\[
\begin{aligned}
    \psi_{01}(p,x)
    &= r(p,x)'\vartheta_{01}
    +0.30\sin(p_{c1}+0.5x)+0.20p_{c1}p_{c2}x+0.10p_{c1}^3,\\
    \psi_{02}(p,x)
    &= r(p,x)'\vartheta_{02}
    -0.27\sin(p_{c2}-0.4x)+0.18p_{c1}p_{c2}x-0.10p_{c2}^3.
\end{aligned}
\]
The fixed-sieve estimator continues to use the baseline $K_\psi=6$ price basis and the baseline $K_q=9$ instrument basis. The growing-sieve estimator uses polynomial bases in $(p_1-2,p_2-2,x)$ and in $(w_1,w_2,x)$. The dimension schedule is
\[
\begin{array}{c|ccccc}
T & 100 & 200 & 400 & 800 & 1600\\
\midrule
K_\psi \text{ (growing)} & 9 & 19 & 19 & 34 & 34\\
K_q \text{ (growing)} & 10 & 20 & 20 & 35 & 35
\end{array}
\]
with 50 replications at each design point.

\begin{figure}[!t]
    \centering
    \includegraphics[scale=0.78]{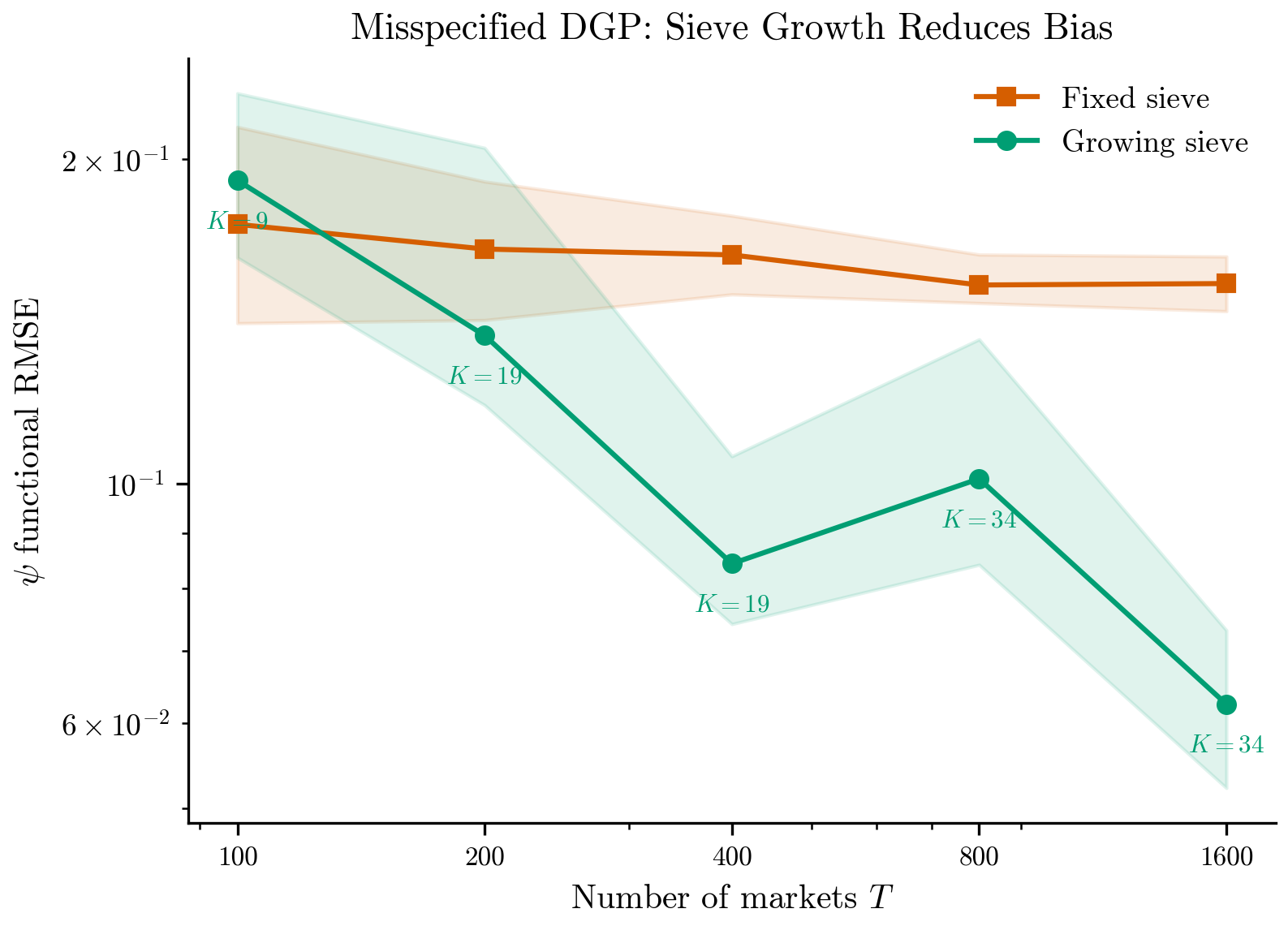}
    \caption{Fixed versus growing price-side sieve under approximation error.}
    \label{fig:mc-growing-sieve}

    \vspace{0.5em}
    \begin{minipage}{0.95\textwidth}
    \footnotesize
    \emph{Notes:}
    The figure reports the oracle-$a_t$ price-side exercise with a smooth truth outside the baseline price-side sieve. The vertical axis is the median $\psi$ grid RMSE across 50 replications. The fixed specification uses the baseline six-dimensional price basis throughout. The growing specification increases the polynomial degree as $T$ increases; the annotations report the corresponding value of $K_\psi$. Because the true composite intercepts are supplied to the price-side estimator, the figure isolates price-side approximation and sampling error rather than inner-profile error.
    \end{minipage}
\end{figure}

Figure \ref{fig:mc-growing-sieve} illustrates the approximation role of the sieve. With a fixed low-dimensional basis, sampling error decreases but approximation error remains, so the curve flattens. With a growing basis, the estimator can reduce approximation bias as $T$ increases. The jumps in performance occur when the polynomial basis becomes rich enough to approximate the additional nonlinear terms in the DGP more closely. This exercise is not the main consistency simulation; rather, it supports the sieve-approximation condition imposed in the consistency theorem.

\subsection{Consistency simulation: price-side robustness}\label{app:simulation-price-robustness}

The final exercise compares the oracle-$a_t$ price-side GMM/SMD estimator to an OLS benchmark. OLS regresses the true composite intercept $a_{0t}$ directly on the price-side basis $r(P_t,X_t)$ and therefore does not use the excluded instruments. Since both estimators are given the true $a_{0t}$, this comparison isolates the price-side endogeneity problem.

The simulation uses the correctly specified baseline price-side basis, $T=200$ markets, $n_t=200$ consumers per market, and 50 replications at each design point. Three one-dimensional sweeps are considered. First, the endogeneity scale $\lambda$ in \eqref{eq:app-mc-price-equation} varies over
\[
    \lambda\in\{0,0.5,1.0,1.5,2.0,2.5\}.
\]
Second, instrument strength varies by replacing $B$ with $\rho B$, where
\[
    \rho\in\{0.05,0.20,0.50,1.00,1.50,2.00\}.
\]
Third, the scale of the demand shock varies by replacing $\Sigma_h$ with $s^2\Sigma_h$, where
\[
    s\in\{0.25,0.50,1.00,1.50,2.00,2.50\}.
\]

\begin{figure}[!t]
    \centering
    \includegraphics[scale=0.65]{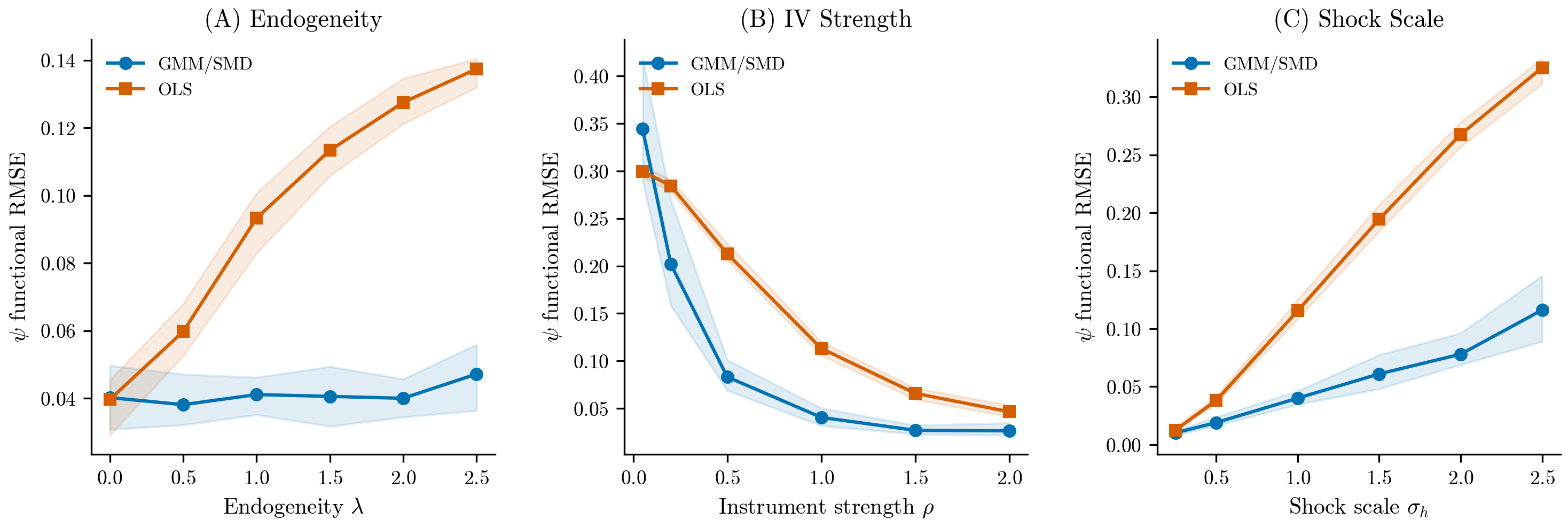}
    \caption{Oracle price-side robustness: GMM/SMD versus OLS.}
    \label{fig:mc-price-robustness}

    \vspace{0.5em}
    \begin{minipage}{0.95\textwidth}
    \footnotesize
    \emph{Notes:}
    The figure compares the oracle-$a_t$ price-side GMM/SMD estimator with an OLS benchmark. Each panel reports the median $\psi$ grid RMSE across 50 replications. The left panel varies the endogeneity scale $\lambda$ in the price equation. The middle panel varies instrument strength by scaling the price-instrument loading matrix by $\rho$. The right panel varies the scale of the unobserved demand shock by replacing $\Sigma_h$ with $s^2\Sigma_h$. Since the true composite intercepts are supplied to both estimators, the comparison isolates the price-side effect of using instruments rather than the generated-object error from profiling.
    \end{minipage}
\end{figure}

Figure \ref{fig:mc-price-robustness} shows the expected qualitative pattern. When prices are exogenous, OLS and GMM/SMD are close because the omitted demand shock is not correlated with prices. As endogeneity increases, the OLS error rises while the instrumented estimator remains comparatively stable. When instruments become stronger, the GMM/SMD estimator improves because the price-side moments become more informative. Finally, increasing the scale of the unobserved demand shock makes the OLS benchmark worse because the endogenous component of prices becomes more important. These patterns confirm that the price-side block is doing the intended IV work: it separates $\psi_0(P_t,X_t)$ from $h_{0t}$ using variation in prices induced by excluded instruments.

\subsection{Counterfactual simulation details}\label{app:counterfactual-details}

The counterfactual simulation is a separate design from the consistency simulation above. It has $J=3$ inside goods and one outside option. Individual utility is
\[
    U_{ijt}
    =
    g z_i+\psi_{0j}(P_t,X_t)+h_{jt}+\varepsilon_{ijt},
    \qquad
    U_{i0t}=\varepsilon_{i0t},
\]
where $\varepsilon_{ijt}$ is i.i.d. Type-I extreme value, $z_i\sim N(0,1)$, and $g=0.3$. The heterogeneity term is common across inside goods and acts as a random coefficient on the inside-good constant. The systematic price-side function $\psi_0$ is the object that differs across models. Depending on the counterfactual, $\psi_0$ is specified to be cubic, quadratic, or to contain cross-price nonlinearities.

Prices are endogenous. They depend on excluded cost shifters, the structural demand shock, and noise. The excluded cost shifters are valid instruments because they shift prices but are independent of the demand shock. Each simulated dataset contains $T=300$ markets and $n_t=500$ consumers per market. The Monte Carlo uses 40 independent replications. In every replication, the simulation records first choices and second choices. The second-choice data are used only by Micro-BLP. The sieve degree is chosen by cross-validation using out-of-sample fit of the composite intercept. It is not chosen using the true counterfactual errors.

% Four estimators are compared.
% \begin{enumerate}[label=(\roman*),leftmargin=2em]
%     \item \emph{Logit-IV} is an aggregate logit demand model with an instrumented linear price term.
%     \item \emph{BLP} is a random-coefficients logit benchmark. It includes random coefficients on the inside-good constant and on price, uses cost-shifter instruments, and is estimated with multi-start numerical optimization.
%     \item \emph{Micro-BLP} is the same random-coefficients logit benchmark augmented with a second-choice micro moment. The moment is the share of inside-good buyers whose second choice is the outside option.
%     \item \emph{Sieve-$\psi$} is the proposed flexible price-side estimator. In the counterfactual exercise, the heterogeneity loading is held fixed at the DGP value so that the comparison focuses on recovery of the flexible price-side function and the counterfactual objects depending on it.
% \end{enumerate}

The simulation has two regimes. In the \emph{oracle} regime, every estimator is supplied with the true counterfactual demand shock. For BLP, this means evaluating counterfactuals using the true unobserved quality. For the sieve and logit estimators, it means using the true $h_t$ when predicting counterfactual demand. This regime isolates functional-form error. In the \emph{full} regime, each estimator uses its own recovered demand shock. This is the realistic case and is the regime shown in Figure \ref{fig:mc-counterfactual-full}.

For a price perturbation $\delta$, let $s_{jt}^{0}(\delta)$ denote the true counterfactual share of product $j$ in market $t$, and let $\hat s_{jt}^{m}(\delta)$ denote the corresponding prediction from model $m$. Let $\delta=0$ denote the observed-price baseline. The share-change RMSE plotted in panel (a) of Figure \ref{fig:mc-counterfactual-full} is
\[
    \mathrm{RMSE}^{\Delta s}_m(\delta)
    =
    \left\{
    \frac{1}{TJ}\sum_{t=1}^T\sum_{j=1}^J
    \left[
    \{\hat s_{jt}^{m}(\delta)-\hat s_{jt}^{m}(0)\}
    -
    \{s_{jt}^{0}(\delta)-s_{jt}^{0}(0)\}
    \right]^2
    \right\}^{1/2}.
\]
Thus the object is not the level of market shares but the predicted change in market shares under the counterfactual price movement.

For diversion, let $D_{jk,t}^{0}(\delta)$ denote the true diversion ratio from product $j$ to product $k$ at the evaluation point indexed by $\delta$, and let $\hat D_{jk,t}^{m}(\delta)$ be the corresponding model prediction. In differential form,
\[
    D_{jk,t}(\delta)
    =
    -\frac{\partial s_{kt}(\delta)/\partial p_{jt}}
    {\partial s_{jt}(\delta)/\partial p_{jt}},
    \qquad k\ne j.
\]
The reported diversion error is the average absolute error over the evaluated markets and product pairs:
\[
    \mathrm{AE}^{D}_m(\delta)
    =
    \frac{1}{|\mathcal I_D|}
    \sum_{(t,j,k)\in\mathcal I_D}
    \left|\hat D_{jk,t}^{m}(\delta)-D_{jk,t}^{0}(\delta)\right|,
\]
where $\mathcal I_D$ is the set of market-product-product triples used in the evaluation.

For welfare, let
\[
    CS_t(\delta)
    =
    \E_Z\left[
    \log\left\{1+\sum_{j=1}^J
    \exp\bigl(gZ+\psi_{0j}(P_t^\delta,X_t)+h_{jt}\bigr)
    \right\}
    \right]
\]
denote expected inclusive value in market $t$ under the counterfactual price vector $P_t^\delta$. The welfare change is $\Delta CS_t(\delta)=CS_t(\delta)-CS_t(0)$. Each model produces its own estimate $\widehat{\Delta CS}^{m}_t(\delta)$. The reported welfare error is
\[
    \mathrm{AE}^{CS}_m(\delta)
    =
    \frac{1}{T}\sum_{t=1}^T
    \left|
    \widehat{\Delta CS}^{m}_t(\delta)
    -
    \Delta CS_t^{0}(\delta)
    \right|.
\]
Welfare is measured in utility units, i.e. as a change in expected inclusive value.

\begin{table}[!t]
\centering
\caption{Counterfactuals: oracle regime.}
\label{tab:counterfactual-oracle-scorecard}
\resizebox{\textwidth}{!}{%
\begin{tabular}{lccccc}
\toprule
Front & Logit-IV & BLP & Micro-BLP & Sieve-$\psi$ & Winner\\
\midrule
Share-change RMSE & 0.0183 & 0.0212 & 0.0124 & 0.0074 & Sieve-$\psi$\\
Elasticity & 1.5503 & 1.6948 & 1.7289 & 0.1440 & Sieve-$\psi$\\
Diversion & 0.3320 & 0.1834 & 0.3353 & 0.1163 & Sieve-$\psi$\\
Pass-through & 0.1029 & 0.0400 & 0.1343 & 0.0107 & Sieve-$\psi$\\
Merger & 0.3288 & 0.1153 & 0.3300 & 0.1786 & BLP\\
Welfare & 0.0328 & 0.0348 & 0.0425 & 0.0208 & Sieve-$\psi$\\
\bottomrule
\end{tabular}%
}

\vspace{0.5em}
\begin{minipage}{0.95\textwidth}
\footnotesize
\emph{Notes:}
Each entry is the mean over evaluation points of the median counterfactual error across 40 replications. Lower values are better. In the oracle regime, each estimator is supplied with the true counterfactual demand shock, so the comparison isolates functional-form error.
\end{minipage}
\end{table}

\begin{table}[!t]
\centering
\caption{Counterfactuals: full regime.}
\label{tab:counterfactual-full-scorecard}
\resizebox{\textwidth}{!}{%
\begin{tabular}{lccccc}
\toprule
Front & Logit-IV & BLP & Micro-BLP & Sieve-$\psi$ & Winner\\
\midrule
Share-change RMSE & 0.0183 & 0.0211 & 0.0120 & 0.0088 & Sieve-$\psi$\\
Elasticity & 1.7187 & 2.2089 & 1.8093 & 5.1783 & Logit-IV\\
Diversion & 0.3299 & 0.2037 & 0.3308 & 0.1239 & Sieve-$\psi$\\
Pass-through & 0.1523 & 0.1175 & 0.1695 & 0.1913 & BLP\\
Merger & 0.3305 & 0.1409 & 0.3299 & 0.1753 & BLP\\
Welfare & 0.0319 & 0.0361 & 0.0328 & 0.0226 & Sieve-$\psi$\\
\bottomrule
\end{tabular}%
}

\vspace{0.5em}
\begin{minipage}{0.95\textwidth}
\footnotesize
\emph{Notes:}
Each entry is the mean over evaluation points of the median counterfactual error across 40 replications. Lower values are better. In the full regime, each estimator recovers its own demand shock from the simulated data. The main text focuses on share changes, diversion, and welfare because these are central counterfactual objects and because the full-regime plots for these fronts are visually stable. The table also reports the derivative and equilibrium objects, where the flexible sieve can be more sensitive to extrapolation and generated-profile error.
\end{minipage}
\end{table}

The tables above clarify the scope of the counterfactual results. In the oracle regime, where functional form is isolated, the sieve estimator wins five of the six fronts. This is consistent with the design: the price side of the DGP is nonlinear, and the sieve is the only model allowed to approximate that nonlinear price surface flexibly. In the full regime, the sieve wins three of the six fronts: share changes, diversion, and welfare. BLP performs best on merger and pass-through, while Logit-IV performs best on elasticity. These exceptions are informative in that they show that the flexible estimator is not uniformly dominant in the presence of generated-object error and extrapolation. In particular, derivative and equilibrium objects can amplify noise in the recovered price surface. The economically important message is therefore not that flexibility always dominates a linear model, but that flexibility can materially improve central counterfactual objects when the true price side is nonlinear.

\clearpage

\section{Coefficient Counts and Dimensionality}
\label{app:coefficient-counts}

This appendix clarifies the dimensions of the estimator. There are three objects to count: structural coefficients, profiled market intercepts, and auxiliary coefficients used in the price-side conditional-moment projection.

Let \(d_z\), \(d_x\), and \(d_w\) denote the dimensions of \(Z_{it}\), \(X_t\), and \(W_t\), respectively, and let \(J\) be the number of inside goods. Suppose total-degree polynomial sieves are used. The number of polynomial terms in \(d\) variables of degree at most \(L\), including the constant, is
\[
    K(d,L)=\binom{d+L}{L},
\]
where \(\binom{d+L}{L}\) is the binomial coefficient. For example, with two variables and degree \(2\), the terms are \(\{1,u_1,u_2,u_1^2,u_1u_2,u_2^2\}\), so \(K(2,2)=6\).

The heterogeneity function \(g_0\) is a \(J\)-vector of functions of \((Z_{it},X_t)\). Since \(g_0\) is normalized to remove terms that depend only on \(X_t\), the number of \(g\)-basis terms per product is
\[
    K_g^{\mathrm{norm}}
    =
    \binom{d_z+d_x+L_g}{L_g}
    -
    \binom{d_x+L_g}{L_g}.
\]
The price-side function \(\psi_0\) is a \(J\)-vector of functions of \((P_t,X_t)\), where \(P_t\in\mathbb R^J\). Thus the number of \(\psi\)-basis terms per product is
\[
    K_\psi
    =
    \binom{J+d_x+L_\psi}{L_\psi}.
\]
Therefore the number of common structural demand coefficients is
\begin{equation}
\label{eq:app-coeff-structural-short}
    d_{\mathrm{struct}}
    =
    J\left(K_g^{\mathrm{norm}}+K_\psi\right)
    =
    J\left[
    \binom{d_z+d_x+L_g}{L_g}
    -
    \binom{d_x+L_g}{L_g}
    +
    \binom{J+d_x+L_\psi}{L_\psi}
    \right].
\end{equation}
These are the coefficients that parameterize the common demand system.

The market intercepts add \(J\) nuisance components per market, or \(JT\) total. The estimator profiles these objects market by market rather than optimizing over them jointly with the structural coefficients. Profiling reduces the dimension of the outer optimization, but the intercepts are still estimated objects; this is why the asymptotic framework requires \(n_{\min}\to\infty\).

The price-side SMD basis \(q^{K_q}(W_t,X_t)\) does not parameterize utility; rather, it approximates the conditional moment
\[
    \E\!
    \left[
    a_{0t}-\psi_0(P_t,X_t)
    \mid W_t,X_t
    \right]
    =0.
\]
With a total-degree polynomial basis of order \(L_q\),
\[
    K_q
    =
    \binom{d_w+d_x+L_q}{L_q},
\]
so the auxiliary price-side projection dimension is \(J K_q\). Thus \(K_q\) matters for the number of moments and for the statistical complexity of the SMD block, but it is not a structural demand coefficient.

If one counts every coefficient or nuisance object computed inside the estimator, the total is
\begin{equation}
\label{eq:app-coeff-total-short}
    d_{\mathrm{computed}}
    =
    d_{\mathrm{struct}}+JT+JK_q.
\end{equation}

As a simple example, suppose \(J=3\), \(d_z=2\), \(d_x=1\), \(d_w=2\), and \(L_g=L_\psi=L_q=2\). Then \(K_g^{\mathrm{norm}}=\binom{5}{2}-\binom{3}{2}=7\), \(K_\psi=\binom{6}{2}=15\), and \(K_q=\binom{5}{2}=10\). The model has \(3(7+15)=66\) common structural coefficients. If \(T=100\), the profile step computes \(300\) market-intercept components, and the price-side SMD block uses \(30\) auxiliary projection coefficients.

These formulas show where the curse of dimensionality enters. The heterogeneity side grows with \(d_z+d_x\) and \(L_g\), while the price side grows with \(J+d_x\) and \(L_\psi\). The price-side block is especially demanding because it is learned from market-level variation, so its effective sample size is \(T\), not \(N=\sum_t n_t\). The estimator mitigates the dimensionality problem by using individual-level observations to learn \(g_0\) and by profiling the \(JT\) market intercepts rather than including them in the outer optimization. It does not eliminate the curse of dimensionality in \(\psi_0\). Future work could impose additional economic structure, use regularization or model selection, exploit symmetries across products, or target lower-dimensional counterfactual objects directly.

\clearpage

\bibliographystyle{chicago}
\nocite{*}
\bibliography{Metrics2ndyearpaper}

\end{document}